\documentclass[11pt,a4paper,bibliography=totoc,bibliography=openstyle, reqno]{amsart}
\usepackage{amsfonts}
\usepackage{amsthm}
\usepackage{amsmath}
\usepackage{amscd}
\usepackage{amssymb}
\usepackage{subfigure}
\usepackage[T1]{fontenc}            
\usepackage[utf8]{inputenc}       
\usepackage[english]{babel}                  
\usepackage{csquotes}               
\usepackage{t1enc}
\usepackage[mathscr]{eucal}
\usepackage{indentfirst}
\usepackage{graphicx}
\usepackage{graphics}
\usepackage{pict2e}
\usepackage{color}
\usepackage{epic}
\usepackage{bbm}
\usepackage{enumitem}
\usepackage{tikz}
\usepackage{tkz-graph}
\usetikzlibrary{arrows}
\numberwithin{equation}{section}
\usepackage[margin=2.6cm]{geometry}
\usepackage{epstopdf} 
\usepackage{hyperref}
\usepackage{textcomp} 
\usepackage{autobreak}
\usepackage{subfigure}
\usepackage{float}
\usepackage{dsfont}
\usepackage{mathtools}
\usepackage{pdflscape}
\usepackage{appendix}
\usepackage{booktabs}

\usepackage[backend=bibtex,               
          natbib=true,                 
          block = space,
           bibstyle=authoryear,        
           citestyle=authoryear,        
           babel=other,                 
           maxcitenames=2,              
           maxbibnames=10,              
           url=true,                   
           doi=true,                   
           giveninits=true]{biblatex} 
           
\renewbibmacro{in:}{} 
\DeclareNameAlias{sortname}{last-first} 
\DefineBibliographyStrings{german}{andothers={\mbox{et~al.}}} 
\newcommand{\N}{\mathbb{N}}

\newtheorem{Theorem}{Theorem}[section]
\newtheorem{Lemma}[Theorem]{Lemma}
\newtheorem{Cor}[Theorem]{Corollary}

\newtheorem{asum}{Assumption}
\newtheorem{?}[Theorem]{Problem}

\newcommand{\fett}{\boldsymbol}

\newcommand{\tz}{\theta_0}
\newcommand{\gz}{G_0}
\newcommand{\thn}{\widehat{\theta}_n}
\newcommand{\ghn}{\widehat{G}_n}
\newcommand{\thns}{\widehat{\theta}_n^*}
\newcommand{\ghns}{\widehat{G}_n^*}

\newcommand{\ahns}{\widehat\alpha_n^*}
\newcommand{\ax}{\alpha_{\xi_n}}
\newcommand{\gx}{G_{\xi_n}}
\newcommand{\ahn}{\widehat\alpha_n}
\newcommand{\az}{\alpha_0}

\newcommand{\hpsxn}{\widehat{P}_{S,x_n}}

\DeclareRobustCommand{\subtitle}[1]{\\#1}

\allowdisplaybreaks

\numberwithin{equation}{section}

\begin{document}

\title[Semi-parametric INAR bootstrap inference]{Joint semi-parametric INAR bootstrap inference for model coefficients and innovation distribution 
}


\author[Maxime Faymonville \& Carsten Jentsch]{Maxime Faymonville
   \\  \texttt{Department of Statistics,  TU Dortmund University, D-44221 Dortmund, Germany; faymonville@statistik.tu-dortmund.de}\vspace*{0.5cm} \\  \vspace*{0.5cm}
\and  \\ 
Carsten Jentsch  \\ \texttt{ Department of Statistics,  TU Dortmund University, D-44221 Dortmund, Germany; jentsch@statistik.tu-dortmund.de
}}

\begin{abstract}
For modeling the serial dependence in time series of counts, various approaches have been proposed in the literature. In particular, models based on a recursive, autoregressive-type structure such as the well-known integer-valued autoregressive (INAR) models are very popular in practice. The distribution of such INAR models is fully determined by a vector of autoregressive binomial thinning coefficients and the discrete innovation distribution. While fully parametric estimation techniques for these models are mostly covered in the literature, a semi-parametric approach allows for consistent and efficient joint estimation of the model coefficients and the innovation distribution without imposing any parametric assumptions. Although the limiting distribution of this estimator is known, which, in principle, enables asymptotic inference and INAR model diagnostics on the innovations, it is cumbersome to apply in practice.

In this paper, we consider a corresponding semi-parametric INAR bootstrap procedure and show its joint consistency for the estimation of the INAR coefficients and for the estimation of the innovation distribution. We discuss different application scenarios that include goodness-of-fit testing, predictive inference and joint dispersion index analysis for count time series. In simulations, we illustrate the finite sample performance of the semi-parametric INAR bootstrap using several innovation distributions and provide real-data applications.

\end{abstract}



\maketitle

\renewcommand{\subtitle}[1]{}

\newpage

\section{Introduction}\label{sec:intro}
Count time series consist of sequences of observations over time taking values in the non-negative integers $\mathbb{N}_0 := \mathbb{N} \cup \{0\} = \{0,1,2,\ldots\}$.
They arise naturally when counting things or events over time and have therefore many relevant applications in various fields as, e.g., the number of infectious diseases, extreme weather events or phishing attacks. Unlike continuous time series, count data is inherently discrete-valued, which often leads to modeling challenges caused, e.g., by the presence of overdispersion (variance exceeding the mean) or zero inflation (excessive zeros in the data). One of the probably most used count time series models is the Integer-valued AutoRegressive (INAR) model of order $p\in\mathbb{N}$ introduced by \citet{duli}. An INAR($p$) process $(X_t,t\in\mathbb{Z})$ is defined by the recursion 
\begin{align} \label{inarp}
X_t = \alpha_1 \circ X_{t-1}+ \alpha_2 \circ X_{t-2} + \ldots + \alpha_p \circ X_{t-p} + \varepsilon_t, \quad t \in \mathbb{Z},
\end{align} 
where $(\varepsilon_t,t\in\mathbb{Z})$ denotes an i.i.d. innovation process with distribution $G$ having range $\mathbb{N}_0$. We write $\varepsilon_t\sim G$ and identify the distribution $G$ by its probability mass function (pmf), that is, $G=\{G(k),k\in\mathbb{N}_0\}$, where $G(k)=P(\varepsilon_t=k)$. Further, let $\fett\alpha = (\alpha_1, \ldots, \alpha_p)' \in [0,1)^p$ with $\sum_{i=1}^p \alpha_i<1$ denote the vector of model coefficients and 
define
\[\alpha_i \circ X_{t-i} = \sum\limits_{j=1}^{X_{t-i}} Z_j^{(t,i)},\] 
where \enquote{$\circ$} is the binomial thinning operator first introduced by \citet{steutel}. Here, $(Z_j^{(t,i)}, \, j \in \mathbb{N}, \, t \in \mathbb{N}_0 ), \, i \in \{1, \ldots, p\}$, are mutually independent Bernoulli-distributed random variables with $Z_j^{(t,i)} \sim \text{Bin}(1, \alpha_i)$ such that $P(Z_j^{(t,i)}=1)=\alpha_i=1-P(Z_j^{(t,i)}=0)$ independent of $(\varepsilon_t,t\in\mathbb{N}_0)$. Note that, according to this construction, we have $\alpha_i \circ X_{t-i}|X_{t-i}\sim \text{Bin}(X_{t-i}, \alpha_i)$. All the thinning operations \enquote{$\circ$} are independent of each other and of $(\varepsilon_t, t\in\mathbb{Z})$. Furthermore, the thinning operation at time $t$ and $\varepsilon_t$ are both independent of $X_s, \, s < t$. The special case $p=1$ results in the INAR(1) model introduced by  \citet{mck} and \citet{alosh}.

The existing literature mainly deals with \emph{fully} parametric estimation of INAR models \citep[see, e.g.,][]{duli, franke, brhell, freeland, jung, silsil, bu},
i.e., $G$ is assumed to belong to a certain family of parametric distributions $\{G_\gamma \mid \, \gamma \in \Gamma \subset \mathbb{R}^q\}$ for some finite (and typically small) $q\in\N$. A summary of all the standard parametric estimation methods for INAR models as, e.g., moment estimation and (conditional) maximum likelihood estimation can be found in Section 2.2 of \citet{bookweiss}. However, the use of such parametric assumptions considerably restricts the flexibility of the INAR model \eqref{inarp} - we refer to \citet{gof_sp} for a discussion. A way more flexible estimator is the one presented by \citet{drost}. While keeping the \emph{parametric} binomial thinning operation, it is able to estimate the innovation distribution in a completely \emph{non-parametric} way. In the following Section \ref{sec:prelim}, we introduce this semi-parametric estimator that allows to estimate \emph{jointly} the INAR coefficients $\fett\alpha$ and the innovation distribution $G$ and recap some of its theoretical (limiting) properties.
An appropriate \emph{semi-parametric} bootstrap procedure leading to corresponding bootstrap estimators is proposed in Section \ref{sec:main}, where we also establish asymptotic theory and prove bootstrap consistency. Different application scenarios that
include goodness-of-fit testing, predictive inference and joint dispersion index analysis for time series
of counts are discussed in Section \ref{sec:meth_appl}. We provide simulation results illustrating the finite sample performance of the semi-parametric INAR bootstrap procedure in Section \ref{sec:sim} and discuss a real data application in Section \ref{sec:appl}. 
Section \ref{sec:concl} concludes and the proofs of this paper are deferred to an appendix.

\section{Preliminaries} \label{sec:prelim}
 
By construction, the INAR($p$) process defined by \eqref{inarp} is a $p$th order Markov chain. Under the model parameters $\fett\alpha$ and $G$ and for $x_{t-i}\in\mathbb{N}_0$, $i=0,1,\ldots,p$, its transition probabilities are given by \begin{align} \label{tpinar}
P^{\fett\alpha,G}_{(x_{t-1}, \ldots, x_{t-p}), x_t}
&= \mathds{P}_{\fett\alpha, G} \left( \sum\limits_{i=1}^p \alpha_i \circ X_{t-i}+\varepsilon_t=x_t \mid X_{t-1}=x_{t-1}, \ldots, X_{t-p}=x_{t-p} \right) \\
&= \Big(\text{Bin}(x_{t-1}, \alpha_1) \ast \cdots \ast \text{Bin}(x_{t-p}, \alpha_p) \ast G \Big)\{x_t\}, 
\notag
\end{align}
where $\mathds{P}$ is the underlying probability measure and \enquote{$\ast$} denotes the convolution of distributions. In the special case of an INAR(1) model, the transition probabilities can be written as 
\begin{align} \label{tpinar1} \mathds{P}_{\alpha, G} (X_t=x_t \mid X_{t-1}=x_{t-1})= \sum\limits_{j=0}^{\text{min}(x_t, x_{t-1})} \binom{x_{t-1}}{j} \alpha^j (1-\alpha)^{x_{t-1}-j} G(x_t-j), 
\end{align} 
where $\alpha$ is the coefficient of the INAR(1) model \citep{mck, alosh}. 
When estimating the INAR model as proposed in \citet{drost}, that is, without using any parametric assumption on the family of the innovation distributions, one treats the error distribution as an (infinite-dimensional) parameter and jointly estimates the parameter sequence consisting of the $p$-dimensional vector of INAR model coefficients and the infinite-dimensional pmf of the innovations distribution. In the setup of generalized linear models, \citet{huang} considers a similar approach for the estimation of the error distribution.
For the derivation of asymptotic theory, \citet{drost} impose the following (moment) assumptions on the innovation distribution and the model parameters.

\begin{asum}[Innovation distribution and model parameters; \citet{drost}, Assumption 1] \label{ass1}
Let $\widetilde{\mathcal{G}}$ denote the set of all probability measures on $\mathbb{N}_0$. We assume that $G=\mathcal{L}(\varepsilon_t) \in \mathcal{G}$, where
\begin{align*}
\mathcal{G}=\left\{ G \in \widetilde{\mathcal{G}}: 0 < G(0) < 1, \text{E}_G (\varepsilon_t^{p+4}) < \infty \right\}    
\end{align*}
is the set of  all probability measures on $\mathbb{N}_0$ with $0<G(0)<1$ and finite $(p+4)$th moments. Furthermore, we assume that $\fett\alpha = (\alpha_1,\ldots,\alpha_p) \in \Theta= \{\fett\alpha \in (0,1)^p: \sum_{i=1}^p \alpha_i <1 \}$.
\end{asum} 
For some of the results that \citet{drost} derive in their paper, weaker conditions than listed in Assumption \ref{ass1} are sufficient. Nevertheless, the conditions included in Assumption \ref{ass1} are not restrictive: The condition $0 < G(0) <1$ makes sure that the innovations can be equal to zero, but that they are not always equal to zero, which is suitable for basically every practical application. The existence of the $(p+4)$th moment of the innovation distribution $G$ is required to ensure the weak convergence of certain empirical processes; see \citet{drost} for details. The assumption $\fett\alpha \in (0,1)^p$ with $\sum_{i=1}^p \alpha_i <1$ entails the standard assumption to ensure the stationarity of the INAR($p$) process 
and avoids boundary issues by ruling out $\alpha_i=0$ for all $i=1,\ldots,p$.

Formally, \citet{drost} base their work on the experiments
\[\mathcal{E}^{(n)} = \left(\mathbb{N}_0^{n+1+p}, \mathcal{P}(\mathbb{N}_0^{n+1+p}), \mathds{P}^{(n)}_{\nu_{\fett\alpha,G},\fett\alpha,G} | \fett\alpha \in \Theta, G \in \mathcal{G}\right), \quad n \in \mathbb{N}_0,\] 
where $\mathcal{P}(\mathbb{N}_0^{n+1+p})$ denotes the power set of $\mathbb{N}_0^{n+1+p}$, $\mathds{P}^{(n)}_{\nu_{\fett\alpha,G},\fett\alpha,G}$ denotes the law of $(X_{-p}, \ldots, X_n)$ under $\mathds{P}_{\nu_{\fett\alpha,G},\fett\alpha,G}$, on the measurable space $(\mathbb{N}_0^{n+1+p}, \mathcal{P}(\mathbb{N}_0^{n+1+p}))$ and with stationary initial distribution $\nu_{\fett\alpha,G}$ of $(X_{-p}, \ldots, X_0)$.\footnote{Although a sample $X_1,\ldots,X_n$ plus $p$ pre-sample values $X_{-(p-1)},\ldots,X_0$ would be sufficient, throughout this paper, we stick to the notation in \citet{drost}, who suppose availability of $(X_{-p},\ldots,X_{-1}),X_0,\ldots,X_n$, but nevertheless use a pre-factor of $1/n$ when taking sample averages.} 
As one might expect, this semi-parametric model is more complicated to deal with than a parametric approach, but it is way more flexible with respect to the innovation distribution \citep[see][for the parametric counterparts]{drost2}.

The estimator is derived using a non-parametric (conditional) maximum likelihood approach. 
For any fixed $n\in\mathbb{N}_0$ and observations $(X_{-p},\ldots,X_n)$ at hand, a non-parametric maximum likelihood estimator (NPMLE) $\widehat \theta_n:=(\widehat{\fett\alpha}_n, \widehat{G}_n)=(\widehat{\alpha}_{n,1}, \ldots, \widehat{\alpha}_{n,p}, \widehat{G}_n(0), \widehat{G}_n(1), \ldots)$ of $(\fett\alpha,G)$ is defined to maximize the conditional likelihood, i.e., 
\begin{align}\label{NPMLE_maximization_problem} 
(\widehat{\fett\alpha}_n, \widehat{G}_n) \in \underset{(\fett\alpha,G) \in [0,1]^p \times \widetilde{\mathcal{G}}}{\text{argmax}} \left( \prod\limits_{t=0}^n P^{\fett\alpha,G}_{(X_{t-1}, \ldots, X_{t-p}),X_t}  \right).
\end{align}  
To guarantee its existence, note that they allow for values of $(\widehat{\fett\alpha}_n, \widehat{G}_n)$ outside of $\Theta \times \mathcal{G}$. They further state that all the mass of $\widehat{G}_n$ is assigned to a subset $\{u_-, \dots, u_+\}\subset \mathbb{N}_0$, where
\[
u_- = \max\left\{0,\min_{t = 0, \dots, n} \left\{ X_t - \sum_{i = 1}^{p} X_{t - i} \right\}\right\},
\quad
u_+ = \max_{t = 0, \dots, n} \{X_t\}.
\]
Then, $(\widehat{\theta}_n, \widehat{G}_n)$ maximizes the likelihood if and only if the following conditions hold:
\begin{itemize}
    \item[i)] $\widehat{G}_n(k) = 0$ for $k < u_-$ and $k > u_+$, and
    \item[ii)] $(\widehat{\alpha}_{n,1}, \dots, \widehat{\alpha}_{n,p}, \widehat{G}_n(u_-), \dots, \widehat{G}_n(u_+))$ is a solution to the (constrained) polynomial optimization problem
\end{itemize}
\begin{align}\label{eq:restictied_optimization_problem}
\max_{x_1, \dots, x_p \atop z_{u_-}, \dots, z_{u_+}} 
\left\{ 
\prod_{t = 0}^n 
 \sum_{e=0 \vee X_t - \sum_{i = 1}^{p} X_{t - i}}^{X_t} z_e 
\sum_{ \substack{ 0 \leq k_l \leq  X_{t-l}, \, l=1, \ldots, p \\ k_1 + \ldots + k_p = X_t - e } } 
\prod_{l = 1}^p \binom{X_{t - l}}{k_l} x_l^{k_l}(1 - x_l)^{X_{t - l} - k_l}
\right\}
\end{align}
subject to
\[
0 \leq x_k \leq 1 \quad \text{for } k = 1, \ldots, p, \quad
z_j \geq 0 \quad \text{for } j = u_-, \ldots, u_+, \quad
z_{u_-} + \ldots + z_{u_+} = 1.
\]
They stress that they do not impose the uniqueness of such a maximum location.
A modification of the likelihood in \eqref{NPMLE_maximization_problem} is proposed by \textcite{sp_penal}, who add a penalization term to account for the ``smoothness'' of most (discrete) innovation distributions, i.e., that $G(k+1)-G(k)$, $k\in\mathbb{N}_0$ is usually ``small'', which leads to an improved estimation performance for small sample sizes.

In their paper, \textcite{drost} prove consistency, asymptotic normality and efficiency of their NPMLE under suitable regularity conditions. However, the practical use of their asymptotic theory is rather restricted as the derived limiting distribution follows a (transformed) Gaussian process, which is cumbersome to work with. Hence, with the goal of constructing a suitable and asymptotically valid semi-parametric INAR bootstrap procedure,  in the following, we recap the consistency and asymptotic normality results established in \citet{drost}. 

\begin{Theorem}[\citet{drost}, Theorem 1] \label{cons}
Let Assumption \ref{ass1} hold.
For all $(\fett\alpha_0, G_0) \in \Theta \times \mathcal{G}$ and all initial probability measures  $\nu_{\fett\alpha_0,G_0}$ on $\mathbb{N}_0^p$, any NPMLE $(\widehat{\fett\alpha}_n, \widehat{G}_n)$ 
defined in \eqref{NPMLE_maximization_problem} is consistent (as $n\rightarrow \infty$) in the following sense: 
\[\widehat{\fett\alpha}_n \overset{p}{\rightarrow} \fett\alpha_0 \quad \text{and} \quad \sum\limits_{k=0}^\infty \left|\widehat{G}_n(k) - G_0(k)\right| \overset{p}{\rightarrow} 0, \quad \text{under} \, \, \, \mathds{P}_{\nu_{\fett\alpha_0,G_0}, \fett\alpha_0, G_0}.\]
\end{Theorem}

To derive the limiting distribution of the NPMLE, \citet{drost} show that their NPMLE is actually an infinite dimensional $Z$-estimator (see, e.g., \citet{kosorok} for a definition), i.e., that it solves an infinite number of moment conditions. 
We get these \emph{infinitely many} moment conditions, because we allow for unbounded innovation distributions having support $\mathbb{N}_0$. This makes the setting flexible, but also complicates the derivation of the limiting distribution. \citet{drost} handle this by constructing (artificial) probability distributions on $\mathbb{N}_0$ bounded in direction $h: \mathbb{N}_0 \rightarrow \mathbb{R}$. They only use moment conditions arising from $h \in \mathcal{H}_1$ with $\mathcal{H}_1$ being the unit ball of $\ell^\infty(\mathbb{N}_0)$, where the latter denotes the Banach space of bounded sequences equipped with the supremum norm, i.e., they only consider functions $h:\mathbb{N}_0 \rightarrow \mathbb{R}$ with $\sup_{e \in \mathbb{N}_0}|h(e)| \leq 1$. Finally, the estimating equations of the NPMLE are derived as $\Psi_n=(\Psi_{n1}, \Psi_{n2}):(0,1)^p \times \widetilde{\mathcal{G}} \rightarrow \mathbb{R}^p \times \ell^\infty(\mathcal{H}_1)$ defined by \begin{align} 
\Psi_{n1}(\fett\alpha, G) &= \frac{1}{n} \sum\limits_{t=0}^n \dot l_{\fett\alpha}(X_{t-p}, \ldots, X_t; \fett\alpha, G), \label{eq:psi_n1} \\
\Psi_{n2}(\fett\alpha, G) h &= \frac{1}{n} \sum\limits_{t=0}^n \left(A_{\fett\alpha,G} h(X_{t-p}, \ldots, X_t)-\int h dG  \right), \quad h \in \mathcal{H}_1, \label{eq:psi_n2} 
\end{align} 
where, for $x_{t-p}, \ldots, x_t \in \mathbb{N}_0$, 
\begin{align} \label{dot_l}
\dot l_{\fett\alpha} (x_{t-p}, \ldots, x_t; \fett\alpha, G) = \frac{\partial}{\partial \fett\alpha} \log \left(P^{\fett\alpha, G}_{(x_{t-1}, \ldots, x_{t-p}), x_t} \right)
\end{align}
with the convention that $\dot l_{\fett\alpha} (x_{t-p}, \ldots, x_t; \fett\alpha, G)=0$ if $P^{\fett\alpha, G}_{(x_{t-1}, \ldots, x_{t-p}), x_t}=0$ and 
\begin{align} \label{Ah}
A_{\fett\alpha,G} h(x_{t-p}, \ldots, x_t)=\text{E}_{\fett\alpha,G}\big(h(\varepsilon_t)| X_t=x_t, \ldots, X_{t-p}=x_{t-p}\big). 
\end{align}
Additionally, for $(\fett\alpha_0, G_0)$, they introduce the population counterparts of these estimating equations, $\Psi^{\fett\alpha_0,G_0}=(\Psi^{\fett\alpha_0,G_0}_1,\Psi^{\fett\alpha_0,G_0}_2):(0,1)^p\times \mathcal{G} \rightarrow \mathbb{R}^p \times \ell^\infty(\mathcal{H}_1)$, where
\begin{align}
\Psi_1^{\fett\alpha_0,G_0}(\fett\alpha,G)&=\text{E}_{\nu_{\fett\alpha_0,G_0}, \fett\alpha_0, G_0} \left( \dot l_{\fett\alpha} (X_{t-p}, \ldots, X_t;\fett\alpha, G) \right), \label{eq:psi_1} \\
\Psi_2^{\fett\alpha_0,G_0}(\fett\alpha,G) h &=\text{E}_{\nu_{\fett\alpha_0,G_0}, \fett\alpha_0, G_0} \left( A_{\fett\alpha,G} \, h(X_{-p}, \ldots, X_0)-\int h \; dG \right), \, h \in \mathcal{H}_1. \label{eq:psi_2}
\end{align}
By exploiting that their NPMLE $(\widehat{\fett\alpha}_n,\widehat G_n)$ provides a solution to the estimating equations, i.e., that $\Psi_n(\widehat{\fett\alpha}_n,\widehat G_n)=0$ holds approximately with $\Psi_n$ defined in equations \eqref{eq:psi_n1} and \eqref{eq:psi_n2}, they derive a weak convergence result
\begin{align}\label{eq:S_n_weak_convergence}
\mathcal{S}_n^{\fett\alpha_0, G_0} = \sqrt{n} \left(\Psi_n(\fett\alpha_0, G_0) -\Psi^{\fett\alpha_0, G_0}(\fett\alpha_0, G_0) \right) \leadsto \mathcal{S}^{\fett\alpha_0, G_0} 
\end{align}
in $\mathbb{R}^p \times \ell^\infty(\mathcal{H}_1)$, under $\mathds{P}_{\nu_0, \fett\alpha_0, G_0}$, where \enquote{$\leadsto$} indicates weak convergence and $\mathcal{S}^{\fett\alpha_0, G_0}$ is a tight, Borel measurable, Gaussian process \citep[see][Eq. (15)]{drost}. Altogether, \citet{drost} get the following result for the limiting distribution of an NPMLE $\widehat \theta_n=(\widehat{\fett\alpha}_n,\widehat G_n)$.

\begin{Theorem}[\citet{drost}, Theorem 2] \label{clt}
Suppose Assumption \ref{ass1} holds. For $\theta_0= (\fett\alpha_0, G_0) \in \Theta \times \mathcal{G}$, any NPMLE $\widehat{\theta}_n=(\widehat{\fett\alpha}_n, \widehat{G}_n)$ satisfies 
\begin{align}\label{eq:CLT}    
\sqrt{n} \left( \widehat{\theta}_n-\theta_0 \right) \leadsto -\dot\Psi^{-1}_{\theta_0} \left(\mathcal{S}^{\theta_0}\right),
\end{align}
in $\mathbb{R}^p \times \ell^1 (\mathbb{N}_0)$, under $\mathds{P}_{\nu_{\theta_0}, \theta_0}$, where $\dot\Psi^{-1}_{\theta_0}$ is the continuous inverse of the Fréchet derivative of $\Psi$ at $\theta_0$ and  $\mathcal{S}^{\theta_0}$ is the tight, Borel measurable, Gaussian process determined by \eqref{eq:S_n_weak_convergence}.
\end{Theorem}
According to Theorem \ref{clt}, the limiting distribution of the NPMLE is a transformed Gaussian process. However, due to the infinite dimensional parameter space, this transformation is rather complicated and relies on the (inverted) Fréchet derivative $\dot\Psi^{-1}_{\theta_0}$ of $\Psi^{\theta_0}$ given in \eqref{eq:psi_1} and \eqref{eq:psi_2} and given in Section 3.2 of \citet{drost}, which is cumbersome in practical applications \citep[see also][for a discussion]{huang}.
This generally motivates to use a suitable bootstrap procedure for statistical inference instead of using asymptotic approximations that requires the cumbersome explicit estimation of many nuisance parameters. We propose and investigate a semi-parametric INAR bootstrap procedure in the following section.


\section{Semi-parametric INAR Bootstrap} \label{sec:main}

In this section, to enable suitable semi-parametric bootstrap inference in INAR models, we propose to use the semi-parametric version of the INAR bootstrap as proposed in \citet{jewe}, which they proved to be consistent for statistics belonging to the large class of functions of generalized means of $(X_t,t\in\mathbb{Z})$ under mild assumptions. However, this class of statistics does not contain statistics depending on the estimated innovation distribution. For the purpose of extending the existing theory and to cover also such statistics, in Section \ref{sec:boot_algo}, we introduce the semi-parametric INAR bootstrap scheme and provide the required notation. Further, in Section \ref{sec:boot_theory},
we prove bootstrap consistency by showing that the bootstrap version of the NPMLE leads in probability (conditional on the data) to the same limiting distribution as obtained for the NPMLE given in Theorem \ref{clt}.

\subsection{The semi-parametric INAR Bootstrap Scheme}\label{sec:boot_algo}

For bootstrap inference of the NPMLE, the semi-parametric INAR bootstrap scheme, as proposed in \citet{jewe} for functions of generalized means and implemented in the R package \textit{spINAR} by \citep{faymonville2024spinar}, is defined as follows:

\begin{itemize}
    \item[Step 1.)] Given a sample (including pre-sample values) of count data $(X_{-p},\ldots,X_{-1}),X_0,\ldots,X_n$, fit semi-parametrically an INAR($p$) process \eqref{inarp}
    using the NPMLE \eqref{NPMLE_maximization_problem} as proposed by \citet{drost} to get estimators $\widehat{\fett\alpha}_n=(\widehat{\alpha}_{n,1}, \ldots, \widehat{\alpha}_{n,p})'$ and $\widehat{G}_{n} = (\widehat{G}_{n}(k), \, k \in \mathbb{N}_0)$ for the INAR coefficients and for the pmf of the innovation distribution, respectively.
    \item[Step 2.)] Generate bootstrap observations $(X_{-p}^*,\ldots,X_{-1}^*),X_0^*, \ldots, X_n^*$ according to \[X_t^* = \widehat{\alpha}_{n,1} \circ^* X_{t-1}^*+ \ldots + \widehat{\alpha}_{n,p} \circ^* X_{t-p}^* + \varepsilon_t^*, \] where \enquote{$\circ^*$} denotes (mutually independent) bootstrap binomial thinning operations and $(\varepsilon_t^*,t\in\mathbb{Z})$ denotes an i.i.d.~bootstrap innovation process with $\varepsilon_t^*\sim \widehat G_n$ (conditionally on the data).
    \item[Step 3.)] Compute the bootstrap NPMLE $\widehat \theta_n^*=(\widehat{\fett\alpha}_n^*,\widehat G_n^*)$, where $\widehat{\fett\alpha}_n^*=(\widehat{\alpha}_{n,1}^*, \ldots, \widehat{\alpha}_{n,p}^*)$ and $\widehat{G}_n^*=(\widehat{G}_n^*(k), \, k\in \mathbb{N}_0)$, according to equation \eqref{NPMLE_maximization_problem}, but applied to the bootstrap sample $(X_{-p}^*,\ldots,X_{-1}^*),X_0^*,\ldots,X_n^*$. That is, for any fixed $n\in\mathbb{N}_0$, a bootstrap NPMLE (bootNPMLE) $\widehat \theta_n^*:=(\widehat{\fett\alpha}_n^*, \widehat{G}_n^*)=(\widehat{\alpha}_{n,1}^*, \ldots, \widehat{\alpha}_{n,p}^*, \widehat{G}_n^*(0), \widehat{G}_n^*(1), \ldots)$ is defined to maximize the bootstrap version of the conditional likelihood in \eqref{NPMLE_maximization_problem}, i.e., 
\begin{align}\label{NPMLE_maximization_problem_boot} 
(\widehat{\fett\alpha}_n^*, \widehat{G}_n^*) \in \underset{(\fett\alpha,G) \in [0,1]^p \times \widetilde{\mathcal{G}}}{\text{argmax}} \left( \prod\limits_{t=0}^n P^{\fett\alpha,G}_{(X_{t-1}^*, \ldots, X_{t-p}^*),X_t^*}  \right),
\end{align}  
under analogous conditions as 
described in the restricted optimization problem \eqref{eq:restictied_optimization_problem}.
\end{itemize}



Analogous to \citet{drost}, for the bootstrap NPMLE (bootNPMLE) $\widehat \theta_n^*=(\widehat{\fett\alpha}_n^*,\widehat G_n^*)$, which fulfills a \enquote{bootstrap version} of Assumption \ref{ass1} by Lemma \ref{lemma_rng}, we get a $Z$-estimator representation with corresponding estimating equations $\Psi_n^*=(\Psi_{n_1}^*, \Psi_{n_2}^*):(0,1)^p \times \widetilde{\mathcal{G}} \rightarrow \mathbb{R}^p \times \ell^\infty(\mathcal{H}_1)$ defined by  \begin{align} 
\Psi_{n,1}^*(\fett\alpha, G) &= \frac{1}{n} \sum\limits_{t=0}^n \dot l_{\fett\alpha}(X_{t-p}^*, \ldots, X_t^*; \alpha, G),  \label{Psi_n1*} \\
\Psi_{n,2}^*(\fett\alpha, G) h &= \frac{1}{n} \sum\limits_{t=0}^n \left( A_{\fett\alpha, G} h(X_{t-p}^*, \ldots, X_t^*)-\int h \; dG  \right), \quad h \in \mathcal{H}_1.  \label{Psi_n2*}
\end{align}
Throughout this paper, as usual in bootstrap literature, we make use of the following notation. With $\text{E}^*(\cdot), \,\text{Var}^*(\cdot)$ and $\text{Cov}^*(\cdot)$, we denote the bootstrap expected value and (co)variance, respectively, given the original observations. That is, for $\mathbb{X}=(X_{-p},\ldots,X_{-1},X_0, \ldots, X_n)$, we have $\text{E}^*(\cdot)=\text{E}(\cdot|\mathbb{X})$, $\text{Var}^*(\cdot)=\text{Var}(\cdot|\mathbb{X})$ and $\text{Cov}^*(\cdot,\cdot)=\text{Cov}(\cdot,\cdot|\mathbb{X}).$ To be in-line with the notation of \citet{drost}, we clarify $\text{E}(\cdot) \, \widehat{=} \, \text{E}_{\fett\alpha_0,G_0}(\cdot)$, $\text{E}^*(\cdot) \, \widehat{=} \allowbreak \, \text{E}_{\widehat{\fett\alpha}_n, \widehat{G}_n}(\cdot |\mathbb{X})$ as well as $P\widehat{=}\mathds{P}_{\fett\alpha, G}$ and $P^*(\cdot)\widehat{=}\mathds{P}_{\fett\alpha, G}(\cdot|\mathbb{X})$.

\subsection{Bootstrap theory}\label{sec:boot_theory}

For the subsequent asymptotic theory for the semi-parametric INAR bootstrap, we 
impose the following regularity condition. 

\begin{asum}[Innovation distribution regularity] \label{asum_mix_gu} \phantom \\
    Suppose that the pmf $G$ is bounded away from zero by an exponentially decaying function. That is, let $G\in \mathcal{G}_u$, where
\begin{align}    
\label{eq:exp_decay}
\mathcal{G}_u =&\Big\{G \in \mathcal{G}: \exists c_1 \in (0,1], \, c_2 \in (0,\infty) \text{ such that } G(k) \geq c_1 e^{-c_2k} \, \forall k \in \mathbb{N}_0  \Big\}.
\end{align}
\end{asum}
The previous Assumption \ref{asum_mix_gu} is rather mild and not restrictive in practice. 
$\mathcal{G}_u$ contains all innovation distributions for which there exists a strictly positive and exponentially decaying function that always \enquote{lies below} the pmf of the innovation distribution. This property is fulfilled by several innovation distributions like, e.g., the Poisson, the negative binomial and the geometric distributions. In particular, it guarantees $G(k)>0$ for all $k\in\mathbb{N}_0$ such that boundary issues are ruled out. Note that \citet{drost} do not make such an assumption. However, in view of the limiting distribution that they obtained in Theorem \ref{cons}, asymptotic normality of $\widehat G(k)$ cannot hold due to $z_k\geq 0$ by construction in \eqref{eq:restictied_optimization_problem}. Hence, it appears that \citet{drost} do require such an assumption as well. 


Along the lines of the asymptotic theory established for the NPMLE in \citet{drost}, we need to first derive ``estimation consistency'' of the bootNPMLE. Such a result constitutes the bootstrap version of Theorem \ref{cons} (i.e., of Theorem 1 in \citet{drost}) and is required to argue that also the bootNPMLE is away from the boundary of the parameter space $\Theta \times \mathcal{G}$. Under the assumptions introduced above, we get the following result.

\begin{Theorem}[Estimation consistency of the bootNPMLE] \label{cons_bs_est}
Suppose Assumptions \ref{ass1} and \ref{asum_mix_gu} hold and we observe data $X_{-p},\ldots,X_n$ from an INAR($p$) process $(X_t,t \in \mathbb{Z})$ with INAR coefficients $\fett\alpha_0$ and innovation distribution $G_0$ such that $\theta_0=(\fett\alpha_0, G_0) \in \Theta \times \mathcal{G}$. Let $\text{E}(X_1^k) < \infty$ for some $k>2(p+4)$ and $\thn=(\widehat{\fett\alpha}_n, \widehat{G}_n)$ be an NPMLE of $(\fett\alpha_0, G_0)$. Suppose the bootstrap proposal from Section \ref{sec:boot_algo} is used to get a bootNPMLE $\thn^*=(\widehat{\fett\alpha}_n^*, \widehat{G}_n^*)
$. Then, for 
all initial probability measures $\nu_{\fett\alpha_0, G_0}$ on $\mathbb{N}^p_0$, 
we have 
\[ \widehat{\fett\alpha}_n^* - \widehat{\fett\alpha}_n  \overset{p^*}{\rightarrow} 0 \quad \text{and} \quad \sum \limits_{k=0}^\infty | \widehat{G}_n^*(k) - \widehat{G}_n(k) |  \overset{p^*}{\rightarrow} 0 \] 
under $P^*(\cdot)= \mathds{P}_{\nu_{\widehat{\fett\alpha}_n,\widehat{G}_n}, \widehat{\fett\alpha}_n, \widehat{G}_n}(\cdot|\mathbb{X})$ 
in $\mathds{P}_{\nu_{\fett\alpha_0,G_0}, \fett\alpha_0, G_0}$-probability. 
\end{Theorem}


Now, we get to the main result of this paper, the limiting distribution of the bootNPMLE. As the resulting limiting distribution does coincide with the limiting distribution derived in Theorem \ref{clt} \citep[i.e., in Theorem 2 in][]{drost}, the following Theorem \ref{bs_clt} proves bootstrap consistency of the semi-parametric INAR bootstrap for the NPMLE.

\begin{Theorem}[Bootstrap consistency of bootNPMLE] \label{bs_clt}
Suppose Assumptions \ref{ass1} and \ref{asum_mix_gu} hold and we observe data $X_{-p},\ldots,X_n$ from an INAR($p$) process $(X_t, t \in \mathbb{Z})$ with INAR coefficients $\fett\alpha_0$ and innovation distribution $G_0$ such that $\theta_0=(\fett\alpha_0, G_0) \in \Theta \times \mathcal{G}$. Let $\text{E}(X_1^k) < \infty$ for some $k>2(p+4)$, $\text{E}((X_t^3(1+\rho)^{X_t}))^{1+\delta}<\infty$ for some $\rho, \delta>0$ and $\widehat{\theta}_n=(\widehat{\fett\alpha}_n, \widehat{G}_n)$ be an NPMLE of $\theta_0=(\fett\alpha_0, G_0)$. Suppose the bootstrap proposal from Section \ref{sec:boot_algo} is used to get a bootNPMLE $\thn^*=(\widehat{\fett\alpha}_n^*, \widehat{G}_n^*)
$.
Then, for all initial probability measures $\nu_{\fett\alpha_0, G_0}$ on $\mathbb{N}^p_0$, 
we have 
\begin{align}\label{eq:bootstrap_CLT}    
\sqrt{n} \left( \widehat{\theta}_n^*-\widehat{\theta}_n \right) = \sqrt{n} \left( (\widehat{\fett\alpha}_n^*, \widehat{G}_n^*)-(\widehat{\fett\alpha}_n, \widehat{G}_n) \right) \leadsto^* -\dot\Psi^{-1}_{\theta_0} \left(\mathcal{S}^{\theta_0}\right)  
\end{align} 
in $\mathbb{R}^p \times \ell^1 (\mathbb{N}_0)$, under $P^*(\cdot)= \mathds{P}_{\nu_{\widehat{\fett\alpha}_n,\widehat{G}_n}, \widehat{\fett\alpha}_n, \widehat{G}_n}(\cdot|\mathbb{X})$ 
in $\mathds{P}_{\nu_{\theta_0}, \theta_0}$-probability.
\end{Theorem}

\begin{proof}
Along the lines of the proof of Theorem 2 in \citet{drost}, we make use of Theorem 3.3.1 in \citet{vdvw} in the following to prove asymptotic normality in \eqref{eq:bootstrap_CLT}. By construction of the (boot)NPMLE as $Z$-estimator, we approximately have 
\begin{align} \label{eq:approxzero}
\Psi_n(\thn)=0 \quad \text{and} \quad \Psi_n^*(\thns)=0.
\end{align}
Using (\ref{eq:approxzero}), we get \begin{align}\label{frst}
\sqrt{n}(\Psi_n(\thns)-\Psi_n(\thn)) = \sqrt{n}(\Psi_n(\thns)-\Psi_n^*(\thns)). 
\end{align}
Knowing from Lemma \ref{frechet} that $\Psi_n$ is Fréchet-differentiable, according to \eqref{taylor_lemma}, we can replace the left-hand side of \eqref{frst} by  
\begin{align} \label{taylor} \sqrt{n}\left(\dot\Psi_{n}^{\xi_n}(\thns-\thn)\right),\end{align} 
where $||\xi_n - \thn|| \leq ||\thns - \thn||$. Using the uniform convergence of $\dot\Psi_{n}^{\xi_n}$ to $\dot\Psi^{\theta_0}$ established in Lemma \ref{unif_zws}, \eqref{taylor} is asymptotically equivalent to
\begin{align}
 \sqrt{n}\left(\dot\Psi^{\theta_0}(\thns-\thn)\right).
\end{align}
Further, using the result of Lemma \ref{negl}, we can rewrite also the right-hand side of \eqref{frst}. Altogether, we have 
\begin{align} \label{eq:last}
\sqrt{n}\left(\dot\Psi^{\theta_0}(\thns-\thn)\right) = \sqrt{n}\left(\Psi_n(\thn)-\Psi_n^*(\thn)\right) +o_{p^*}(1).
\end{align} 
Finally, using the asymptotic normality result from 
Lemma \ref{asympt_norm} and the fact that $\dot\Psi_{\tz}$ is continuously invertible according to Lemma 2 in \citet{drost}, making use of the continuous mapping theorem, the assertion follows.

\end{proof}

Together with the bootstrap consistency for functions of generalized means  derived in \citet{jewe}, the asymptotic results established in Theorems \ref{clt} and \ref{bs_clt} generalize the validity of the semi-parametric INAR bootstrap for a broader class of statistics. In particular, as the limiting distributions in Theorems \ref{clt} and \ref{bs_clt} coincide, we obtain the following corollary.

\begin{Cor}[First-order semi-parametric INAR bootstrap consistency]
Let $d$ be an appropriate metric on (the distributions of) random elements in $\mathbb{R}^p \times \ell^1 (\mathbb{N}_0)$. Under the conditions of Theorem \ref{clt} and \ref{bs_clt}, we have 
\begin{align*}
    d\bigg(\mathcal{L}(\sqrt{n}(\thn-\tz)), \mathcal{L}^*(\sqrt{n}(\thns-\thn))\bigg) \underset{n \rightarrow \infty}{\rightarrow} 0 
\end{align*}
in $P_{\nu_{\theta_0}, \theta_0}$-probability. 
\end{Cor}
Furthermore, it is possible to make use of suitable delta methods to extend the bootstrap consistency result also to smooth functionals applied to the (boot)NPMLE. More precisely, for the bootNPMLE result \eqref{eq:bootstrap_CLT}, we can employ the delta method introduced in Theorem 3.1 of \citet{beutner_zaehle} which - in contrast to the conventional functional delta method as given by Theorems 3.9.11 and 3.9.13 of \citet{vdvw} - does not make use of concepts of integrals and outer probabilities. Similarly, for the result of Theorem \ref{clt}, Theorem 3.9.4. of \citet{vdvw} can be applied which is valid for mappings between linear metric spaces.  The paper of \citet{beutner} discusses several variants of functional delta methods. Consequently, we get the following result. 



\begin{Cor}[First-order bootstrap consistency for smooth functionals] \label{cor:delta_method_boot_consistency}
    Let $\Xi: [0,1]^p \times \tilde{\mathcal{G}} \rightarrow \mathbb{R}^q$ be a sufficiently smooth functional such that the conditions in Theorem 3.1 of \citet{beutner_zaehle} are fulfilled and let $d$ be an appropriate metric on (the distributions of) random elements in $\mathbb{R}^q$. 
    Under the conditions of Theorem \ref{clt} and \ref{bs_clt}, we have 
\begin{align*}
    d\left(\mathcal{L}\left(\sqrt{n} \left( \Xi ( \widehat{\theta}_n ) - \Xi \left(\theta_0 \right) \right)\right), \mathcal{L}^*\left(\sqrt{n} \left( \Xi ( \widehat{\theta}_n^* ) - \Xi (\widehat{\theta}_n ) \right)\right)\right) \underset{n \rightarrow \infty}{\rightarrow} 0 
\end{align*}
in $P_{\nu_{\theta_0}, \theta_0}$-probability. 
\end{Cor}

The bootstrap consistency result from Corollary \ref{cor:delta_method_boot_consistency} finds application in diverse setups. For instance, a joint consistency result that combines the bootstrap consistency for functions of generalized means derived in Theorem 3.2 and Corollary 3.7 in \citet{jewe} and the bootstrap consistency for smooth functionals of $\widehat \theta_n$ from Corollary \ref{cor:delta_method_boot_consistency} above, can be applied for semi-parametric INAR goodness-of-fit testing discussed in Section \ref{subsec:gof}. Further examples include, e.g., predictive inference in count time series setups covered in Section \ref{subsec:pp} or the analysis of the (joint) dispersion index covered in Section \ref{subsec:disp}.


\section{Methodological Applications} \label{sec:meth_appl}

In the following subsections, wo discuss three methodological applications of Theorem \ref{bs_clt} and Corollary \ref{cor:delta_method_boot_consistency} for different statistical tasks. They cover goodness-of-fit testing, predictive inference, and joint dispersion index analysis for (semi-parametrically estimated) INAR processes.

\subsection{Goodness-of-fit test on the whole INAR model class} \label{subsec:gof}

Together with the results from \citet{jewe}, the result of our main Theorem \ref{bs_clt} provides the theoretical foundation of the bootstrap-based semi-parametric goodness-of-fit test introduced in \textcite{gof_sp}. They consider the null hypothesis ``$H_0$: $(X_t,t\in\mathbb{Z})$ is an INAR($p$) process'' and their test is characterized by the lack of any parametric assumption on the innovation distribution which credits the flexibility of the INAR model class. The test statistic consists of an $L_2$-type distance of probability generating functions and can be represented as a degenerate $V$-statistic. This leads to a cumbersome $\chi^2$-type limiting distribution requiring an appropriate bootstrap technique to make the testing procedure practicable. The semi-parametric INAR bootstrap described in Section \ref{sec:boot_algo} is used for simulations leading to good finite sample performance under the null and under the alternative. In this paper, we provide the theoretical justification the use of this bootstrap procedure for goodness-of-fit testing of semi-parametric INAR null hypotheses.

\subsection{Predictive inference for count time series} \label{subsec:pp}

In the setup of count time series and more general for discrete-valued time series, it is not straightforward to perform predictive inference. \citet{pred_paper} discuss this topic and propose to solve this issue by transforming the prediction problem into a parameter estimation problem. That is, given $(X_{-p},\ldots,X_{-1}),X_0,\ldots,X_n$ and for the case of an underlying Markov process of order one, they construct \emph{confidence} intervals for the probability that $X_{n+1}$ falls into a certain set $S$ conditional on $X_n=x_n$ for some $x_n\in\mathbb{N}_0$, i.e., for $P_{S,x_n} = P(X_{n+1} \in S |X_n=x_n)$. 
They introduce asymptotic and bootstrap approaches and illustrate their proposed procedure on INAR and INARCH models. For this purpose, they consider parametric as well as non-parametric approaches and, moreover, they also discuss the practically important case of model misspecification, which we do not touch here.

In the setup of \citet{pred_paper}, also a semi-parametric INAR(1) variant \emph{without} specifying the parametric family of the innovation distribution
becomes suitable. 
Hence, assuming that the data is generated from an INAR(1) model, the procedure makes use of the semi-parametric INAR bootstrap from Section \ref{sec:boot_algo} as described in the following: 
\begin{itemize}
    \item[Step 1.)] Given $X_1,\ldots, X_n$, calculate the NPMLE $\thn=(\widehat{\fett\alpha}_n,\widehat G_n)$ as well as the semi-parametric estimator $ \widehat{P}_{S,x_n}^{sp} :=P_{S,x_n}(\thn)$ for the predictive probability $P_{S,x_n}$.
    \item[Step 2.)] Use the semi-parametric INAR bootstrap from Section \ref{sec:boot_algo} to get bootstrap observations $X_1^*,\ldots,X_n^*$ and to calculate the bootNPMLE $\thn^*=(\widehat{\fett\alpha}_n^*,\widehat G_n^*)$ as well as the semi-parametric bootstrap estimator $ \widehat{P}_{S,x_n}^{sp,*} :=P_{S,x_n}(\thn^*)$. 
    \item[Step 3.)] Repeat Step 2.) $B$-times, where $B$ is large, to get $L_n^{*,b}=\widehat{P}_{S,x_n}^{sp,*,b} - \widehat{P}_{S,x_n}^{sp}$, $b=1, \ldots, B$, and construct the $(1-\delta)$-confidence interval for $P_{S,x_n}$ as $[\hpsxn - q^*_{1-\delta/2} , \hpsxn - q^*_{\delta/2}]$,
where $q_\alpha^*$ denotes the $\alpha$-quantile of the empirical distribution of $L_n^{*,b}$, $b=1,\ldots,B$.
\end{itemize}
In Section \ref{sec:pred_inf_sims}, we provide some simulations to numerically validate the above procedure. 

\subsection{Joint dispersion index analysis for observations and innovations} \label{subsec:disp}

When it comes to the question of a suitable innovation distribution in the INAR model, dispersion plays a major role. In practice, one usually estimates the dispersion index $ID_X$ by the estimator $\widehat{ID}_X=S^2/\overline{X}$, where $S^2$ is the sample variance and $\overline{X}$ is the sample mean based on the observations $X_1,\ldots,X_n$. Then, one makes use of the INAR model to derive also an estimator for the dispersion index $ID_\epsilon$ of the innovations denoted by $\widehat{ID}_\varepsilon$. For instance, in the case of an INAR(1) model with coefficient $\alpha$, we have the relationship
\begin{align*}
    ID_X = \frac{ID_\varepsilon + \alpha}{1+\alpha}.
\end{align*}
Hence, we can estimate the innovation dispersion $ID_\varepsilon$ index by computing $\widehat{ID}_\varepsilon=
\widehat{ID}_X(1+\widehat \alpha)-\widehat \alpha$. However, the observations are over-/equi-/underdispersed if and only if the innovations are over-/equi-/underdispersed \citep[see Section 2.2.1 in][]{bookweiss}. Based on $\widehat{ID}_X$ one usually commits to a parametric family of innovation distributions and thus determines if the innovations are over-/equi-/underdispersed. However, for model diagnostics, the question remains whether a dispersion index deviates enough from $1$ to suggest under- or overdispersion instead of equidispersion. While there are tests based on $\widehat{ID}_X$ proposed, e.g., by \citet{schweer_weiss}, so far, there is no way to approach the innovations and their corresponding dispersion directly that is \emph{without} imposing a parametric assumption on the innovation distribution. However, by assuming a certain family of innovations for testing purposes, the dispersion is (often) fixed from the outset. For example, if a Poi-INAR model is assumed, the innovations will be equidispersed, while in the case of an NB-INAR model, they will be overdispersed. 

With the NPMLE estimation approach proposed of \citet{drost}, the innovation distribution and thus also its dispersion can be estimated directly without a parametric assumption. As described in the following, we exploit this approach together with the semi-parametric INAR bootstrap from Section \ref{sec:boot_algo} to construct confidence intervals for the dispersion index of the innovations (i.e., for $ID_\epsilon$) and of the observations (i.e., for $ID_X$) in one shot. The following algorithm is valid in case of an INAR(1) model: 
\begin{itemize}
    \item[Step 1.)] Given $X_1,\ldots, X_n$, calculate the NPMLE $\thn=(\widehat \alpha_n,\widehat G_n)$ and estimate (semi-parametrically) the dispersion indices of the innovations $\widehat{ID}_\varepsilon^{sp}$ and of the observations $\widehat{ID}_X^{sp}$ by computing
    \begin{align*}
    \widehat{ID}_\varepsilon^{sp} = \frac{\sum_{j \in \mathbb{N}_0}j^2\widehat{G}_n(j) - \left( \sum_{j \in \mathbb{N}_0}j\widehat{G}_n(j) \right)^2}{\sum_{j \in \mathbb{N}_0}j\widehat{G}_n(j)} \quad   \text{and}  \quad
    \widehat{ID}_X^{sp}= \frac{\widehat{ID}_\varepsilon^{sp} +\widehat{\alpha}_n}{1+\widehat{\alpha}_n}.
    \end{align*}
    \item[Step 2.)] Use the semi-parametric INAR bootstrap from Section \ref{sec:boot_algo} to get bootstrap observations $X_1^*,\ldots,X_n^*$ and to calculate the bootNPMLE $\thn^*=(\widehat \alpha_n^*,\widehat G_n^*)$ and estimate (semi-parametrically) the bootstrap dispersion indices of the innovations $\widehat{ID}_\varepsilon^{sp,*}$ and of the observations $\widehat{ID}_X^{sp,*}$ by
    \begin{align*}
    \widehat{ID}_\varepsilon^{sp,*} = \frac{\sum_{j \in \mathbb{N}_0}j^2\widehat{G}_n^*(j) - \left( \sum_{j \in \mathbb{N}_0}j\widehat{G}_n^*(j) \right)^2}{\sum_{j \in \mathbb{N}_0}j\widehat{G}_n^*(j)} \quad   \text{and}  \quad
    \widehat{ID}_X^{sp,*}= \frac{\widehat{ID}_\varepsilon^{sp,*} +\widehat{\alpha}_n^*}{1+\widehat{\alpha}_n^*}. 
    \end{align*}    
    \item[Step 3.)] Repeat Step 2.) $B$-times, where $B$ is large, to get $L_{\varepsilon,n}^{*,b} = \widehat{ID}_\varepsilon^{sp,*,b} - \widehat{ID}_\varepsilon^{sp}$ and  $L_{X,n}^{*,b} = \widehat{ID}_X^{sp,*,b} - \widehat{ID}_X^{sp}$, $b=1,\ldots,B$, and construct the $(1-\delta)$-confidence intervals for $ID_\varepsilon$ and $ID_X$ as $\left[\widehat{ID}_\varepsilon^{sp} - q^*_{\varepsilon,1-\delta/2}, \widehat{ID}_\varepsilon^{sp}-q^*_{\varepsilon,\delta/2}\right]$ and $\left[\widehat{ID}_X^{sp} - q^*_{X,1-\delta/2}, \widehat{ID}_X^{sp}-q^*_{X,\delta/2}\right]$, respectively, where $q_{\varepsilon,\alpha}^*$ and $q_{X,\alpha}^*$ denote the $\alpha$-quantiles of the empirical distributions of $L_{\varepsilon,n}^{*,b}$, $b=1,\ldots,B$ and of $L_{X,n}^{*,b}$, $b=1,\ldots,B$, respectively.
\end{itemize}
In Section \ref{sec:dispersion_sims}, we provide some simulations to validate the above procedure.



\section{Simulations}\label{sec:sim}

In this section, we investigate the finite sample performance of the semi-parametric INAR bootstrap from Subsection \ref{sec:boot_algo} and the methodological applications from Sections \ref{subsec:pp} and \ref{subsec:disp} by simulations.

\subsection{Bootstrap confidence intervals for $\theta_0=(\fett\alpha_0,G_0)$}

First, we illustrate the bootstrap performance for the task of confidence interval construction for the true model parameter $\theta_0=(\fett\alpha_0,G_0)$. That is, in a simulation study with $K=500$ Monte Carlo samples and $B=500$ bootstrap repetitions, we use the semi-parametric INAR bootstrap to 
construct Hall's confidence intervals  
for the entries of $\theta_0=(\fett\alpha_0,G_0)$. Its validity is justified by Theorem \ref{bs_clt} as it uses that $\sqrt{n}(\thns-\thn)$ provides in probability the same limiting distribution as $\sqrt{n}(\thn-\theta_0)$. Consequently, for significance level $\delta=5\%$, this results in the confidence interval
\begin{align}
\left[\thn - q_{1-\delta/2}^*, \thn - q_{\delta/2}^* \right],
\end{align}
where $q_{1-\delta/2}^*$ and $q_{\delta/2}^*$ are the corresponding quantiles of $\thns-\thn$. We consider sample sizes $n \in \{100,500,1000 \}$ and INAR(1) DGPs (i.e., $p=1$ in \eqref{inarp}) with different INAR coefficients $\alpha \in \{ 0.1,0.3,0.5,0.9 \}$ and with innovations following either a Poisson or a negative binomial distribution resulting in equi- and overdispersion, respectively. 
For implementing the simulation study, we use the R package \textit{spINAR} \citep{faymonville2024spinar}.

At first, we consider INAR(1) DGPs with different $\alpha$ parameter and Poi($\lambda$) distributed innovations with $\lambda \in \{1,3\}$. Table \ref{tab:poi_sp} contains the coverage and the average length of the computed confidence intervals in case of $\alpha= 0.5$ and $\lambda=1$. We see that the coverage increases for increasing $n$ approaching the desired coverage level of $95\%$, while the average length decreases. In this parameter setup, where the mean of the observations equals two and the values of $\widehat G(k)$ for $k > 4$ become tiny (not larger than 0.0031), we only consider the first six entries of the parameter vector, that is, $\alpha$ and $G(k)$, $k \in \{0,1,2,3,4\}$. As $G(k)$ already becomes small for rather small $k$, this also explains the comparably low coverage for $G(4)$. These small entries of the parameter vector are difficult to estimate since the corresponding values are rarely observed. In this setup, where we have $G$=Poi(1), we exemplary get $G(4) = 0.015$. The results for the other parameterizations can be found in Tables \ref{tab:poi_sp_1_01}, \ref{tab:poi_sp_3_01}, \ref{tab:poi_sp_1_03}, \ref{tab:poi_sp_3_03}, \ref{tab:poi_sp_3_05} and \ref{tab:poi_sp_1_09} in the appendix. Overall, we get comparable results. However, it has to be noted that the semi-parametric INAR bootstrap approach slightly loses in terms of coverage performance, when the observations' mean increases (which is the case for larger $\alpha$ and $\lambda$). However, this could have been expected as the number of parameters, i.e., the entries of the pmf, to be estimated also increases, when the mean of the innovations increases. In addition, the considered first five entries of the pmf, i.e., $G(k)$, $k \in \{0,1,2,3,4\}$, then only cover a smaller portion of the whole probability mass in cases with large innovations' (and observations') mean.

\begin{table}[t]
\centering
\begin{tabular}{c|ccc|ccc|ccc}
& \multicolumn{3}{|c|}{$\alpha$} & \multicolumn{3}{|c|}{$G(0)$} & \multicolumn{3}{|c}{$G(1)$} \\
$n$ & 100 & 500 & 1000 & 100 & 500 & 1000& 100 & 500 & 1000 \\
\hline
coverage & 0.896&0.932&0.960&0.802&0.914&0.918&0.880&0.924&0.942 \\
average length & 0.369&0.147&0.102&0.385&0.177&0.123&0.391&0.171&0.119 \\
& \multicolumn{3}{|c|}{$G(2)$} & \multicolumn{3}{|c|}{$G(3)$} & \multicolumn{3}{|c}{$G(4)$} \\
$n$ & 100 & 500 & 1000 & 100 & 500 & 1000& 100 & 500 & 1000 \\
\hline
coverage & 0.850&0.928&0.950&0.698&0.946&0.950&0.446&0.730&0.852 \\
average length & 0.328&0.137&0.094&0.181&0.082&0.056&0.072&0.035&0.027 \\
\end{tabular}
\vspace{0.5cm}
\caption{Coverage and average length of the bootstrap confidence intervals based on the semi-parametric INAR bootstrap from Section \ref{sec:boot_algo} for $\alpha, G(0), \ldots, G(4)$ in case of a Poi(1)-INAR(1) DGP with $\alpha=0.5$ for different sample sizes.}
\label{tab:poi_sp}
\end{table}

Next, we consider an INAR(1) DGP with $\alpha=0.5$ but now with an overdispered negative binomial innovation distribution NB($N, \, \pi$) with parameters ($N=10$, $\pi=10/11$), ($N=2$, $\pi=2/3$) and ($N=1$, $\pi=1/2$). These parameter choices also lead to an observations' mean of two, but they cover different levels of overdispersion (increasing in mentioned order). Table \ref{tab:nb_sp} displays the results for the case of moderate overdispersion ($N=2$, $\pi=2/3$). We see that for increasing $n$, we again approach the desired coverage level of $95\%$, whereas the average length of the confidence intervals decreases. In this setting of overdispersion, we have even better coverage and average length of the intervals compared to the results in Table \ref{tab:poi_sp}, where the used innovation distribution is equidispersed. For the parameterizations ($N=1$, $\pi=1/2$) and ($N=10$, $\pi=10/11$), we get similar results shown in Tables \ref{tab:nb_sp_112_05} and $\ref{tab:nb_p_101011_05}$ in the appendix.

\begin{table}[t]
\centering
\begin{tabular}{c|ccc|ccc|ccc}
& \multicolumn{3}{|c|}{$\alpha$} & \multicolumn{3}{|c|}{$G(0)$} & \multicolumn{3}{|c}{$G(1)$} \\
$n$ & 100 & 500 & 1000 & 100 & 500 & 1000& 100 & 500 & 1000 \\
\hline
coverage & 0.946&0.946&0.948&0.848&0.936&0.924&0.886&0.942&0.924 \\
average length & 0.337&0.130&0.090&0.380&0.164&0.115&0.361&0.157&0.110 \\
& \multicolumn{3}{|c|}{$G(2)$} & \multicolumn{3}{|c|}{$G(3)$} & \multicolumn{3}{|c}{$G(4)$} \\
$n$ & 100 & 500 & 1000 & 100 & 500 & 1000& 100 & 500 & 1000 \\
\hline
coverage & 0.846&0.934&0.934&0.734&0.958&0.948&0.568&0.814&0.924 \\
average length & 0.270&0.115&0.081&0.170&0.078&0.054&0.093&0.048&0.036 \\
\end{tabular}
\vspace{0.5cm}
\caption{Coverage and average length of the bootstrap confidence intervals based on the semi-parametric INAR bootstrap from Section \ref{sec:boot_algo} for $\alpha, G(0), \ldots, G(4)$ in case of a NB(2,2/3)-INAR(1) DGP with $\alpha=0.5$ for different sample sizes.}
\label{tab:nb_sp}
\end{table}

\begin{table}[t]
\centering
\begin{tabular}{c|ccc|ccc|ccc}
& \multicolumn{3}{|c|}{$\alpha$} & \multicolumn{3}{|c|}{$G(0)$} & \multicolumn{3}{|c}{$G(1)$} \\
$n$ & 100 & 500 & 1000 & 100 & 500 & 1000& 100 & 500 & 1000 \\
\hline
coverage & 0.926&0.948&0.964&0.906&0.928&0.936&0.984&0.978&0.978 \\
average length & 0.298&0.128&0.090&0.230&0.103&0.074&0.039&0.008&0.006 \\
& \multicolumn{3}{|c|}{$G(2)$} & \multicolumn{3}{|c|}{$G(3)$} & \multicolumn{3}{|c}{$G(4)$} \\
$n$ & 100 & 500 & 1000 & 100 & 500 & 1000& 100 & 500 & 1000 \\
\hline
coverage & 0.848&0.918&0.938&0.866&0.924&0.938&0.846&0.908&0.930 \\
average length & 0.108&0.051&0.037&0.078&0.035&0.025&0.033&0.013&0.010 \\
\end{tabular}
\vspace{0.5cm}
\caption{Coverage and average length of the parametrically constructed bootstrap confidence intervals based on \emph{parametric} ML estimation and a \emph{parametric} Poi-INAR bootstrap for $\alpha, G(0), \ldots, G(4)$ in case of a Poi(1)-INAR(1) DGP with $\alpha=0.5$ for different sample sizes.}
\label{tab:poi_p}
\end{table}

\begin{table}[t]
\centering
\begin{tabular}{c|ccc|ccc|ccc}
& \multicolumn{3}{|c|}{$\alpha$} & \multicolumn{3}{|c|}{$G(0)$} & \multicolumn{3}{|c}{$G(1)$} \\
$n$ & 100 & 500 & 1000 & 100 & 500 & 1000& 100 & 500 & 1000 \\
\hline
coverage & 0.964&0.760&0.490&0.446&0.044&0&0.034&0&0 \\
average length & 0.320&0.137&0.096&0.223&0.101&0.071&0.046&0.013&0.008 \\
& \multicolumn{3}{|c|}{$G(2)$} & \multicolumn{3}{|c|}{$G(3)$} & \multicolumn{3}{|c}{$G(4)$} \\
$n$ & 100 & 500 & 1000 & 100 & 500 & 1000& 100 & 500 & 1000 \\
\hline
coverage & 0.422&0.050&0&0.884&0.848&0.782&0.668&0.512&0.330 \\
average length & 0.102&0.049&0.035&0.086&0.039&0.027&0.041&0.017&0.012 \\
\end{tabular}
\vspace{0.5cm}
\caption{Coverage and average length of the parametrically constructed bootstrap confidence intervals based on \emph{parametric} ML estimation and a \emph{parametric} Poi-INAR bootstrap for $\alpha, G(0), \ldots, G(4)$ in case of a NB(2,2/3)-INAR(1) DGP with $\alpha=0.5$ for different sample sizes.}
\label{tab:nb_p}
\end{table}

In this paragraph, in Tables \ref{tab:poi_p} and \ref{tab:nb_p}, we compare the results of the semi-parametric INAR bootstrap that does not rely on any parametric assumption on the innovation distribution with a fully parametric approach that assumes that the INAR(1) data at hand follows a Poisson distribution. In the latter case, we use the parametric maximum-likelihood method for parameter estimation \citep[see, e.g.,][for details]{bookweiss} together with the \emph{parametric} INAR bootstrap of \textcite{jewe}. Both are also implemented in the R package \textit{spINAR} \citep{faymonville2024spinar}. In Tables \ref{tab:poi_p} as well as in Tables \ref{tab:poi_p_1_01}, \ref{tab:poi_p_3_01}, \ref{tab:poi_p_1_03}, \ref{tab:poi_p_3_03}, \ref{tab:poi_p_3_05} and \ref{tab:poi_p_1_09} in the appendix, we see the results for the already considered Poisson DGPs, but now using the parametric approach. As one could expect, on the one hand, the procedure leads to (considerably) smaller confidence intervals while providing similar or even better coverage. This is not surprising since we use additional \emph{true} information of the DGP. In contrast, on the other hand, Tables \ref{tab:nb_p}, \ref{tab:nb_p_112_05} and \ref{tab:nb_p_101011_05} show what happens, when we use this additional distributional assumption in cases where it does not hold. The displayed results arise from applying parametric estimation and bootstrapping assuming a Poi($\lambda$) innovation distribution in a case, where the innovation distribution is actually \emph{not} Poisson, but negative binomial. Although the considered NB($N, \pi$) innovation distribution has the same mean as before, the procedure breaks down. As expected, the confidence intervals do not hold the prescribed coverage rate. In particular, the coverage performance depends on the level of overdispersion. While the parameterization of $N=10$ and $\pi=10/11$ only leads to slight overdispersion, close to equidispersion, and not so inferior results (see Table \ref{tab:nb_p_101011_05}), the results for the parameterizations with higher overdispersion become worse (see Table \ref{tab:nb_p_112_05}). Moreover, the coverage even decreases for increasing $n$, in some cases drastically. This again underlines the benefit of the semi-parametric approach. Without having to rely on any distribution assumptions that may not hold in practice, we achieve good results even compared to the parametric approach when a correct distribution assumption is used. 

\begin{figure}[t]
\centering
\includegraphics[scale=0.65]{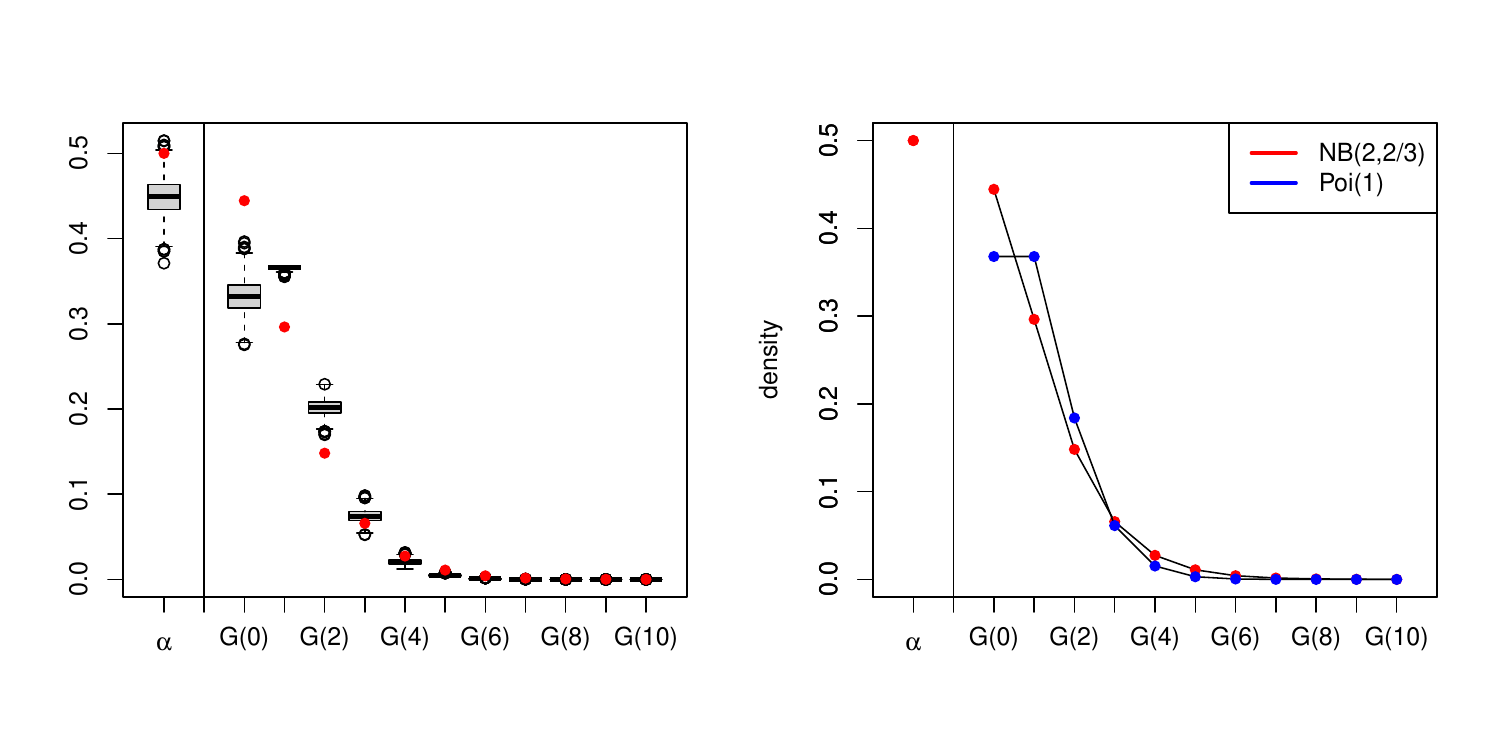} 
\vspace{-0.8cm}
\caption{Left panel: Boxplots of the point estimators for the first twelve entries of the parameter vector $\theta$ for each Monte Carlo sample in case of $n=1000$ with the true values being displayed in red. Right panel: pmfs of the NB(2,2/3) distribution (red) and the Poi(1) distribution (blue).}
\label{fig:gemZD}
\end{figure}

To get an intuition why the coverages may be extremely poor in some cases (even being equal to 0), consider the case of negative binomially distributed innovations with parameters $N=2$ and $\pi=2/3$, see Figure \ref{fig:gemZD}. In the left panel, we see the true parameter values ($\alpha=0.5$ and the first eleven entries of the pmf in case of a NB(2,2/3) distribution) in red. The boxplots show the corresponding point estimates for each of the 500 Monte Carlo samples in case of $n=1000$ when assuming an underlying Poisson distribution. It can be seen that some of these point estimators drastically over- or underestimate the true values. This is particularly the case for $G(0), G(1)$ and $G(2)$. This can be explained by the right panel, where, along with the true value of $\alpha$, the pmf of a NB(2,2/3) distribution (red) and the one of a Poisson distribution providing the same mean (blue) are displayed. As can be seen, these two pmfs differ mainly (absolutely) in their first three entries. 

Another finding is that although only the innovation distribution is misspecified, this misspecification also has a tremendous negative effect on the (point) estimation of the INAR coefficient $\alpha$ itself, as can be seen in Figure \ref{fig:gemZD}. Consequently, this also leads to a (too) low coverage of its confidence intervals. This is due to the fact that we use parametric (conditional)  maximum likelihood estimation for implementing the parametric INAR bootstrap that estimates $\alpha$ and $\lambda$ simultaneously. If we were to use moment estimators, i.e., least-squares or Yule-Walker estimation, to get an estimator for $\alpha$ in the first step and then to get an estimator for $\lambda$ in the second step, this bias effect when estimating $\alpha$ would be avoided.
However, in setups of a correctly specified Poi-INAR(1) process, the Yule-Walker estimation method is usually inferior to the maximum likelihood approach in finite samples such that the latter is mostly preferred in practice. In view of the results shown in Table \ref{tab:nb_p}, however, this should be done with care when doing inference.

\begin{table}[t]
\centering
\begin{tabular}{lr|rrrr}
            &       & \multicolumn{4}{c}{} \\
 DGP &$n$& $100$ & $500$ & $1000$ & $5000$ \\ 
\hline
Poi-INAR & coverage & 0.918&0.938&0.944&0.946 \\ 
  &average length & 0.210&0.089&0.066&0.028 \\ 
  \hline
   NB-INAR & coverage & 0.916&0.948&0.946&0.944    \\ 
  &average length & 0.216&0.090&0.064&0.029  \\ 
  \hline
  INARCH & coverage & 0.916&0.752&0.558&0.372  \\ 
 &average length & 0.215&0.094&0.065&0.029  \\  
\end{tabular}
\vspace{0.3cm}
\caption{Coverage and average length of the confidence intervals for $P_{S,x_n}$ for different sample sizes and different true DGPs.}
\label{tab:pp_sp}
\end{table}

\subsection{Predictive inference for count time series}\label{sec:pred_inf_sims}

To validate the procedure described in Section \ref{subsec:pp}, we set up simulations covering the same DGPs as in \citet{pred_paper}, namely Poi(1)-INAR(1) and NB(2,2/3)-INAR(1) with both $\alpha=0.5$ and INARCH(1) with $\alpha=0.5$ and $\beta=1$. 
The results are displayed in Table \ref{tab:pp_sp} and are as expected. In the case, when the made assumption of an INAR model is true, the coverage of the confidence intervals increases towards $95\%$ for increasing $n$, while the average interval length decreases. However, in case of the INARCH DGP, the INAR assumption is violated resulting in decreasing coverage for increasing $n$. 

\subsection{Joint dispersion index analysis for observations and innovations}\label{sec:dispersion_sims}

To validate the performance of the bootstrap algorithm proposed in Section \ref{subsec:disp}, we applied it again on a Poi(1)-INAR(1) DGP and an NB(2,2/3)-INAR(1) DGP both with $\alpha=0.5$. Table \ref{tab:disp_inds} displays the results. We see that the coverage increases towards the desired coverage level of $95\%$, while the average length decreases. In the setup of the INAR model with negative binomial innovations, the intervals are systematically larger which can be explained by the underlying overdispersion of both the innovations and the observations compared to the equidispersion in case of the Poi-INAR model. 

\begin{table}[t]
\centering
\begin{tabular}{lr|rrrr|rrrr}
            &       & \multicolumn{4}{c|}{innovations}& \multicolumn{4}{c}{observations} \\
 DGP &$n$& $100$ & $500$ & $1000$ & $5000$& $100$ & $500$ & $1000$ & $5000$ \\ 
\hline
Poi-INAR & cov &0.838&0.944&0.942&0.942 &0.876&0.950&0.940&0.942\\
& ave &0.778&0.365&0.260&0.118 &0.545&0.246&0.175&0.079 \\
\hline
NB-INAR & cov &0.854&0.910&0.928&0.940 &0.878&0.918&0.932&0.938\\
 & ave &1.139&0.550&0.393&0.179 &0.771&0.361&0.257&0.117
\end{tabular}
\vspace{0.3cm}
\caption{Coverage (cov) and average length (ave) of the confidence intervals for the dispersion indices for the innovations and the observations for different sample sizes and different true DGPs.}
\label{tab:disp_inds}
\end{table}

\section{Real-Data Application} \label{sec:appl}


As a real-data application, based on the results of Theorem \ref{bs_clt} and Corollary \ref{cor:delta_method_boot_consistency}, we apply the procedures described in Sections \ref{subsec:pp} and \ref{subsec:disp} to semi-parametrically perform predictive inference and to construct confidence intervals for observations' and innovations' dispersion indices. 

In one of their real-data applications, \citet{pred_paper} considered a data set of the Deutsche Börse Group containing $n=404$ transaction counts of structured products per trading day. The data along with the corresponding (P)ACF can be found in Figure \ref{fig:transa}. Based on the test result of \citet{gof_sp} (see also Section \ref{subsec:gof}) to not reject the semi-parametric null of an INAR(1) model at $5\%$ level, they performed non-parametric as well as parametric predictive inference by assuming an INAR(1) model. For the parametric approaches, they separately imposed either a Poisson distribution or a geometric distribution for the innovations. They constructed $95\%$ confidence intervals for $P_{S,x_n}:=P(X_{n+1} \in S | X_n =x_n)$. As sets of interest $S$, they chose $S=\{0\}$ and $S=\{2\}$. We now apply the semi-parametric procedure described in Section \ref{subsec:pp} and give the point estimates along with confidence intervals for $P_{S,x_n}$. We display our results along with the ones of \citet{pred_paper} (in italic) in Table \ref{tab:pp_realdata}. With their results, \citet{pred_paper} could underline the results obtained in \citet{gof_sp} that a Geo-INAR(1) model might be a good fit to the data. Our semi-parametric result fits well into this picture. The corresponding confidence intervals are nearly disjoint with the Poi-INAR ones and include the Geo-INAR ones, while being slightly shorter than the non-parametric ones. 

To illustrate the procedure proposed in Section \ref{subsec:disp}, we additionally construct confidence intervals for the dispersion indices of the innovations and the observations, see Table \ref{tab:disp_inds_realdata} for the result. Both confidence intervals do not include the 1. Hence, both intervals suggest that the data and also the innovations are overdispersed.

\begin{figure}[t]
\centering
\includegraphics[scale=0.65]{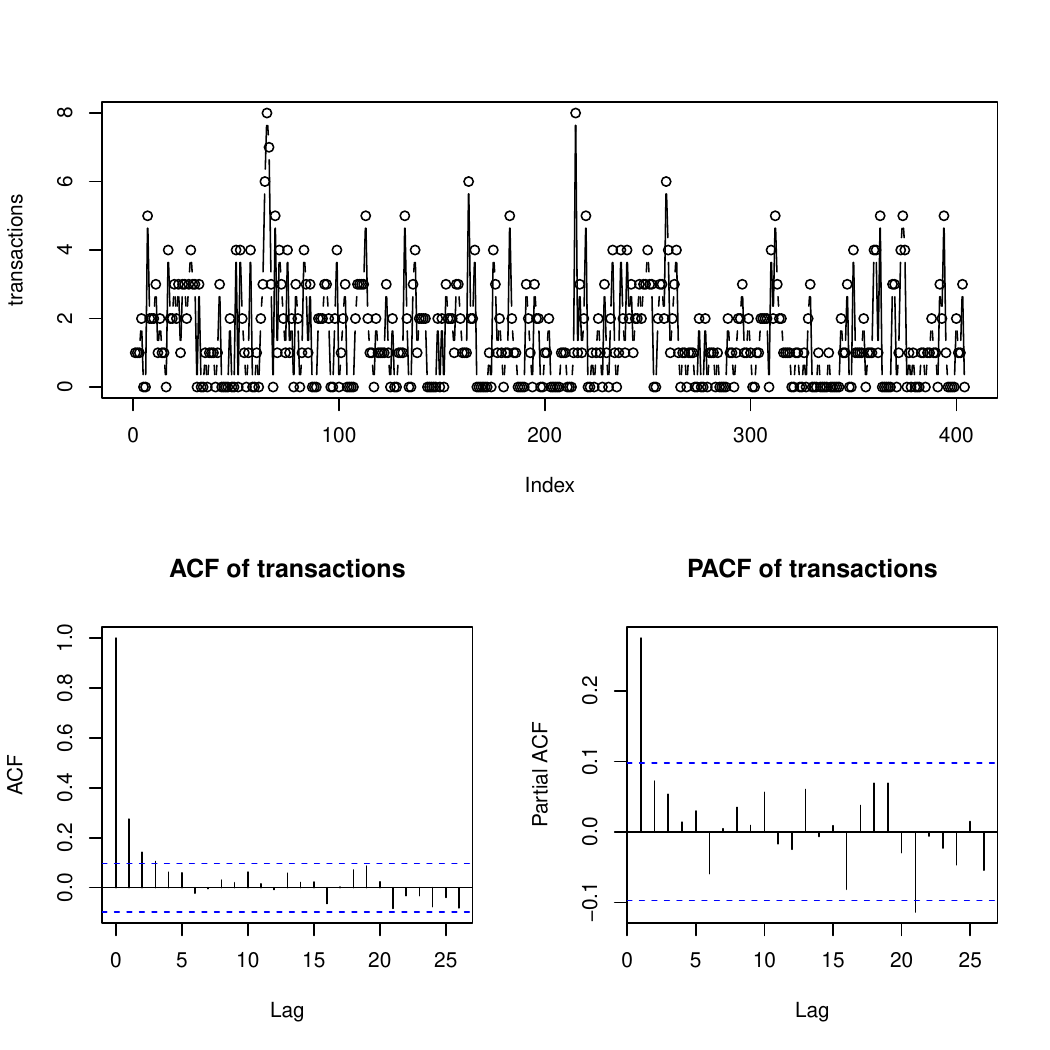} 
\caption{Plot of transaction counts and the corresponding (P)ACF (reproduced from \citet{gof_sp}.}
\label{fig:transa}
\end{figure}

\begin{table}[t]
\centering
\begin{tabular}{cc|cc}
set & assumption  & point estimation & CI \\ 
\hline
$S=\{0\}$ & semi-param. INAR & 0.4593  & [0.3859, 0.5386] \\
&\textit{Poi-INAR}  & \textit{0.3203 }& \textit{[0.2720,0.3676]} \\ 
&\textit{Geo-INAR}  & \textit{0.4929} & \textit{[0.4469,0.5290]} \\ 
& \textit{non-param.}  & \textit{0.4884} & \textit{[0.4023,0.5752]} \\
\hline
$S=\{2\}$ & semi-param. INAR & 0.1425 & [0.0897, 0.1926] \\
&\textit{Poi-INAR}  & \textit{0.2076 }& \textit{[0.1841, 0.2343]} \\ 
&\textit{Geo-INAR}  & \textit{0.1267} & \textit{[0.1189, 0.1365]} \\ 
& \textit{non-param.}  & \textit{0.1318} & \textit{[0.0676, 0.1824]} 
\end{tabular} 
\vspace{0.2cm}
\caption{Point estimations and confidence intervals resulting from the semi-parametric procedure described in Section \ref{subsec:pp} for two different sets $S$  along with the (non-)parametric results of \citet{pred_paper} displayed in italic.}
\label{tab:pp_realdata}
\end{table}


\begin{table}[t]
\centering
\begin{tabular}{l|cc}
            &  point estimation     & CI \\
            \hline
innovations & 1.6180 & [1.3342, 1.8496] \\
observations &  1.4832   & [1.2660, 1.6593]
\end{tabular}
\vspace{0.3cm}
\caption{Point estimation and confidence intervals for the dispersion indices of innovations and observations.}
\label{tab:disp_inds_realdata}
\end{table}

\section{Conclusion}\label{sec:concl}
We proposed a semi-parametric INAR bootstrap procedure that allows for joint inference of the INAR model coefficients and the innovation distribution and circumvents the need for a cumbersome estimation of the limiting distribution. By using such a semi-parametric setting, we avoid the use of any parametric assumptions on the innovation distribution. Indeed, such parametric assumptions may be too restrictive in practice and can have a large negative impact on the conducted inference, when they do not hold. Based on the semi-parametric INAR bootstrap proposed by \textcite{jewe}, which they showed to be consistent for functions of generalized means, we established bootstrap consistency results by proving a corresponding bootstrap central limit theorem. 
We illustrated the usefulness of our results by several methodological applications including goodness-of-fit testing, predictive inference and joint dispersion index analysis.
In simulations, we illustrated our theoretical results by analyzing the coverage and the length of confidence intervals for different INAR model parameters in various setups. Concluding, the semi-parametric INAR bootstrap performs well in comparison to both non-parametric and parametric approaches. In particular, it turns out to be a practically relevant alternative that allows for robust inference in cases, where parametric assumptions are falsely imposed.

\section*{Acknowledgments}

Financial support by the Deutsche Forschungsgemeinschaft (DFG, German Research Foundation) Project ID 437270842 (Model Diagnostics for Count Time Series) and 520388526 (TRR 391: Spatio-temporal Statistics for the Transition of Energy and Transport, Project A03) is gratefully acknowledged. Additionally, the authors gratefully acknowledge the computing time provided on the Linux HPC cluster
at TU Dortmund University (LiDO3), partially funded in the course of the Large-Scale Equipment
Initiative by the German Research Foundation (DFG) as project 271512359.

\newpage                           
\renewcommand*{\refname}{Quellen}  
\printbibliography[title={References}]

\newpage

\appendix

\pagestyle{plain}




\section{Proofs of the Main Results}\label{sec:appendixA}



\subsection{Proof of Theorem \ref{cons_bs_est}}
For proving the estimation consistency of the bootstrap estimator $\thns$, we follow the proof idea of Theorem \ref{cons} that is given as Theorem 1 in \citet{drost}. For this purpose, we adopt the notation introduced in Section B.2 of \textcite{drost2009appendix} and make use of similar arguments outlined in the following. Let $\thn=(\widehat{\fett\alpha}_n, \widehat{G}_n)$ be an NPMLE of $(\fett\alpha_0, G_0)$ based on the sample $X_{-p},\ldots,X_n$ and let $\thn^*=(\widehat{\fett\alpha}_n^*, \widehat{G}_n^*)$  be a bootNPMLE based on a bootstrap sample $X_{-p}^*,\ldots,X_n^*$ that is obtained from the bootstrap proposal from Section \ref{sec:boot_algo}. Hence, the goal is to show that
\[ \widehat{\fett\alpha}_n^* - \widehat{\fett\alpha}_n  \overset{p^*}{\rightarrow} 0 \quad \text{and} \quad \sum \limits_{k=0}^\infty | \widehat{G}_n^*(k) - \widehat{G}_n(k) |  \overset{p^*}{\rightarrow} 0  
\] 
holds in $\mathds{P}_{\nu_{\fett\alpha_0,G_0}, \fett\alpha_0, G_0}$-probability, where $P^*(\cdot)= \mathds{P}_{\nu_{\widehat{\fett\alpha}_n,\widehat{G}_n}, \widehat{\fett\alpha}_n, \widehat{G}_n}(\cdot|\mathbb{X})$.
To prove this result, 
using that both $(\widehat G_n^*(k),k\in\mathbb{N}_0)$ and $(\widehat G_n(k),k\in\mathbb{N}_0)$ sum up to one by construction, i,e., $\sum_{k=0}^\infty \widehat G_n^*(k)=1$ and $\sum_{k=0}^\infty \widehat G_n(k)=1$ according to \textcite{drost2009appendix}, it suffices to prove $\widehat{\fett\alpha}_n^* - \widehat{\fett\alpha}_n  \overset{p^*}{\rightarrow} 0$ as well as $\widehat{G}_n^*(k) - \widehat{G}_n(k)\overset{p^*}{\rightarrow} 0$ for all $k\in\mathbb{N}_0$. We prove the latter by following the arguments of Wald's consistency theorem \citep[see, e.g., the proof of Theorem 5.14 in][]{vandervaart}
and extend it to the bootstrap case. For this purpose, we have to consider the compactification of the parameter space analogously to \citet{drost}. First, we introduce $\overline{\mathcal{G}}$, which denotes the class of all probability measures on $\mathbb{N}_0 \cup \infty$ and identify each $G\in\overline{\mathcal{G}}$ by its probability mass function (pmf) sequence $(G(k),k\in\mathbb{N}_0)$. This correspondence is 1-to-1 by the relationship $G(\infty) = 1-\sum_{k\in\mathbb{N}_0} G(k)$. Hence, $\overline{\mathcal{G}}$ is a subset of $[0,1]^{\mathbb{N}_0}$ equipped with the norm $||a|| = \sum_{k=0}^\infty 2^{-k}|a(k)|$ for $a=(a(k),k\in\mathbb{N}_0)$. Following the arguments in \citet{drost} and \citet{drost2009appendix}, $\overline{\mathcal{G}}$ is a compact subset of $[0,1]^{\mathbb{N}_0}$. Further, for $G\in\overline{\mathcal{G}}$, we define $P_{x, \infty}^{\fett\alpha, G} = 1 - \sum_{j \in \mathbb{N}_0} P_{x, j}^{\fett\alpha, G} = G(\infty)$ for $x \in \mathbb{N}_0^p$ and $P_{x, \infty}^{\fett\alpha, G}=1$ if $\text{max}_{i=1}^p x_i = \infty$. Second, by considering the compactification of the parameter space $\Theta=(0,1)^p$ of $\theta$ as well, i.e., $[0,1]^p$, we define $\overline{E} := [0,1]^p\times \overline{\mathcal{G}}$ equipped with the ``sum distance'' $d((\fett\alpha, G),(\fett\alpha', G'))= |\fett\alpha-\fett\alpha'| + ||(G(k),k \in \mathbb{N}_0) - (G'(k),k \in \mathbb{N}_0) ||$. As $\overline{E}$ is the product of two compact spaces, it is itself compact.

Further, for $x_{-p}, \ldots, x_0\in\mathbb{N}_0^{p+1}$, we define $m^{\fett\alpha, G}(x_{-p}, \ldots, x_0) = \log P^{\fett\alpha, G}_{(x_{-1}, \ldots, x_{-p}), x_0}$, where $P^{\fett\alpha, G}_{(x_{-1}, \ldots, x_{-p}), x_0}$ is defined in \eqref{tpinar}, and introduce the (random) function $M_n: \overline{E} \rightarrow [-\infty, \infty)$ by
\begin{align*}
M_n(\fett\alpha, G)=\frac{1}{n} \sum\limits_{t=0}^n m^{\fett\alpha, G}(X_{t-p}, \ldots, X_t).    
\end{align*}
and its bootstrap analogue $M_n^*: \overline{E} \rightarrow [-\infty, \infty)$ by
\begin{align*}
M_n^*(\fett\alpha, G)=\frac{1}{n} \sum\limits_{t=0}^n m^{\fett\alpha, G}(X^*_{t-p}, \ldots, X_t^*).    
\end{align*}
We have to show that $M_n^*(\fett\alpha, G) - M_n(\fett\alpha, G)\overset{p^*}{\rightarrow} 0$ in probability for any $(\fett\alpha, G)\in \overline{E}$. Noting that
\begin{align*}
\text{E}^*(M_n^*(\fett\alpha, G))=M_n(\fett\alpha, G)
\end{align*}
holds by construction \citep[see, e.g., (A.6) in ][]{pred_paper},
we need to show that $M_n^*(\fett\alpha, G) - \text{E}^*(M_n^*(\fett\alpha, G))\overset{p^*}{\rightarrow}$ 0 in probability. For this purpose, it remains to show that $\text{Var}^*(M_n^*(\fett\alpha, G)) \overset{p^*}{\rightarrow} 0$ in probability as $n \rightarrow \infty$, where
\begin{align*}
\text{Var}^*(M_n^*(\fett\alpha, G)) = \text{Var}^* \left(   \frac{1}{n} \sum\limits_{t=0}^n m^{\fett\alpha, G} (X^*_{t-p}, \ldots, X_t^*)      \right)= \text{Var}^* \left(   \frac{1}{n} \sum\limits_{t=0}^n \log P^{\fett\alpha, G}_{(X^*_{t-1}, \ldots, X_{t-p}^*),X_t^*}      \right).
\end{align*}
For notational convenience, we focus on the case of an INAR(1) model, i.e., $p=1$ exclusively, but the following can be extended analogously to the case of higher order INAR models. Using (\ref{tpinar1}), we get
\begingroup
\allowdisplaybreaks
\begin{align} 
& \quad \text{Var}^*(M_n^*(\alpha, G)) \notag \\ 
&=\text{Var}^* \left( \frac{1}{n} \sum\limits_{t=0}^n \log \left( \sum\limits_{j=0}^{\text{min}(X_t^*,X_{t-1}^*)} \binom{X_{t-1}^*}{j}\alpha^j(1-\alpha)^{X_{t-1}^*-j} G(X_t^*-j) \right)   \right) \notag \\
&= \frac{1}{n^2} \sum\limits_{t_1,t_2=0}^n \text{Cov}^* \left( \log \left( \sum\limits_{j_1=0}^{\text{min}(X_{t_1}^*,X_{t_1-1}^*)} \binom{X_{t_1-1}^*}{j_1}\alpha^{j_1}(1-\alpha)^{X_{t_1-1}^*-j_1} G(X_{t_1}^*-j_1) \right)\right. , \notag \\*
& \qquad \qquad \qquad \quad \quad \; \left.\log \left( \sum\limits_{j_2=0}^{\text{min}(X_{t_2}^*,X_{t_2-1}^*)} \binom{X_{t_2-1}^*}{j_2}\alpha^{j_2}(1-\alpha)^{X_{t_2-1}^*-j_2} G(X_{t_2}^*-j_2)      \right) \right) \notag \\
&= \frac{1}{n}\sum\limits_{h=-(n-1)}^{n-1} \frac{1}{n} \sum\limits_{t=\text{max}(1,1-h)}^{\text{min}(n,n-h)} \text{Cov}^* \left(  \log \left( \sum\limits_{j_1=0}^{\text{min}(X_t^*,X_{t-1}^*)} \binom{X_{t-1}^*}{j_1}\alpha^{j_1}(1-\alpha)^{X_{t-1}^*-j_1}G(X_t^*-j_1) \right)\right. , \notag \\
& \qquad \qquad \qquad \qquad \quad \left.\log \left( \sum\limits_{j_2=0}^{\text{min}(X_{t+h}^*,X_{t+h-1}^*)} \binom{X_{t+h-1}^*}{j_2}\alpha^{j_2}(1-\alpha)^{X_{t+h-1}^*-j_2}G(X_{t+h}^*-j_2)  \right) \right). \notag 
\end{align}
Due to the (strict) stationarity of $(X^*_t,t\in\mathbb{Z})$ (conditional on the data), the covariances on the last right-hand side only depend on $h$ (and not on $t$). Further, due to $\min(n,n-h)-\max(1,1-h)+1=n-|h|\leq n$, the last right-hand side can be bounded by
\begin{align} 
& \frac{1}{n}\sum\limits_{h=-(n-1)}^{n-1}  \text{Cov}^* \left(  \log \left( \sum\limits_{j_1=0}^{\text{min}(X_1^*,X_0^*)} \binom{X_0^*}{j_1}\alpha^{j_1}(1-\alpha)^{X_0^*-j_1} G(X_1^*-j_1) \right)\right. , \label{cov_def2} \\
 & \qquad  \qquad \quad \quad \;  \; \; \left.\log \left( \sum\limits_{j_2=0}^{\text{min}(X_{h+1}^*,X_h^*)} \binom{X_h^*}{j_2}\alpha^{j_2}(1-\alpha)^{X_h^*-j_2}G(X_{h+1}^*-j_2)  \right) \right).\notag
\end{align}
\endgroup
As argued by \citet{drost} in their Lemma 1(b), the INAR($p$) process (here we focus on $p=1$) is (geometrically) $\beta$-mixing under $P_{\nu_{\theta,G},\theta,G}$ for $(\theta,G)\in\Theta\times \mathcal{G}$, which is imposed by Assumption \ref{ass1}. Consequently, according to Lemma \ref{lemma_rng}, with probability tending to one, the same holds for the bootstrap process $(X_t^*,t \in \mathbb{Z})$. Next, we define 
\begin{align}Y_h^* := f(X_h^*, X_{h+1}^*):=  \log \left(   \sum\limits_{j=0}^{\min(X_{h+1}^*, X_h^*)} \binom{X_h^*}{j} \alpha^j(1-\alpha)^{X_h^*-j}  G(X_{h+1}^*-j)  \right)    \label{def_Y_0*}
\end{align} 
with an obvious notation for $f(\cdot,\cdot)$. As the argument of $f$ consists of finitely many $X_t^*$'s, using Theorem 14.1 of \citet{davidson} and that $\beta$-mixing implies $\alpha$-mixing \citep{bradley_mix}, we get that $(f(X_h^*, X^*_{h+1},h\in\mathbb{Z})$ is also geometrically $\alpha$-mixing with probability tending to one.
Hence, using Corollary 14.3 in \citet{davidson}, we get 
\begin{align} \label{cov_ineq}
|\text{Cov}^*(Y_0^*, Y_{h}^* )| \leq 2\, (2^{1-1/q}+1) \, \alpha_{\text{mix},n}(h)^{1-1/q-1/r} \, ||Y_0^*||^*_q\; ||Y_{h}^*||^*_r
\end{align}
with $q>1, r>q/(q-1)$, where $||X_i^*||^*_s=(\text{E}^*(|X_i^*|^s))^{1/s}$, $s \in \{q,r\}$ and $ \alpha_{\text{mix},n}(h)$ denotes the $\alpha$-mixing coefficient of the bootstrap process  $(X_t^*,t \in \mathbb{Z})$ at lag $h$, not to confuse with the coefficient estimator $\widehat \alpha_n$ 
of the INAR model. Then, as the bootstrap process is geometrically mixing with probability tending to one according to Lemma \ref{lemma_rng} and the sequence $(\alpha_{\text{mix},n}(h)^{1-1/q-1/r},h\in\mathbb{N}_0)$ converges also exponentially fast to zero for $h\rightarrow \infty$, the covariances $\text{Cov}^*(Y_0^*, Y_{h}^* )$ become absolutely summable with probability tending to one, whenever $||Y_0^*||^*_q$ and $||Y_{h}^*||^*_r$ are bounded in probability for some suitable $q$ and $r$. The above term $2\, (2^{1-1/q}+1)$ is just a constant. 

Hence, to conclude that the covariances in \eqref{cov_ineq} are absolutely summable with probability tending to one, it remains to show that the moments of the function of the bootstrap data are bounded in probability, i.e.
\begin{align} \label{eq:qr}
\exists q > 1, r > q/(q-1):||Y_0^*||^*_q=O_p(1) \quad \text{and} \quad ||Y_{h}^*||^*_r=O_p(1).
\end{align} 
Again, due to the strict stationarity of $(X_t^*,t \in \mathbb{Z})$ (conditional on the data), it suffices to show the previous for $h=0$ and for some $q>2$. If we want to prove that $\text{E}^*(|Y_0^*|^q)^{1/q}$ is bounded in probability, it remains to prove that $\text{E}^*(|Y_0^*|^q)$ is bounded in probability. Plugging in for $Y_0^*$ as defined in \eqref{def_Y_0*} leads to
\[
\text{E}^*(|Y_0^*|^q) = \sum\limits_{(x,y) \in \mathbb{N}_0^2} \left|  \log \left( \sum\limits_{j=0}^{\min(x,y)} \binom{x}{j} \alpha^j(1-\alpha)^{x-j}G(y-j) \right) \right|^q P^*\big( (X_0^*,X_1^*)=(x,y)  \big),
\]
with $P^*(\cdot)= \mathds{P}_{\nu_{\widehat{\fett\alpha}_n,\widehat{G}_n}, \widehat{\fett\alpha}_n, \widehat{G}_n}(\cdot|\mathbb{X})$.
Let $k>0$ and consider the probability that the last right-hand side is larger than $k$, i.e.,
\begin{align*}
& P \left(  \sum\limits_{(x,y) \in \mathbb{N}_0^2} \left|  \log \left( \sum\limits_{j=0}^{\min(x,y)} \binom{x}{j} \alpha^j(1-\alpha)^{x-j}G(y-j) \right) \right|^q P^*( (X_0^*,X_1^*)=(x,y)  ) \geq k   \right) \\
\leq \, & P \left(  \sum\limits_{(x,y) \in \mathbb{N}_0^2, y \leq x} \left|  \log \left( \sum\limits_{j=0}^y \binom{x}{j} \alpha^j(1-\alpha)^{x-j}G(y-j) \right) \right|^q P^*( (X_0^*,X_1^*)=(x,y)  )    \geq k/2  \right) \\ 
& \, + P \left(  \sum\limits_{(x,y) \in \mathbb{N}_0^2, y > x} \left|  \log \left( \sum\limits_{j=0}^x \binom{x}{j} \alpha^j(1-\alpha)^{x-j}G(y-j) \right) \right|^q P^*( (X_0^*,X_1^*)=(x,y)  )   \geq k/2  \right) \\
 =: & I + II.
\end{align*}
Let us consider the first term $I$. As $\sum_{j=0}^y \binom{x}{j} \alpha^j(1-\alpha)^{x-j}G(y-j) = P_{x,y}^{\alpha,G} \in [0,1]$, we have $\log(\sum_{j=0}^y \binom{x}{j} \alpha^j(1-\alpha)^{x-j}G(y-j))\leq 0$. Hence, using that $-\log(\cdot)$ is a decreasing function and that all summands in $\sum_{j=0}^y \binom{x}{j} \alpha^j(1-\alpha)^{x-j}G(y-j)$ are non-negative, we get an upper bound by only considering the last term of the sum, i.e., $-\log(\sum_{j=0}^y \binom{x}{j} \alpha^j(1-\alpha)^{x-j}G(y-j))\leq -\log(\binom{x}{y} \alpha^y (1-\alpha)^{x-y} G(0))$ and we have  
\begin{align*}
\left| \log \left( \sum\limits_{j=0}^y \binom{x}{j} \alpha^j(1-\alpha)^{x-j}G(y-j) \right) \right|^q &=  \left(- \log \left( \sum\limits_{j=0}^y \binom{x}{j} \alpha^j(1-\alpha)^{x-j}G(y-j) \right) \right) ^q \\
& \leq \left( -\log \left( \binom{x}{y} \alpha^y (1-\alpha)^{x-y} G(0)  \right) \right)^q   \\ 
& \leq \Big(  -\log \left( \alpha^y (1-\alpha)^{x-y} G(0) \right)  \Big)^q 
\\
&= \Big(-y \log(\alpha) - (x-y) \log(1-\alpha) - \log(G(0))\Big)^q, 
\end{align*}
where we used that $\binom{x}{y}\geq 1$ such that $-\log(\binom{x}{y}z)\leq -\log(z)$ holds for all $z\geq 0$. Finally, using that $\log(x)<0$ for $x \in (0,1)$ and that $0<G(0)<1$ as well as $\alpha\in(0,1)$ by Assumption \ref{ass1}, we see that $-y \log(\alpha) - (x-y) \log(1-\alpha) - \log(G(0))$ is strictly positive and finite. Define $g(x) := x^q$ which is convex for $x \geq 0$. For a convex function $g: (0,\infty) \rightarrow (0,\infty)$, we know that $\forall x, \, y \in (0,\infty), \, \forall \theta \in [0,1]: \; g(\theta x + (1-\theta)y) \leq \theta g(x) + (1-\theta)g(y)$. Hence, with $\theta=\frac{1}{2}$ and $g(x)=x^q$, $x\geq 0$, we get
 \begin{align*}
& \quad \Big(-y \log(\alpha) - (x-y) \log(1-\alpha) - \log(G(0))\Big)^q \\
&= \Bigg( \theta \Big(2 \left(  -y \log(\alpha) - (x-y) \log(1-\alpha) \right) \Big) + \left(1-\theta\right)\Big(2\left( - \log(G(0))\right)\Big)  \Bigg)^q \\
& \leq \theta \Big( 2 (  -y \log(\alpha) - (x-y) \log(1-\alpha)  )   \Big)^q + \left( 1- \theta \right) \Big( 2 (- \log(G(0)))  \Big)^q \\
&= \frac{1}{2} \Big( 2 (  -y \log(\alpha) - (x-y) \log(1-\alpha)  )   \Big)^q + \frac{1}{2} \Big( 2 (- \log(G(0)))  \Big)^q.
\end{align*}
Similarly, we get
\begin{align*}
\Big( 2 (  -y \log(\alpha) - (x-y) \log(1-\alpha)  )   \Big)^q \leq \frac{1}{2} \Big(4(-y \log(\alpha))\Big)^q + \frac{1}{2}\Big(4(-(x-y) \log(1-\alpha))\Big)^q,
\end{align*}
which altogether leads to
 \begin{align*}
& \quad \Big(-y \log(\alpha) - (x-y) \log(1-\alpha) - \log(G(0))\Big)^q \\
& \leq \frac{1}{4}\Big(4(-y \log(\alpha))\Big)^q + \frac{1}{4}\Big(4(-(x-y) \log(1-\alpha))\Big)^q + \frac{1}{2}\Big(2(-\log(G(0)))\Big)^q \\
& = 4^{q-1}\Big(-y \log(\alpha)\Big)^q + 4^{q-1}\Big(-(x-y) \log(1-\alpha)\Big)^q +2^{q-1}\Big(-\log(G(0))\Big)^q.
\end{align*}
In total, we have 
\begin{align*}
I  \leq& P \Biggl( \sum\limits_{x,y \in \mathbb{N}_0, y \leq x} \Bigg[4^{q-1} y^q (-\log(\alpha))^q + 4^{q-1}(x-y)^q (-\log(1-\alpha))^q + 2^{q-1} (-\log(G(0)))^q\Bigg] \\
 & \qquad \qquad \qquad \qquad  P^*( (X_0^*,X_1^*)=(x,y)  ) \geq k/2  \Biggr) \\
 \leq&  P \left(  \sum\limits_{x,y \in \mathbb{N}_0, y \leq x} 4^{q-1} y^q (-\log(\alpha))^q P^*( (X_0^*,X_1^*)=(x,y)  )   \geq k/6  \right) \\
&+  P \left(  \sum\limits_{x,y \in \mathbb{N}_0, y \leq x} 4^{q-1}(x-y)^q (-\log(1-\alpha))^q P^*( (X_0^*,X_1^*)=(x,y)  )   \geq k/6  \right) \\
&+  P \left(  \sum\limits_{x,y \in \mathbb{N}_0, y \leq x} 2^{q-1} (-\log(G(0)))^q P^*( (X_0^*,X_1^*)=(x,y)  )   \geq k/6  \right) \\
 =:& I_1+I_2+I_3,
\end{align*} 
where
\begin{align*}
I_1 &= P \left( 4^{q-1}(- \log(\alpha))^q  \sum\limits_{x,y \in \mathbb{N}_0, y \leq x} y^q  P^*( (X_0^*,X_1^*)=(x,y)  ) \geq k/6  \right) \\
& \leq P \left( 4^{q-1}(- \log(\alpha))^q  \sum\limits_{x,y \in \mathbb{N}_0} y^q  P^*( (X_0^*,X_1^*)=(x,y)  ) \geq k/6  \right) \\
&= P \Big( 4^{q-1}(- \log(\alpha))^q \text{E}^*((X_1^*)^q) \geq k/6   \Big),    \\
I_2 &= P \left(   4^{q-1}(- \log(1-\alpha))^q \sum\limits_{x,y \in \mathbb{N}_0, y \leq x} (x-y)^q  P^*( (X_0^*,X_1^*)=(x,y)  ) \geq k/6  \right) \\
& \leq P \left( 4^{q-1}(- \log(1-\alpha))^q  \sum\limits_{x,y \in \mathbb{N}_0} |x-y|^q P^*( (X_0^*,X_1^*)=(x,y)  ) \geq k/6  \right)\\
&= P \Big(  4^{q-1}(- \log(1-\alpha))^q  \text{E}^*(  |X_0^*-X_1^*|^q   ) \geq k/6   \Big),
\end{align*} 
and
\begin{align*}
I_3 &= P \left( 2^{q-1}(-\log(G(0)))^q \sum\limits_{x,y \in \mathbb{N}_0, y \leq x}   P^*( (X_0^*,X_1^*)=(x,y)  ) \geq k/6  \right) \\
& \leq P \left( 2^{q-1}(-\log(G(0)))^q \sum\limits_{x,y \in \mathbb{N}_0}   P^*( (X_0^*,X_1^*)=(x,y)  ) \geq k/6  \right) \\
& = P \Big( 2^{q-1}(-\log(G(0)))^q \geq k/6  \Big) = \mathds{1}_{ \{ 2^{q-1}(-\log(G(0)))^q\geq k/6\}}.
\end{align*}
Altogether, for the first term $I$, we have
\begin{align*}
I & \leq P \left( 4^{q-1}(- \log(\alpha))^q E^*((X_1^*)^q) \geq k/6   \right) + P \left(  4^{q-1}(- \log(1-\alpha))^q  \text{E}^*(  |X_0^*-X_1^*|^q   ) \geq k/6   \right) \\ & \qquad + \mathds{1}_{ \{ 2^{q-1}(-\log(G(0)))^q\geq k/6\}}.   
 \end{align*}
Now, let us consider the second term $II$. We use similar arguments as employed above for deriving the bound for $I$. But instead of using the last summand (for $j=y$), we now use the first one (for $j=0$) to construct an upper bound for $II$. Hence, we get
\begin{align*}
\left|  \log \left( \sum\limits_{j=0}^x \binom{x}{j} \alpha^j(1-\alpha)^{x-j}G(y-j) \right) \right|^q  &= \left( - \log \left( \sum\limits_{j=0}^x \binom{x}{j} \alpha^j(1-\alpha)^{x-j}G(y-j) \right) \right)^q \\
& \leq \Big(-\log((1-\alpha)^x G(y))\Big)^q \\
&= \Big(-x \log(1-\alpha) - \log(G(y))\Big)^q. 
\end{align*}
Using Assumption \ref{asum_mix_gu}, we have $G(y)\geq c_1 e^{-c_2 y}$ such that $- \log(G(y))\leq -\log(c_1)+c_2y$. Altogether, by using similar convexity arguments as used for term $I$, we get 
\begin{align*}
\left|  \log \left( \sum\limits_{j=0}^x \binom{x}{j} \alpha^j(1-\alpha)^{x-j}G(y-j) \right) \right|^q \leq 4^{q-1}\Big(-x \log(\alpha)\Big)^q + 4^{q-1} \Big(- \log(c_1)\Big)^q + 2^{q-1}\Big(c_2 y\Big)^q.
\end{align*}
Analogously to the steps conducted for $I$, we get
\begin{align*}
II &\leq  P \Biggl( \sum\limits_{x,y \in \mathbb{N}_0, y > x} \Bigg[4^{q-1}(-x \log(\alpha))^q + 4^{q-1} (- \log(c_1))^q + 2^{q-1}(c_2 y)^q \Bigg] P^*( (X_0^*,X_1^*)=(x,y)  ) \geq k/2  \Biggr)\\
& \leq P \left( 4^{q-1} (-\log(\alpha))^q \text{E}^*((X_0^*)^q) \geq k/6   \right) + \mathds{1}_{\{(4^{q-1}(-\log(c_1))^q \geq k/6\}} + P \left( 2^{q-1} c_2 \text{E}^*((X_1^*)^q) \geq k/6   \right).
\end{align*}
In summary, we showed 
\begin{align*}
& P \left(  \sum\limits_{(x,y) \in \mathbb{N}_0^2} \left|  \log \left( \sum\limits_{j=0}^{\min(x,y)} \binom{x}{j} \alpha^j(1-\alpha)^{x-j}G(y-j) \right) \right|^q P^*( (X_0^*,X_1^*)=(x,y)  ) \geq k   \right) \\
& \leq  P \left( 4^{q-1}(- \log(\alpha))^q \text{E}^*((X_1^*)^q) \geq k/6   \right) 
+ P \left(  4^{q-1}(- \log(1-\alpha))^q  \text{E}^*(  |X_0^*-X_1^*|^q   ) \geq k/6   \right) \\
& \quad + \mathds{1}_{\{ 2^{q-1}(-\log(G(0)))^q \geq k/6\}} +  P \left( 4^{q-1} (-\log(\alpha))^q E^*((X_0^*)^q) \geq k/6   \right) \\
& \quad + \mathds{1}_{\{(4^{q-1}(-\log(c_1))^q \geq k/6\}}  + P \left( 2^{q-1} c_2 \text{E}^*((X_1^*)^q) \geq k/6   \right),
\end{align*}
where $c_1 \in (0,1]$ and $c_2 \in \mathbb{R}_+$. Due to the boundedness of the moments of the original observations, we conclude that the moments of the bootstrap innovations are also bounded with probability tending to one \citep[see Lemma B.1 in][]{jewe}. Using the Minkowski inequality, the same holds for their differences. With the previous arguments, the last right-hand side above becomes arbitrarily small for $k$ chosen sufficiently large. Hence, \eqref{eq:qr} follows and, consequently, the covariances in \eqref{cov_def2} are summable. Hence, the right-hand side in \eqref{cov_def2} is of order $O_P(\frac{1}{n})$ and vanishing in probability as $n\rightarrow \infty$.
In summary, we have shown $\text{Var}^*(M_n^*) \overset{p}{\rightarrow} 0$ as $n \rightarrow \infty$.

As shown by  \citet{drost} in the proof of Theorem \ref{cons}, for fixed $x_{-p}, \ldots, x_0 \in \mathbb{N}_0$, the map $\overline{E} \ni (\fett\alpha, G) \mapsto m^{\fett\alpha, G}(x_{-p}, \ldots, x_0)$ is continuous because there appear only a finite number of $G(j)$'s in $P^{\fett\alpha, G}_{(x_{-1}, \ldots, x_{-p}), x_0} $. Moreover, for all $x_{-p}, \ldots, x_0 \in \mathbb{N}_0$, we have $ m^{\fett\alpha, G}(x_{-p}, \ldots, x_0) \leq \log(1) = 0$. Additionally, for the bootNPMLE, we have $M_n^*(\widehat{\fett\alpha}_n^*, \widehat{G}_n^*) \geq M_n^*(\widehat{\fett\alpha}_n, \widehat{G}_n)$, because $(\widehat{\fett\alpha}_n^*, \widehat{G}_n^*)$ maximizes the likelihood by construction. Moreover, we need to show that the map $\overline{E} \ni (\fett\alpha, G) \mapsto M_n(\fett\alpha, G)$ has a unique maximum at ($\widehat{\fett\alpha}_n, \widehat{G}_n$) with probability tending to one. Since, in analogy to the identification argument used in  Section B.2, part (C) of \textcite{drost2009appendix}, with probability tending to one, we have the identification property (for general $p\in\mathbb{N}_0$), that is, 
\begin{align}
& P^{\fett\alpha, G}_{(X^*_{-1}, \ldots, X^*_{-p}), X_0^*}= P^{\widehat{\fett\alpha}_n, \widehat{G}_n}_{(X^*_{-1}, \ldots, X^*_{-p}), X_0^*} \quad\mathds{P}_{\nu_{\widehat{\fett\alpha}_n, \widehat{G}_n},\widehat{\fett\alpha}_n,\widehat{G}_n}\text{-a.s. (in $\mathds{P}_{\nu_{\fett\alpha_0, G_0},\fett\alpha_0,G_0}$-probability)} \notag \\
&\Rightarrow (\fett\alpha, G) = (\widehat{\fett\alpha}_n, \widehat{G}_n)\quad\mathds{P}_{\nu_{\fett\alpha_0, G_0},\fett\alpha_0,G_0}\text{-a.s.},   \label{identification_property_boot}
\end{align}
the assertion follows from 
\begingroup
\allowdisplaybreaks
\begin{align*}
& \quad \, \, \text{E}^*(M_n^*(\fett\alpha, G))-\text{E}^*(M_n^*(\widehat{\fett\alpha}_n, \widehat{G}_n) \\ &= \text{E}^* \left( \frac{1}{n} \sum\limits_{t=0}^n m^{\fett\alpha, G}(X^*_{t-p}, \ldots, X_t^*)  \right) - \text{E}^* \left( \frac{1}{n} \sum\limits_{t=0}^n m^{\widehat{\fett\alpha}_n, \widehat{G}_n}(X^*_{t-p}, \ldots, X_t^*)  \right) \\
  &=  \frac{1}{n} \sum\limits_{t=0}^n  \text{E}^* \left( \log P^{\fett\alpha, G}_{(X^*_{t-1}, \ldots, X_{t-p}^*),X_t^*}  \right) - \frac{1}{n} \sum\limits_{t=0}^n \text{E}^* \left(  \log P^{\widehat{\fett\alpha}_n, \widehat{G}_n}_{(X^*_{t-1}, \ldots, X_{t-p}^*),X_t^*}  \right) \\
   &=  \frac{1}{n} \, n \,  \text{E}^* \left( \log P^{\fett\alpha, G}_{(X^*_{-1}, \ldots, X_{-p}^*),X_0^*}  \right) - \frac{1}{n} \, n \, \text{E}^* \left(  \log P^{\widehat{\fett\alpha}_n, \widehat{G}_n}_{(X^*_{-1}, \ldots, X_{-p}^*),X_0^*}  \right) \\
    &\leq 2 \, \text{E}^* \left( \sqrt{\frac{P^{\fett\alpha, G}_{(X^*_{-1}, \ldots, X_{-p}^*),X_0^*} }{P^{\widehat{\fett\alpha}_n, \widehat{G}_n}_{(X^*_{-1}, \ldots, X_{-p}^*),X_0^*}}}  -1 \right) \\
     &= 2 \sum\limits_{y \in \mathbb{N}_0^p} \nu_{\widehat{\fett\alpha}_n, \widehat G} \{ y \} \sum\limits_{x_0 =0}^\infty \sqrt{P_{y, x_0}^{\fett\alpha, G} P_{y,x_0}^{\widehat{\fett\alpha}_n, \widehat{G}_n}} -2 \\ 
     & \leq - \sum\limits_{y \in \mathbb{N}_0^p} \nu_{\widehat{\fett\alpha}_n, \widehat G} \{ y \} \sum\limits_{x_0 =0}^\infty \left( \sqrt{P_{y, x_0}^{\fett\alpha, G}} - \sqrt{P_{y,x_0}^{\widehat{\fett\alpha}_n, \widehat{G}_n}}   \right)^2 \\
      &\leq 0   \quad   \text{in }\mathds{P}_{\nu_{\fett\alpha_0, G_0},\fett\alpha_0,G_0}\text{-probability},
\end{align*}
\endgroup
where we used $\log(x) \leq 2(\sqrt{x}-1)$ for $x \geq 0$ and the (conditional) stationarity of the bootstrap process. Note that the last inequality above holds in $\mathds{P}_{\nu_{\fett\alpha_0, G_0},\fett\alpha_0,G_0}$-probability due to \eqref{identification_property_boot}.

In summary, we showed that all conditions of Wald's consistency theorem \citep{wald} hold (in probability) and we obtain $d((\widehat{\fett\alpha}_n^*, \widehat{G}_n^*),(\widehat{\fett\alpha}_n, \widehat{G}_n)) \overset{p^*}{\rightarrow} 0$ in $\mathds{P}_{\nu_{\fett\alpha_0,G_0}, \fett\alpha_0, G_0}$-probability, which immediately implies $\widehat{\fett\alpha}_n^* - \widehat{\fett\alpha}_n \overset{p^*}{\rightarrow} 0$ in $\mathds{P}_{\nu_{\fett\alpha_0,G_0}, \fett\alpha_0, G_0}$-probability as well as, for all $k \in \mathbb{N}_0$, $\widehat{G}_n^*(k) - \widehat{G}_n(k) \overset{p^*}{\rightarrow} 0$ in $\mathds{P}_{\nu_{\fett\alpha_0,G_0}, \fett\alpha_0, G_0}$-probability. \qed

\bigskip




\begin{Lemma}[Fréchet derivative of $\Psi_n$] \label{frechet}
Let Assumption \ref{ass1} hold true. For fixed $n$ and for $\theta_0=(\fett\alpha_0,G_0)\in\Theta\times \mathcal{G}$, the finite sample moment equations \eqref{eq:psi_n1} and \eqref{eq:psi_n2}, i.e., the maps $\Psi_n:(0,1)^p \times \widetilde{\mathcal{G}} \rightarrow \mathbb{R}^p \times \ell^\infty(\mathcal{H}_1)$, are Fréchet differentiable with derivative $\dot\Psi^{\theta_0}_n : lin([0,1]^p \times \widetilde{\mathcal{G}}) \rightarrow \mathbb{R}^p \times \ell^\infty(\mathcal{H}_1)$ at $\theta_0$ given by
\begin{align}
\dot \Psi^{\theta_0}_n(\theta-\theta_0) = \begin{pmatrix}
  \dot \Psi^{\theta_0}_{n11}(\fett\alpha-\fett\alpha_0) + \dot \Psi^{\theta_0}_{n12}(G-G_0) \\
   \dot \Psi^{\theta_0}_{n21}(\fett\alpha-\fett\alpha_0) + \dot \Psi^{\theta_0}_{n22}(G-G_0)
 \end{pmatrix}
\end{align}
with $\dot \Psi^{\theta_0}_{n11}: \mathbb{R}^p \rightarrow \mathbb{R}^p, \dot \Psi^{\theta_0}_{n12}: lin(\mathcal{G}) \rightarrow \mathbb{R}^p, \dot \Psi^{\theta_0}_{n21}: \mathbb{R}^p \rightarrow \ell^\infty(\mathcal{H}_1)$ and $\dot \Psi^{\theta_0}_{n22}: lin(\mathcal{G}) \rightarrow \ell^\infty(\mathcal{H}_1)$ defined in \eqref{psi_dot_n11} - \eqref{psi_dot_n22}, where $lin(\cdot)$ denotes the linear span. That is, for fixed $n$, we have
\begin{align}
\|\Psi_n(\theta)-\Psi_n(\theta_0)-\dot \Psi^{\theta_0}_n(\theta-\theta_0)\| = o_P\left(\|\theta-\theta_0\|\right) 
\end{align}
and 
\begin{align}
\Psi_n(\theta) = \Psi_n(\theta_0) + \dot \Psi^{\xi}_n(\theta-\theta_0)
\end{align}
for some $\xi$, where $||\xi - \theta_0|| \leq ||\theta - \theta_0||$, 
where $\dot\Psi^{\theta_0}_n : lin([0,1]^p \times \widetilde{\mathcal{G}}) \rightarrow \mathbb{R}^p \times \ell^\infty(\mathcal{H}_1)$ is a continuous, linear mapping. Furthermore, as $\widehat \theta_n=(\widehat{\fett\alpha}_n,\widehat G_n)\in\Theta\times \mathcal{G}$ holds with $P$-probability tending to one as well as $\widehat \theta_n^*=(\widehat{\fett\alpha}_n^*,\widehat G_n^*)\in\Theta\times \mathcal{G}$ with $P^*$-probability tending to one (conditional on the data), by plugging-in $\widehat \theta_n$ for $\theta_0$ and $\widehat \theta_n^*$ for $\theta$ in the above, we have 
\begin{align}
\|\Psi_n(\widehat \theta_n^*)-\Psi_n(\widehat \theta_n)-\dot \Psi^{\widehat \theta_n}_n(\widehat \theta_n^*-\widehat \theta_n)\| = o_{P^*}\left(\|\widehat \theta_n^*-\widehat \theta_n\|\right) 
\end{align}
and
\begin{align} \label{taylor_lemma}
\Psi_n(\widehat \theta_n^*)=\Psi_n(\widehat \theta_n)+\dot \Psi^{\xi_n}_n(\widehat \theta_n^*-\widehat \theta_n)
\end{align}
for some $\xi_n$, where $||\xi_n - \thn|| \leq ||\thns - \thn||$.
\end{Lemma}

\begin{proof}  
For notational convenience, we consider only the case of $p=1$, but the following arguments can be extended to higher order $p>1$. Let $n$ be fixed, $\theta=(\alpha,G)=(\alpha, G(0), G(1), \ldots)$, $\theta_0=(\alpha_0,G_0)=(\alpha_0, G_0(0), G_0(1), \ldots)$ and let $\Psi_n=(\Psi_{n1}, \Psi_{n2})$, where $\Psi_{n1}$ and $\Psi_{n2}$ are defined according to \eqref{eq:psi_n1} and \eqref{eq:psi_n2}. Following Lemma \ref{lem_frechet_derivation}, its Fréchet derivative $\dot \Psi_n^{\theta_0}$ is defined as
\begin{align} \label{eq:dot_psi_n_comp}
\dot \Psi_n^{\theta_0}(\theta-\theta_0) 
&= \begin{pmatrix}
\frac{\partial \Psi_{n1}(\alpha,G)}{\partial \alpha}_{|_{\theta=\theta_0}}(\alpha- \alpha_0) + \sum\limits_{k=0}^\infty \frac{\partial \Psi_{n1}(\alpha,G)}{\partial G(k)}_{|_{\theta=\theta_0}} (G(k) - G_0(k)) \\
\frac{\partial \Psi_{n2}(\alpha,G)}{\partial \alpha}_{|_{\theta=\theta_0}} (\alpha- \alpha_0) + \sum\limits_{k=0}^\infty \frac{\partial \Psi_{n2}(\alpha,G)}{\partial G(k)}_{|_{\theta=\theta_0}} (G(k) - G_0(k))
\end{pmatrix} \notag \\
&=: \begin{pmatrix}
\dot \Psi_{n11}^{\theta_0}(\alpha-\alpha_0) + \dot \Psi_{n12}^{\theta_0}(G-G_0) \\
\dot \Psi_{n21}^{\theta_0}(\alpha-\alpha_0) + \dot \Psi_{n22}^{\theta_0}(G-G_0)
\end{pmatrix},
\end{align}
where 
\begin{align}
\dot \Psi_{n11}^{\theta_0}(\alpha-\alpha_0) &= \left(\frac{1}{n} \sum\limits_{t=0}^n \frac{\frac{\partial^2}{\partial \alpha^2} P_{X_{t-1},X_t}^{\alpha,G}}{P_{X_{t-1},X_t}^{\alpha,G}}
- \frac{1}{n} \sum\limits_{t=0}^n \dot{l}_\alpha^2(X_{t-1},X_t;\alpha,G)\right)_{|_{\theta=\theta_0}} (\alpha-\alpha_0) \label{psi_dot_n11},    \\
\dot \Psi_{n12}^{\theta_0}(G-G_0) &= \sum\limits_{k=0}^\infty \Bigg(\frac{1}{n} \sum\limits_{t=0}^n  \frac{\frac{\partial}{\partial G(k)} \frac{\partial}{\partial \alpha} P_{X_{t-1},X_t}^{\alpha,G} }{P_{X_{t-1},X_t}^{\alpha,G}}  \label{psi_dot_n12}\\
& \quad
- \frac{1}{n} \sum\limits_{t=0}^n  \dot{l}_\alpha(X_{t-1},X_t;\alpha,G) \frac{\frac{\partial}{\partial G(k)}P_{X_{t-1},X_t}^{\alpha,G}}{P_{X_{t-1},X_t}^{\alpha,G}}\Bigg)_{|_{\theta=\theta_0}}
(G(k) - G_0(k)) \notag
\end{align}
and, for $h \in \mathcal{H}_1$,
\begin{align}
& \dot \Psi_{n21}^{\theta_0}(\alpha-\alpha_0)h   \label{psi_dot_n21}  \\
&=
\Bigg(\frac{1}{n} \sum\limits_{t=0}^n   \sum\limits_{j=0}^\infty h(j) \frac{\Big(\frac{\partial}{\partial \alpha} P^{\alpha,G}(\varepsilon_t =j, X_t=x_t|X_{t-1}=x_{t-1})\Big)_{|_{(x_t,x_{t-1})=(X_t,X_{t-1})}}}{P^{\alpha,G}_{X_{t-1},X_t}} \notag \\
& \quad -\frac{1}{n} \sum\limits_{t=0}^n   
\dot{l}_\alpha(X_{t-1},X_t;\alpha,G) A_{\alpha,G}h(X_{t-1},X_t)\Bigg)_{|_{\theta=\theta_0}}(\alpha- \alpha_0), \notag 
\end{align}
and
\begin{align}
& \dot \Psi_{n22}^{\theta_0}(G-G_0)h \label{psi_dot_n22} \\
&= \sum\limits_{k=0}^\infty \Bigg(\frac{1}{n} \sum\limits_{t=0}^n  \sum \limits_{j=0}^\infty h(j) \frac{\Big(\frac{\partial}{\partial G(k)} P^{\alpha,G}(\varepsilon_t=j, X_t=x_t | X_{t-1}=x_{t-1})\Big)_{|_{(x_t,x_{t-1})=(X_t,X_{t-1})}}}{P^{\alpha,G}_{X_{t-1},X_t}} -h(k)  \notag  \\
& \quad- \frac{1}{n} \sum\limits_{t=0}^n A_{\alpha,G}h(X_{t-1},X_t) \frac{\Big(\frac{\partial}{\partial G(k)} P^{\alpha,G}(X_t=x_t | X_{t-1}=x_{t-1})\Big)_{|_{(x_t,x_{t-1})=(X_t,X_{t-1})}}}{P^{\alpha,G}_{X_{t-1},X_t})}\Bigg)_{|_{\theta=\theta_0}}(G(k) - \widehat{G}_n(k)),   \notag
\end{align}
where $\dot l_\alpha$ and $A_{\alpha,G}h$ are defined in \eqref{dot_l} and \eqref{Ah}.
Using similar arguments as in \citet{drost} to prove their Lemma 2(a) in  \citet{drost2009appendix}, which makes also heavy use of inequalities derived in \citet{drost2}, the assertion 
\begin{align*}
\|\Psi_n(\theta)-\Psi_n(\theta_0)-\dot \Psi^{\theta_0}_n(\theta-\theta_0)\| = o_P\left(\|\theta-\theta_0\|\right).  
\end{align*}
follows. 

Similarly, when plugging-in $\widehat \theta_n^*$ and $\widehat \theta_n$ for $\theta$ and $\theta_0$, respectively, making use of Theorem \ref{cons_bs_est}, the same arguments lead to
\begin{align*}
\|\Psi_n(\widehat \theta_n^*)-\Psi_n(\widehat \theta_n)-\dot \Psi^{\widehat \theta_n}_n(\widehat \theta_n^*-\widehat \theta_n)\| = o_{P^*}\left(\|\widehat \theta_n^*-\widehat \theta_n\|\right). 
\end{align*}
\end{proof}

\begin{Lemma} \label{unif_zws}
Suppose Assumptions \ref{ass1} and \ref{asum_mix_gu} hold. Let $\text{E}(X_1^k) < \infty$ for some $k>2(p+4)$ and $\text{E}((X_t^3(1+\rho)^{X_t}))^{1+\delta}<\infty$ for some $\rho, \delta>0$.  
Then, we have
\begin{align*}
\sqrt{n}\left(\dot\Psi^{\xi_n}_n(\thns-\thn)\right) - \sqrt{n}\left(\dot \Psi^{\theta_0}(\thns-\thn)\right) = o_{p^*}(1).
\end{align*}
Recall that $\dot\Psi^{\xi_n}_n : lin([0,1]^p \times \widetilde{\mathcal{G}}) \rightarrow \mathbb{R}^p \times \ell^\infty(\mathcal{H}_1)$ and $\dot\Psi^{\theta_0} : lin([0,1]^p \times \widetilde{\mathcal{G}}) \rightarrow \mathbb{R}^p \times \ell^\infty(\mathcal{H}_1)$ are the Fr\'echet derivatives of $\Psi_n$ at $\xi_n$ and of $\Psi$ at $\theta_0$, respectively.
\end{Lemma}

\begin{proof} 
By adding and subtracting the Fr\'echet derivative of $\Psi_n$ at $\theta_0$, we get
\begin{align}
& || \sqrt{n}\Big(\dot\Psi^{\xi_n}_n(\thns-\thn)\Big) - 
\sqrt{n}\Big(\dot \Psi^{\theta_0}(\thns-\thn)\Big) || \notag\\
&= ||\dot\Psi^{\xi_n}_n(\sqrt{n}(\thns-\thn))  \mp \dot \Psi^{\theta_0}_n(\sqrt{n}(\thns-\thn))  - \dot \Psi^{\theta_0}(\sqrt{n}(\thns-\thn))||    \label{eq:anfang} \\
&\leq ||\dot\Psi^{\xi_n}_n(\sqrt{n}(\thns-\thn)) -  \dot \Psi^{\theta_0}_n(\sqrt{n}(\thns-\thn))|| + ||\dot \Psi^{\theta_0}_n(\sqrt{n}(\thns-\thn))  - \dot \Psi^{\theta_0}(\sqrt{n}(\thns-\thn))|| \notag.
\end{align} 
We examine both resulting differences separately. 
Let $\varepsilon, \delta >0$ and consider the first term on the last right-hand side of \eqref{eq:anfang}. Then, we have
\begin{align*}
    & \quad  P( || \dot\Psi^{\xi_n}_n(\sqrt{n}(\thns-\thn))  - \dot \Psi^{\theta_0}_n(\sqrt{n}(\thns-\thn)) ||  > \varepsilon) \\
    &= P( || \dot\Psi^{\xi_n}_n(\sqrt{n}(\thns-\thn))  - \dot \Psi^{\theta_0}_n(\sqrt{n}(\thns-\thn)) ||  > \varepsilon, ||\xi_n - \tz || \leq \delta ) \\ & \quad + P(|| \dot\Psi^{\xi_n}_n(\sqrt{n}(\thns-\thn))  - \dot \Psi^{\theta_0}_n(\sqrt{n}(\thns-\thn)) ||  > \varepsilon, ||\xi_n - \tz || > \delta) \\
    & \leq P( || \dot\Psi^{\xi_n}_n(\sqrt{n}(\thns-\thn))  - \dot \Psi^{\theta_0}_n(\sqrt{n}(\thns-\thn)) ||  > \varepsilon, ||\xi_n - \tz || \leq \delta)  + P(||\xi_n - \tz || > \delta),
\end{align*}
where the last probability converges to zero for increasing $n$, because we have $||\xi_n - \thn|| \leq ||\thns - \thn|| \overset{p^*}{\rightarrow} 0$ in probability due to Theorem \ref{cons_bs_est} as well as $\|\thn - \tz\|\overset{p}{\rightarrow} 0$ due to Theorem \ref{cons}, which altogether gives $||\xi_n - \theta_0||\overset{p^*}{\rightarrow} 0$ in probability.
Hence, it remains to investigate the asymptotic behavior of
\[ P( || \dot\Psi^{\xi_n}_n(\sqrt{n}(\thns-\thn))  - \dot \Psi^{\theta_0}_n(\sqrt{n}(\thns-\thn)) ||  > \varepsilon, ||\xi_n - \tz || \leq \delta). \] 
For convenience, we will again consider only the case $p=1$, but the same arguments can be used to prove the results for higher model order. Let us consider the first summand of the first entry of \eqref{eq:dot_psi_n_comp} and follow the calculations of the derivatives in the proof of Lemma \ref{frechet}. Further, let $\xi_n  = (\alpha_{\xi_n}, G_{\xi_n})$ and $\tz=(\alpha_0,G_0)$. Then, we have
\begin{align} \label{eq:I+II}
    & |\dot \Psi_{n11}^{\xi_n}(\sqrt{n}(\ahns - \ahn)) - \dot \Psi_{n11}^{\tz} (\sqrt{n}(\ahns-\ahn))| \notag \\
    \leq & \left| \frac{1}{n} \sum\limits_{t=0}^n \left( 
    \frac{\frac{\partial^2}{\partial \alpha^2} P_{X_{t-1},X_t |\alpha=\ax}^{\alpha,\gx}}{P_{X_{t-1},X_t}^{\xi_n}} 
    - \frac{\frac{\partial^2}{\partial \alpha^2} P_{X_{t-1},X_t |\alpha=\alpha_0}^{\alpha,\gz}}{P_{X_{t-1},X_t}^{\tz}}
    - \dot{l}_\alpha^2(X_{t-1},X_t;\xi_n)
    +\dot{l}_\alpha^2(X_{t-1},X_t;\tz)
    \right) \right| \notag \\
    & \qquad \times|\sqrt{n}(\ahns-\ahn)| \notag \\
    \leq & \Bigg\{\left| \frac{1}{n} \sum\limits_{t=0}^n \left( 
    \frac{\frac{\partial^2}{\partial \alpha^2} P_{X_{t-1},X_t |\alpha=\ax}^{\alpha,\gx}}{P_{X_{t-1},X_t}^{\xi_n}} 
    - \frac{\frac{\partial^2}{\partial \alpha^2} P_{X_{t-1},X_t |\alpha=\alpha_0}^{\alpha,\gz}}{P_{X_{t-1},X_t}^{\tz}} \right)  \right| \notag \\ 
    & + \left| \frac{1}{n} \sum\limits_{t=0}^n \bigg(\dot{l}_\alpha^2(X_{t-1},X_t;\tz) - \dot{l}_\alpha^2(X_{t-1},X_t;\xi_n)\bigg) \right| \Bigg\}O_{P^*}(1)  \notag \\
    =: & \big(I + II\big)O_{P^*}(1).
\end{align}
Continuing with the second term $II$, using a binomial formula, it can be bounded by
\begin{align} \label{eq:dotl_prod}
     \frac{1}{n} \sum\limits_{t=0}^n \left| \dot{l}_\alpha(X_{t-1},X_t;\tz) + \dot{l}_\alpha(X_{t-1},X_t;\xi_n)  \right| \left| \dot{l}_\alpha(X_{t-1},X_t;\tz) - \dot{l}_\alpha(X_{t-1},X_t;\xi_n) \right|.
\end{align}
Hence, using $||\xi_n - \tz || \leq \delta$, according to Lemma \ref{Lipschitz_resultsII}, we get the bound
\begin{align} \label{eq:II_final}
    II & \leq  \frac{1}{n} \sum\limits_{t=0}^n \left| \dot{l}_\alpha(X_{t-1},X_t;\tz) + \dot{l}_\alpha(X_{t-1},X_t;\xi_n)  \right| \notag \\
    & \qquad \delta\left(\frac{2 X_{t-1} + C X_{t-1}^2\left(1+\rho\right)^{X_{t-1}-1}}{P^{\tz}_{X_{t-1}, X_{t}}} + \frac{2 X_{t-1} \left( C X_{t-1}\left(1+\rho\right)^{X_{t-1}-1}+ 1  \right)  }{P^{\tz}_{X_{t-1}, X_{t}} P^{\xi_n}_{X_{t-1}, X_{t}}}\right),
\end{align}
for some (generic) constant $C=C(\delta)$ and some $\rho=\rho(\delta)$, which becomes arbitrarily small for $\delta$ sufficiently small.
Similarly, from Lemma \ref{Lipschitz_resultsI}, we get 
\begin{align} \label{eq:I_final}
I & \leq \frac{1}{n} \sum\limits_{t=0}^n \delta\left(\frac{\widetilde C X_{t-1}^3(1+\widetilde\rho)^{X_{t-1}-1}+\widetilde C X_{t-1}^2}{P^{\xi_n}_{X_{t-1}, X_t}} +  \frac{\widetilde C X_{t-1}^2 \left( C  X_{t-1} (1+\rho)^{X_{t-1}-1} + 1 \right)}{P^{\xi_n}_{X_{t-1}, X_t} P^{\tz}_{X_{t-1}, X_t}}\right).      
\end{align}
With \eqref{eq:II_final} and \eqref{eq:I_final}, we can now finally tackle \eqref{eq:I+II} and get
\begin{align*}
    & P(|\dot \Psi_{n11}^{\xi_n}(\sqrt{n}(\ahns - \ahn)) - \dot \Psi_{n11}^{\tz} (\sqrt{n}(\ahns-\ahn))| > \varepsilon , ||\xi_n - \tz || \leq \delta) \\
    \leq & P \left(   \frac{1}{n} \sum\limits_{t=0}^n \delta\left(\frac{\widetilde C X_{t-1}^3(1+\widetilde\rho)^{X_{t-1}-1}+\widetilde C X_{t-1}^2}{P^{\xi_n}_{X_{t-1}, X_t}} +  \frac{\widetilde C X_{t-1}^2 \left( C  X_{t-1} (1+\rho)^{X_{t-1}-1} + 1 \right)}{P^{\xi_n}_{X_{t-1}, X_t} P^{\tz}_{X_{t-1}, X_t}}\right)\right.  \\
    & \left. +
    \frac{1}{n} \sum \limits_{t=0}^n |\dot{l}_\alpha(X_{t-1},X_t;\tz) + \dot{l}_\alpha(X_{t-1},X_t;\xi_n)| \right. \\
    & \left. \quad \delta\left(\frac{2 X_{t-1} + C X_{t-1}^2\left(1+\rho\right)^{X_{t-1}-1}}{P^{\tz}_{X_{t-1}, X_{t}}} + \frac{2 X_{t-1} \left( C X_{t-1}\left(1+\rho\right)^{X_{t-1}-1}+ 1  \right)  }{P^{\tz}_{X_{t-1}, X_{t}} P^{\xi_n}_{X_{t-1}, X_{t}}}\right) O_{P^*}(1) > \varepsilon  \right) \\
     \leq & P \left( \delta \left(\frac{1}{n} \sum\limits_{t=0}^n w(X_t, X_{t-1}; \rho)\right)  O_{P^*}(1) > \varepsilon  \right),
\end{align*}
where we used similar techniques to deal with the factor $|\dot{l}_\alpha(X_{t-1},X_t;\tz) + \dot{l}_\alpha(X_{t-1},X_t;\xi_n)|$ and with $P^{\xi_n}_{X_{t-1}, X_{t}}$ in the denominators.
Using Assumption \ref{asum_mix_gu} and Theorem 14.1 of \textcite{davidson}, we can conclude that $w(X_t, X_{t-1}; \rho)$ is also geometrically mixing. \textcite{wlln_geomix} gives a weak law of law large numbers under geometric mixing and together with the moment assumptions we made we can conclude that $\forall \varepsilon > 0 \; \forall \gamma > 0 \; \exists \delta > 0,\; \rho=\rho(\delta)>0: \;  P \left( \delta \frac{1}{n}\sum\limits_{t=0}^n w(X_t, X_{t-1}; \rho) > \varepsilon  \right) < \gamma$.

So far, we showed the convergence result only for the first term of the first entry in \eqref{eq:dot_psi_n_comp} by proving Lipschitz-type continuity properties of $\frac{\partial^2}{\partial \alpha^2} P^{\alpha, G}_{X_{t-1}, X_t}$ and $\left( \frac{\partial}{\partial \alpha} P^{\alpha, G}_{X_{t-1}, X_t}  \right)^2$. For the second term of that first entry in \eqref{eq:dot_psi_n_comp}, we can proceed analogously by showing with similar arguments the Lipschitz continuity of $\frac{\partial}{\partial G} \frac{\partial}{\partial \alpha}  P^{\alpha, G}_{X_{t-1}, X_t} $ and $\frac{\partial}{\partial \alpha}  P^{\alpha, G}_{X_{t-1}, X_t}\frac{\partial}{\partial G}  P^{\alpha, G}_{X_{t-1}, X_t}$ since 
\begin{align*}
& |\dot \Psi_{n12}^{\xi_n}\big(\sqrt{n}(\ghns-\ghn)\big) - \dot \Psi_{n12}^{\tz}\big(\sqrt{n}(\ghns-\ghn)\big)| \\
= & \left|  \sum\limits_{k=0}^\infty \frac{1}{n} \sum\limits_{t=0}^n \left(  \frac{\frac{\partial}{\partial G(k)} \frac{\partial}{\partial \alpha} P^{\alpha,G}_{X_{t-1},X_t|(\alpha,G) = \xi_n}}{P^{\xi_n}_{X_{t-1}, X_t}} - \dot l_\alpha(X_{t-1}, X_t; \xi_n) \frac{\frac{\partial}{\partial G(k)} P ^{\ax,G}_{X_{t-1}, X_t| G= \gx}}{P^{\xi_n}_{X_{t-1}, X_t}} \right)  \big(\sqrt{n}(\ghns(k) - \ghn(k))\big) \right. \\
& \left. - \frac{1}{n} \sum\limits_{t=0}^n \left(  \frac{\frac{\partial}{\partial G(k)} \frac{\partial}{\partial \alpha} P^{\alpha,G}_{X_{t-1},X_t|(\alpha,G) = \tz}}{P^{\tz}_{X_{t-1}, X_t}} - \dot l_\alpha(X_{t-1}, X_t; \tz) \frac{\frac{\partial}{\partial G(k)} P ^{\az,G}_{X_{t-1}, X_t| G= \gz}}{P^{\tz}_{X_{t-1}, X_t}} \right) \big(\sqrt{n}(\ghns(k) - \ghn(k))\big) \right| \\
\leq & \left(\sum\limits_{k=0}^\infty \sqrt{n}|\ghns(k) - \ghn(k)|\right) \sum\limits_{k=0}^\infty\left|  \frac{1}{n} \sum\limits_{t=0}^n  \left(   \frac{\frac{\partial}{\partial G(k)} \frac{\partial}{\partial \alpha} P^{\alpha,G}_{X_{t-1},X_t|(\alpha,G) = \xi_n}}{P^{\xi_n}_{X_{t-1}, X_t}} - \frac{\frac{\partial}{\partial G(k)} \frac{\partial}{\partial \alpha} P^{\alpha,G}_{X_{t-1},X_t|(\alpha,G) = \tz}}{P^{\tz}_{X_{t-1}, X_t}} \right. \right. \\
    & \left. \left. \qquad \qquad \qquad + \dot l_\alpha(X_{t-1}, X_t; \tz) \frac{\frac{\partial}{\partial G(k)} P ^{\az,G}_{X_{t-1}, X_t| G= \gz}}{P^{\tz}_{X_{t-1}, X_t}} - \dot l_\alpha(X_{t-1}, X_t; \xi_n) \frac{\frac{\partial}{\partial G(k)} P ^{\ax,G}_{X_{t-1}, X_t| G= \gx}}{P^{\xi_n}_{X_{t-1}, X_t}}  \right)\right| \\
    \leq & O_{P^*}(1)\sum\limits_{k=0}^\infty \left( \left| \frac{1}{n} \sum\limits_{t=0}^n \left(   \frac{\frac{\partial}{\partial G(k)} \frac{\partial}{\partial \alpha} P^{\alpha,G}_{X_{t-1},X_t|(\alpha,G) = \xi_n}}{P^{\xi_n}_{X_{t-1}, X_t}} - \frac{\frac{\partial}{\partial G(k)} \frac{\partial}{\partial \alpha} P^{\alpha,G}_{X_{t-1},X_t|(\alpha,G) = \tz}}{P^{\tz}_{X_{t-1}, X_t}} \right)   \right| \right. \\
    & \left. + \left|  \frac{1}{n} \sum\limits_{t=0}^n  \left( \dot l_\alpha(X_{t-1}, X_t; \tz) \frac{\frac{\partial}{\partial G(k)} P ^{\az,G}_{X_{t-1}, X_t| G= \gz}}{P^{\tz}_{X_{t-1}, X_t}} - \dot l_\alpha(X_{t-1}, X_t; \xi_n) \frac{\frac{\partial}{\partial G(k)} P ^{\ax,G}_{X_{t-1}, X_t| G= \gx}}{P^{\xi_n}_{X_{t-1}, X_t}}   \right)  \right|   \right).
\end{align*}
Then, using similar arguments as used for the first term of the first entry of \eqref{eq:dot_psi_n_comp}, we get the claimed result. We omit the details.

Finally, for both terms of the second entry of $\Psi_{n}^{\theta_0}$, we require some Lipschitz continuity of $\sum_{j=0}^\infty \frac{\partial}{\partial \alpha} P^{\alpha,G}(\varepsilon_t=j, X_t=x_t|X_{t-1}=x_{t-1})$, $\sum_{j=0}^\infty \frac{\partial}{\partial G} P^{\alpha,G}(\varepsilon_t=j, X_t=x_t|X_{t-1}=x_{t-1})$ and $A_{\alpha,G} h(X_{t-1}, X_t)$, respectively. Starting with $P^{\alpha,G}(\varepsilon_t=j, X_t=x_t|X_{t-1}=x_{t-1})$, we get the following closed form representation 
\begin{align} \label{closed_form}
    & \quad P^{\alpha,G}(\varepsilon_t=j,X_t=x_t | X_{t-1}=x_{t-1}) 
    = P^{\alpha,G}(\varepsilon_t=j,\alpha \circ X_{t-1} + \varepsilon_t=x_t | X_{t-1}=x_{t-1}) \notag \\
    &= P^{\alpha,G}(\varepsilon_t=j,\alpha \circ X_{t-1} + j=x_t | X_{t-1}=x_{t-1})
    = P^{\alpha,G}(\varepsilon_t=j,\alpha \circ X_{t-1} =x_t-j | X_{t-1}=x_{t-1}) \notag \\
    &= P^{\alpha,G}(\varepsilon_t=j |X_{t-1}=x_{t-1}) P^{\alpha,G}(\alpha \circ X_{t-1}=x_t-j |X_{t-1}=x_{t-1}) \notag \\
    &= P^{\alpha,G}(\varepsilon_t=j) P^{\alpha,G}(\alpha \circ x_{t-1}=x_t-j) \notag \\
    &= G(j) \binom{x_{t-1}}{x_t-j} \alpha^{x_t -j}(1-\alpha)^{x_{t-1}-(x_t-j)} \mathds{1}_{j \in \{ \max(0,x_t-x_{t-1}), \ldots, x_t\}}, 
\end{align}
where we used that $\epsilon_t$ and $\alpha\circ X_{t-1}$ are independent given $X_{t-1}$. 
Using the same arguments as above, we conclude Lipschitz continuity of both $\sum_{j=0}^\infty \frac{\partial}{\partial \alpha} P^{\alpha,G}(\varepsilon_t=j, X_t=x_t|X_{t-1}=x_{t-1})$ and $\sum_{j=0}^\infty \frac{\partial}{\partial G} P^{\alpha,G}(\varepsilon_t=j, X_t=x_t|X_{t-1}=x_{t-1})$. The same applies to argue Lipschitz continuity of $A_{\alpha,G}h(X_{t-1},X_t)$ since, with \eqref{eq: A_expanded_representation} and \eqref{closed_form}, we have
\begin{align}
    A_{\alpha,G} h(x_{t-1},x_t) 
    &= \sum\limits_{j=0}^\infty h(j) \frac{P^{\alpha,G}(\varepsilon_t=j,X_t=x_t|X_{t-1}=x_{t-1})}{P^{\alpha,G}(X_t=x_t|X_{t-1}=x_{t-1})} \notag \\
    &= \sum\limits_{j=0}^\infty h(j) \frac{G(j) \binom{x_{t-1}}{x_t-j} \alpha^{x_t-j}(1-\alpha)^{x_{t-1}-(x_t-j)} \mathds{1}_{j \in \{ \max(0,x_t-x_{t-1}), \ldots, x_t  \}  }}{\sum\limits_{s=0}^{\min(x_t,x_{t-1})} \binom{x_{t-1}}{s} \alpha^s(1-\alpha)^{x_{t-1}-s} G(x_t-s)} \notag \\
    &= \frac{\sum\limits_{s=0}^{\min(x_t,x_{t-1})} h(x_t-s) \binom{x_{t-1}}{s} \alpha^s(1-\alpha)^{x_{t-1}-s} G(x_t-s)}{\sum\limits_{s=0}^{\min(x_t,x_{t-1})} \binom{x_{t-1}}{s} \alpha^s(1-\alpha)^{x_{t-1}-s} G(x_t-s)},  \label{eq:Ah_explicit}
\end{align}
which can be treated analogously.

Coming back to \eqref{eq:anfang}, it remains to investigate also the asymptotic behavior of the second difference, i.e., of 
\begin{align} \label{eq:anfang2}
P( ||  \dot \Psi^{\theta_0}_n\big(\sqrt{n}(\thns-\thn)\big) - \dot \Psi^{\theta_0}\big(\sqrt{n}(\thns-\thn)\big)   ||  > \widetilde\varepsilon).
\end{align}
Here again, we argue componentwise and apply the idea of the proof of weak of large numbers. For the first summand of the first entry of \eqref{eq:dot_psi_n_comp}, we again have 
\begin{align*}
    P(|\dot \Psi_{n11}^{\tz}\big(\sqrt{n}(\ahns- \ahn)\big) - \dot \Psi_{11}^{\tz}\big(\sqrt{n}(\ahns- \ahn)\big) | > \widetilde\varepsilon) \leq P(|\dot \Psi_{n11}^{\tz} - \dot \Psi_{11}^{\tz} |O_{P^*}(1) > \widetilde\varepsilon)
\end{align*}
and 
\begin{align*}
   E_{\alpha_0, G_0}( \dot \Psi_{n11}^{\tz} ) &= E_{\alpha_0, G_0} \left(\frac{1}{n} \sum\limits_{t=0}^n \frac{\frac{\partial^2}{\partial \alpha^2} P_{X_{t-1},X_t}^{\alpha,G_0}}{P_{X_{t-1},X_t|\alpha=\alpha_0}^{\alpha,G}}
- \frac{1}{n} \sum\limits_{t=0}^n \dot{l}_\alpha^2(X_{t-1},X_t;\alpha_0,G_0) \right) \\
&=  E_{\alpha_0, G_0} \left(\frac{\frac{\partial^2}{\partial \alpha^2} P_{X_{t-1},X_t}^{\alpha,G_0}}{P_{X_{t-1},X_t|\alpha=\alpha_0}^{\alpha,G}} \right)
- E_{\alpha_0, G_0} \left( \dot{l}_\alpha^2(X_{t-1},X_t;\alpha_0,G_0) \right).  
\end{align*} 
The second component corresponds to the term \textcite{drost} derived for $\dot \Psi_{11}^{\theta_0}$ in case of $p=1$. 
For the first term, by dominated convergence, we have 
\begin{align*}
E_{\alpha_0, G_0}\left( \frac{\frac{\partial^2}{\partial \alpha^2} P_{X_{-1},X_0|\alpha=\alpha_0}^{\alpha,G_0}}{P_{X_{-1},X_0}^{\alpha_0,G_0}} \right) &= \sum\limits_{m,l=0}^\infty \frac{\frac{\partial^2}{\partial \alpha^2} P^{\alpha,G_0}(X_0=m| X_{-1}=l)_{|\alpha=\alpha_0}}{P^{\alpha_0,G_0}(X_0=m| X_{-1}=l)} P^{\alpha_0,G_0}(X_0=m, X_{-1}=l) \\
&= \sum\limits_{m,l=0}^\infty \frac{\partial^2}{\partial \alpha^2} P^{\alpha,G_0}(X_0=m| X_{-1}=l)_{|\alpha=\alpha_0} P^{\alpha_0,G_0}(X_{-1}=l) \\
&= \sum\limits_{l=0}^\infty P^{\alpha_0,G_0}(X_{-1}=l) \sum\limits_{m=0}^\infty \frac{\partial^2}{\partial \alpha^2} P^{\alpha,G_0}(X_0=m| X_{-1}=l)_{|\alpha=\alpha_0} \\
&= \sum\limits_{l=0}^\infty P^{\alpha_0,G_0}(X_{-1}=l)  \frac{\partial^2}{\partial \alpha^2} \sum\limits_{m=0}^\infty P^{\alpha,G_0}(X_0=m| X_{-1}=l)_{|\alpha=\alpha_0} \\
&= \sum\limits_{l=0}^\infty P^{\alpha_0,G_0}(X_{-1}=l)  \frac{\partial^2}{\partial \alpha^2}  1 
\\
&= 0.
\end{align*}
With the same arguments, we can treat the other three expressions to show that the corresponding first terms have mean zero.
As a preliminary result, we proved that the expectations of the two terms in \eqref{eq:anfang2}, i.e., of $\dot\Psi_{n11}^{\theta_0}$ and $\dot\Psi_{11}^{\theta_0}$, coincide. To 
conclude that the whole expression in \eqref{eq:anfang2} converges to zero, by Chebychev's inequality, it remains to show that 
\begin{align*}
 Var_{\az,\gz}(\dot \Psi_{n11}^{\tz}) \rightarrow 0 \quad\text{as $n\rightarrow \infty$}.
\end{align*}
We proceed as in the proof of Theorem \ref{cons_bs_est} and, analogously to \eqref{cov_def2}, we get
\begin{align} \label{cov_deriv2}
   & Var \left( \frac{\partial^2}{\partial \alpha^2} \frac{1}{n} \sum\limits_{t=0}^n log \left( \sum\limits_{j=0}^{\min\{ X_{t-1}, X_t\}} \binom{X_{t-1}}{j} \alpha^j (1-\alpha)^{X_{t-1}-j} G(X_t-j) \right)_{|\alpha=\alpha_0} \right) \notag \\
    &\leq \frac{1}{n}\sum\limits_{h=-(n-1)}^{n-1}  \text{Cov} \left( \frac{\partial^2}{\partial \alpha^2}  \log \left( \sum\limits_{j_1=0}^{\text{min}(X_1,X_0)} \binom{X_0}{j_1}\alpha^{j_1}(1-\alpha)^{X_0-j_1} G(X_1-j_1) \right)_{|\alpha=\alpha_0} \right. , \notag \\
 & \qquad  \qquad \quad \quad \;  \; \; \left. \frac{\partial^2}{\partial \alpha^2} \log \left( \sum\limits_{j_2=0}^{\text{min}(X_{h+1},X_h)} \binom{X_h}{j_2}\alpha^{j_2}(1-\alpha)^{X_h-j_2}G(X_{h+1}-j_2)  \right)_{|\alpha=\alpha_0} \right).
\end{align}
Consequently, by defining
\begin{align}
    Y_h := \frac{\partial^2}{\partial \alpha^2}   \log \left(   \sum\limits_{j=0}^{\min(X_{h+1}, X_h)} \binom{X_h}{j} \alpha^j(1-\alpha)^{X_h-j}  G(X_{h+1}-j)  \right)_{|\alpha=\alpha_0},
\end{align}
and again using Assumption \ref{asum_mix_gu}, Lemma \ref{lemma_rng}, Theorem 14.1 of \citet{davidson} and Corollary 14.3 of \citet{davidson}, we get 
\begin{align}
 |\text{Cov}(Y_0, Y_{h} )| \leq 2\, (2^{1-1/q}+1) \, \alpha_\text{mix}(h)^{1-1/q-1/r} \, ||Y_0||_q ||Y_{h}||_r.   
\end{align}
Hence, it remains to show that 
\begin{align} \label{exi_q}
    \exists q > 1, r > q/(q-1):||Y_0||_q=O(1) \quad \text{and} \quad ||Y_{h}||_r=O(1),
\end{align} 
where it again suffices to prove it for $h=0$ due to stationarity of $(X_t,t \in \mathbb{Z})$ and, for some $q>2$,
we have to investigate
\begin{align} \label{bip_anfang}
      \sum\limits_{(x,y) \in \mathbb{N}_0^2} \left| \frac{\partial^2}{\partial \alpha^2}  \log \left( \sum\limits_{j=0}^{\min(x,y)} \binom{x}{j} \alpha^j(1-\alpha)^{x-j}G(y-j) \right)_{|\alpha=\alpha_0}  \right|^q P( (X_0,X_1)=(x,y)  ).
\end{align}
Using that $q>1$ and $|a-b| \leq |a| + |b|$, we get an upper bound for \eqref{bip_anfang} given by 
\begin{align*}
\sum\limits_{(x,y) \in \mathbb{N}_0^2} \left( \left| \frac{\frac{\partial^2}{\partial \alpha^2} P^{\alpha,G}_{x,y|\alpha=\alpha_0}}{P^{\alpha,G}_{x,y}} \right| + \left| \frac{ \frac{\partial}{\partial \alpha}P^{\alpha,G}_{x,y|\alpha=\alpha_0}}{P^{\alpha,G}_{x,y}} \right|^2 \right)^q P( (X_0,X_1)=(x,y)  ).
\end{align*}
Using as in the Proof of Theorem \ref{cons_bs_est} that $x^q$ is a convex function for $x \geq 0$, that by \eqref{abs_c} we have \begin{align*}
    \left|\frac{\partial^2}{\partial \alpha^2} P^{\alpha, G}_{x,y|\alpha=\alpha_0} \right| \leq 4x(x-1)(1-\alpha_0)^{-2}
\end{align*} and that by \eqref{abs_an}, we have
\begin{align*}
    \left| \frac{\partial}{\partial \alpha} P^{\alpha, G}_{x,y|\alpha=\alpha_0}   \right|^2 \leq 4x^2,
\end{align*} 
we get an upper bound for \eqref{bip_anfang} given by
\begin{align*}
    2^{3q-1}(1-\alpha_0)^{-2}\text{E}( (X_0(X_0-1))^q  ) + 2^{3q-1} \text{E}((X_0^2)^q).
\end{align*}
Using the boundedness of the moments of the observations together with Assumption \ref{ass1}, \eqref{exi_q} follows and the covariance in \eqref{cov_deriv2} is summable. As in the proof of Theorem \ref{cons_bs_est}, the assertion that $Var_{\az,\gz}(\dot \Psi_{n11}^{\tz}) \rightarrow 0$ as $n \rightarrow \infty$ follows. For the other components, we proceed analogously. Thus, altogether, \eqref{eq:anfang2} converges to zero in probability.
\end{proof}

\begin{Lemma} \label{negl}
Suppose Assumptions \ref{ass1} and \ref{asum_mix_gu} hold true. Let $\text{E}(X_1^k) < \infty$ for some $k>2(p+4)$ and
$\text{E}((X_t^3(1+\rho)^{X_t}))^{1+\delta}<\infty$ for some $\rho, \delta>0$.
Then, we have
\begin{align*}
\sqrt{n}(\Psi_n^*-\Psi_n)(\thns) - \sqrt{n}(\Psi_n^*-\Psi_n)(\thn)=o_{P^*}(1)
\end{align*}
in probability.
\end{Lemma}

\begin{proof}
Following the proof of (L4) in \citet{drost2009appendix}, we consider both components separately. That is, we have to show
\begin{align}
    \sqrt{n}(\Psi_{n1}^*-\Psi_{n1})(\thns) - \sqrt{n}(\Psi_{n1}^*-\Psi_{n1})(\thn) = & o_{P^*}(1),  \label{eq:LemmaA3_1}    \\
    \sqrt{n}(\Psi_{n2}^*-\Psi_{n2})(\thns) - \sqrt{n}(\Psi_{n2}^*-\Psi_{n2})(\thn) = & o_{P^*}(1)   \label{eq:LemmaA3_2}
\end{align}
in probability, respectively. We consider only \eqref{eq:LemmaA3_2} for $p=1$ and omit the details for \eqref{eq:LemmaA3_1}. Then, for $h\in\mathcal{H}_1$, we have
\begin{align*}
& \sqrt{n}\big(\Psi_{n2}^*(\thns)h-\Psi_{n2}(\thns)h-\Psi_{n2}^*(\thn)h + \Psi_{n2}(\thn)h\big) \\
= & \sqrt{n}\Bigg(\frac{1}{n} \sum\limits_{t=0}^n \left( A_{\widehat \alpha^*, \widehat G^*} h(X_{t-1}^*,X_t^*)-\int h \; d\widehat G^*  \right)-\frac{1}{n} \sum\limits_{t=0}^n \left(A_{\widehat\alpha^*,\widehat G^*} h(X_{t-1},X_t)-\int h d\widehat G^*  \right)    \\
& \qquad   -\frac{1}{n} \sum\limits_{t=0}^n \left( A_{\widehat \alpha, \widehat G} h(X_{t-1}^*,X_t^*)-\int h \; d\widehat G  \right)+\frac{1}{n} \sum\limits_{t=0}^n \left(A_{\widehat\alpha,\widehat G} h(X_{t-1},X_t)-\int h d\widehat G  \right)\Bigg)   \\
= & \frac{1}{\sqrt{n}} \sum\limits_{t=0}^n \Bigg( A_{\widehat \alpha^*, \widehat G^*} h(X_{t-1}^*,X_t^*)-A_{\widehat\alpha^*,\widehat G^*} h(X_{t-1},X_t)-A_{\widehat \alpha, \widehat G} h(X_{t-1}^*,X_t^*)+A_{\widehat\alpha,\widehat G} h(X_{t-1},X_t)\Bigg)
\end{align*}
where, according to \eqref{eq:Ah_explicit} and for $x_{t-1},x_t\in\mathbb{N}_0$ and $(\alpha,G)\in \Theta\times \mathcal{G}$, we have
\begin{align*}
    A_{\alpha,G} h(x_{t-1},x_t) 
    &= \frac{\sum\limits_{s=0}^{\min(x_t,x_{t-1})} h(x_t-s) \binom{x_{t-1}}{s} \alpha^s(1-\alpha)^{x_{t-1}-s} G(x_t-s)}{\sum\limits_{s=0}^{\min(x_t,x_{t-1})} \binom{x_{t-1}}{s} \alpha^s(1-\alpha)^{x_{t-1}-s} G(x_t-s)}.  
\end{align*}
Hence, we have 
\begin{align}
& \sqrt{n}\big(\Psi_{n2}^*(\thns)-\Psi_{n2}(\thns)-\Psi_{n2}^*(\thn) + \Psi_{n2}(\thn)\big) \notag \\ 
=& \frac{1}{\sqrt{n}} \sum\limits_{t=0}^n \Bigg(\frac{d_{h,n}^*(\widehat \theta_n^*)}{d_{n}^*(\widehat \theta_n^*)}-\frac{d_{h,n}(\widehat \theta_n^*)}{d_{n}(\widehat \theta_n^*)}-\frac{d_{h,n}^*(\widehat \theta_n)}{d_{n}^*(\widehat \theta_n)}+\frac{d_{h,n}(\widehat \theta_n)}{d_{n}(\widehat \theta_n)}\Bigg),  \label{eq:represenation_dn}
\end{align}
where
\begin{align*}
d_{h,n}(\widehat \theta_n) =& \sum\limits_{s=0}^{\min(X_t,X_{t-1})} h(X_t-s) \binom{X_{t-1}}{s} \widehat\alpha^s(1-\widehat\alpha)^{X_{t-1}-s} \widehat G(X_t-s),  \\
d_n(\widehat \theta_n) =& \sum\limits_{s=0}^{\min(X_t,X_{t-1})} \binom{X_{t-1}}{s} \widehat\alpha^s(1-\widehat\alpha)^{X_{t-1}-s} \widehat G(X_t-s),   \\
d_{h,n}^*(\widehat \theta_n^*) =& \sum\limits_{s=0}^{\min(X_t,X_{t-1})} h(X_t-s) \binom{X_{t-1}}{s} \widehat\alpha^{*s}(1-\widehat\alpha^*)^{X_{t-1}-s} \widehat G^*(X_t-s), \\
d_n^*(\widehat \theta_n^*) =& \sum\limits_{s=0}^{\min(X_t,X_{t-1})} \binom{X_{t-1}}{s} \widehat\alpha^{*s}(1-\widehat\alpha^*)^{X_{t-1}-s} \widehat G^*(X_t-s)
\end{align*}
and $d_{h,n}(\widehat \theta_n^*), d_{n}(\widehat \theta_n^*), d^*_{h,n}(\widehat \theta_n)$ and $d_{n}^*(\widehat \theta_n)$ analog.
Finally, by adding suitable zero and using the same technique as in the proof of Lemma \ref{Lipschitz_resultsII}, we can bound the above expression \eqref{eq:represenation_dn} by $\sqrt{n}\|\widehat \theta_n^*-\widehat \theta_n\|$, which is $O_{P^*}(1)$ multiplied by a term that converges to zero in probability.
\end{proof}

\begin{Lemma} \label{asympt_norm}
Suppose Assumptions \ref{ass1} and \ref{asum_mix_gu} hold. Let $\text{E}(X_1^k) < \infty$ for some $k>2(p+4)$ and
$\text{E}((X_t^3(1+\rho)^{X_t}))^{1+\delta}<\infty$ for some $\rho, \delta>0$. Then, we have
\begin{align}\label{eq:S_n_weak_convergence_boot}
\mathcal{S}_n^{\widehat{\fett\alpha}_n, \widehat G_n,*} = \sqrt{n} \left(\Psi_n^*(\widehat{\fett\alpha}_n, \widehat G_n) -\Psi_n(\widehat{\fett\alpha}_n, \widehat G_n) \right) \leadsto^* \mathcal{S}^{\fett\alpha_0, G_0}
\end{align}
in $\mathbb{R}^p \times \ell^\infty(\mathcal{H}_1)$, under $P^*(\cdot)= \mathds{P}_{\nu_{\widehat{\fett\alpha}_n,\widehat{G}_n}, \widehat{\fett\alpha}_n, \widehat{G}_n}(\cdot|\mathbb{X})$ in $\mathds{P}_{\nu_{\theta_0}, \theta_0}$-probability, where $\mathcal{S}^{\fett\alpha_0, G_0}$ is the tight, Borel measurable, Gaussian process obtained in \eqref{eq:S_n_weak_convergence}.
\end{Lemma} 

\begin{proof}
Note that $\text{E}^*(\Psi_n^*(\widehat{\fett\alpha}_n, \widehat G_n))= \Psi_n(\widehat{\fett\alpha}_n, \widehat G_n)$ and by similar techniques as used in Lemma \ref{negl}, we get 
\begin{align*}
\sqrt{n} \big(\Psi_n^*(\widehat{\fett\alpha}_n, \widehat G_n) -\Psi_n(\widehat{\fett\alpha}_n, \widehat G_n) \big) = \sqrt{n} \big(\Psi_n^*(\fett\alpha_0,G_0) - \text{E}^*(\Psi_n^*(\fett\alpha_0,G_0))\big)+o_{P^*}(1).   
\end{align*}
Hence, recalling the definition of $\Psi_n^*(\fett\alpha,G)$ in \eqref{Psi_n1*} and \eqref{Psi_n2*}, the leading term on the last right-hand side can be represented as a function of generalized means. That is, we have
\begin{align*}
& \sqrt{n} \big(\Psi_n^*(\fett\alpha_0,G_0) - \text{E}^*(\Psi_n^*(\fett\alpha_0,G_0))\big)    \\
= &  \sqrt{n}\left(f\left(\frac{1}{n_m} \sum\limits_{t=1}^{n_m} g(X_t^*, \ldots, X_{t+m-1}^*)\right) -f\left(\text{E}^*\left(\frac{1}{n_m} \sum\limits_{t=1}^{n_m} g(X_t^*, \ldots, X_{t+m-1}^*)\right)\right)  \right),
\end{align*}
where $f$ is the identity function, $m=p+1$, $n_m=n-p-1$ and $g=(g_1,g_2)$ with
\begin{align*}
g_1(x_{t}, \ldots, x_{t-p}) &= \dot{l}_{\fett\alpha}(x_{t-p}, \ldots, x_t; \fett\alpha_0, G_0) 
 \end{align*}
and 
\begin{align*}
g_2(x_t, \ldots, x_{t-p})h &= A_{\fett\alpha_0 G_0} h(x_{t-p}, \ldots, x_t) - \int hdG_0,  \quad h \in \mathcal{H}_1
\end{align*}
for $x_t,\ldots,x_{t-p}\in\mathbb{N}_0$. This representation as a function of generalized means allows to use Corollary 4.2 in \citet{jewe}.
Finally, for allowing the application of this result, according to Assumption 1 in \citet{jewe}, sufficient smoothness properties have to be fulfilled. As $f$ is just the identity, all smoothness properties hold. It remains to argue that all partial derivatives of $g_1(x_{t}, \ldots, x_{t-p})$ and $g_2(x_{t}, \ldots, x_{t-p})$ with respect to $x_t,\ldots,x_{t-p}$ are Lipschitz continuous. However, as both function $g_1$ and $g_2$ ask for arguments from $\mathbb{N}_0^{p+1}$ only, they can be arbitrarily extended to functions on $\mathbb{R}^{p+1}$ such that all sufficient smoothness conditions will be fulfilled by construction. 
\end{proof}

\newpage

\section{Auxiliary Results}\label{sec:appendixB}

\begin{Lemma}[NPMLE fulfills $\widehat \theta_n=(\widehat{\fett\alpha}_n,\widehat G_n)\in \Theta\times\mathcal{G}$ in probability] \label{lemma_rng}
Suppose Assumption \ref{ass1} holds true. Let $\text{E}(X_1^k) < \infty$ for some $k>2(p+4)$. Then, 
with $\mathds{P}_{\nu_{\fett\alpha_0,G_0}, \fett\alpha_0, G_0}$-probability tending to one, it holds that $\widehat{G}_n 
\in \mathcal{G}=\{G \in \widetilde{\mathcal{G}}: 0 < G(0)  <1: \text{E}_G (\varepsilon^{p+4}_t) < \infty  \}$ 
and $\widehat{\fett\alpha}_n \in \Theta =\{ \fett\alpha \in (0,1)^p: \sum_{i=1}^p \alpha_i <1 \}$, respectively. Note that $\widehat{G}_n = \mathcal{L}^*(\varepsilon_t^*)$ and $\text{E}^* (\varepsilon_t^{*\, p+4})=\sum_{k=0}^\infty k^{p+4} \widehat{G}_n(k)$ for the bootstrap procedure described in Section \ref{sec:boot_algo}
\end{Lemma}

\begin{proof} We divide the proof into four parts. In the following, we will prove
\begin{itemize}
\item[(i)] $0< \widehat{G}_n(0) < 1$ in $P=\mathds{P}_{\nu_{\fett\alpha_0,G_0}, \fett\alpha_0, G_0}$-probability, that is, with $\mathds{P}_{\nu_{\fett\alpha_0,G_0}, \fett\alpha_0, G_0}$-probability tending to one,
\item[(ii)] $\sum_{k=0}^\infty k^{p+4} \widehat{G}_n(k)=O_P(1)$,
\item[(iii)] $\widehat{\fett\alpha}_n \in (0,1)^p$ in $\mathds{P}_{\nu_{\fett\alpha_0,G_0}, \fett\alpha_0, G_0}$-probability,
\item[(iv)] $\sum\limits_{i=1}^p \widehat{\alpha}_{n,i} < 1$ in $\mathds{P}_{\nu_{\fett\alpha_0,G_0}, \fett\alpha_0, G_0}$-probability. 
\end{itemize}
For showing part (i), we make use of the equivalence
\begin{align*}
P(0< \widehat{G}_n(0) < 1) \underset{n \rightarrow \infty}{\rightarrow} 1    \quad \Leftrightarrow   \quad P(\widehat{G}_n(0)\in \{0,1\})\underset{n \rightarrow \infty}{\rightarrow} 0. 
\end{align*}
Further, we have 
\begingroup
\allowdisplaybreaks
\begin{align*}
P(\widehat{G}_n(0) \in \{0,1\}) &= P\left(\widehat{G}_n(0) \in \{0,1\}, |\widehat{G}_n(0)-G_0(0)| < \frac{G_0(0)}{2}\right) \\ & \quad +  P\left(\widehat{G}_n(0) \in \{0,1\}, |\widehat{G}_n(0)-G_0(0)| \geq \frac{G_0(0)}{2}\right) \\  &=: \text{I} + \text{II},
\end{align*}
\endgroup
where
\begin{align*}
 \text{I} 
 =  P\left( \left\lbrace \widehat{G}_n(0) \in \{0,1\} \right\rbrace \cap \left\lbrace  \frac{-G_0(0)}{2} < \widehat{G}_n(0) - G_0(0) < \frac{G_0(0)}{2} \right\rbrace \right) = P\left(\emptyset\right) = 0
\end{align*}
and 
\begin{align*}
\text{II} 
\leq  P\left(| \widehat{G}_n(0)-G_0(0) | \geq \frac{G_0(0)}{2}\right) \underset{n \rightarrow \infty}{\rightarrow  0}
\end{align*}
according to Theorem \ref{cons}. For proving part (ii), we have 
\begin{align*} & \sum\limits_{k=0}^\infty k^{p+4} \widehat{G}_n(k)  \\
&= \sum\limits_{k=0}^{\text{max}(X_1, \ldots, X_n)} k^{p+4} \widehat{G}_n(k) \\
&= \sum\limits_{k=0}^{\text{max}(X_1, \ldots, X_n)} k^{p+4} \big(\widehat{G}_n(k) - G_0(k) \big) + \sum\limits_{k=0}^{\text{max}(X_1, \ldots, X_n)} k^{p+4} G_0(k) \\
&= \sum\limits_{k=0}^{\text{max}(X_1, \ldots, X_n)} k^{p+4} \big(\widehat{G}_n(k) - G_0(k) \big) + \sum\limits_{k=0}^\infty k^{p+4} G_0(k)- \sum\limits_{k=\text{max}(X_1, \ldots, X_n)+1}^\infty  k^{p+4} G_0(k) \\ & =:  I + II + III.
\end{align*}
For term $I$, we have
\begin{align*}
I &= \sum\limits_{k=0}^{\text{max}(X_1, \ldots, X_n)} k^{p+4} \big(\widehat{G}_n(k) - G_0(k) \big)  \leq \text{max}(X_1, \ldots, X_n)^{p+4} \sum\limits_{k=0}^{\text{max}(X_1, \ldots, X_n)} \big(\widehat{G}_n(k) - G_0(k) \big).
\end{align*}
With the use of $\sqrt{n}\sum_{k=0}^\infty (\widehat{G}_n(k) - G_0(k) ) = O_p(1)$ (see Theorem \ref{clt}), the latter becomes
\begin{align*}
\frac{\text{max}(X_1, \ldots, X_n)^{p+4}}{\sqrt{n}} \left(\sqrt{n}\sum\limits_{k=0}^\infty \big(\widehat{G}_n(k) - G_0(k) \big)\right) = O_p\left(\frac{\text{max}(X_1, \ldots, X_n)^{p+4}}{\sqrt{n}}\right).
\end{align*}
Further, for all $x>0$ and all $k \geq 0$, from Markov's inequality, we get
\begin{align*}
& \quad P \left(\frac{\text{max}(X_1, \ldots, X_n)^{p+4}}{\sqrt{n}} > x \right) = P \left(\text{max}(X_1, \ldots, X_n) > x^{1/(p+4)} n^{1/(2(p+4))}\right) \\ 
& = P\left(\bigcup_{i=1}^n \left\lbrace X_i >  x^{1/(p+4)} n^{1/(2(p+4))}  \right\rbrace  \right) \leq \sum\limits_{i=1}^n P\left(X_i >  x^{1/(p+4)} n^{1/(2(p+4))}\right) \\ 
&= n \cdot P\left(X_1 >  x^{1/(p+4)} n^{1/(2(p+4))}\right)  \leq \frac{n \cdot \text{E}(X_1^k)}{x^{k/(p+4)} n^{k/(2(p+4))}} \\ &= \frac{\text{E}(X_1^k)}{x^{k/(p+4)}} \cdot n^{1-k/(2(p+4))},
\end{align*} 
which converges to zero for $k>2(p+4)$ if $\text{E}(X_1^k) < \infty$. This implies $\frac{\text{max}(X_1, \ldots, X_n)^{p+4}}{\sqrt{n}} = o_p(1)$ such that $I=o_P(1)$ holds. As $II<\infty$, i.e., $II=O(1)$, by Assumption \ref{ass1}, the assertion follows from $III=o_P(1)$, which we show next. Here, we distinguish two cases. If the support of $(\varepsilon_t,t \in \mathbb{Z})$ is bounded, then $\sum_{k=\text{max}(X_1, \ldots, X_n)+1}^\infty  k^{p+4} G_0(k)=0$ after one observation $X_t$ attains a value greater or equal to the largest possible innovation which happens with probability tending to one. If the support is unbounded, we have $\text{max}(X_1, \ldots, X_n) \rightarrow \infty$ such that $\sum_{k=\text{max}(X_1, \ldots, X_n)+1}^\infty  k^{p+4} G_0(k) = o_P(1)$ as again by Assumption \ref{ass1}, we have $\text{E}_{G_0}(\varepsilon_t^{p+4}) < \infty$, i.e., $(k^{p+4} \, G_0(k), k\in\mathbb{N}_0)$ is summable. 
For part (iii), we have that $ \widehat{\fett\alpha}_n \in (0,1)^p$ if and only if $\widehat{\alpha}_{n,i} \in (0,1)$ for all $i=1, \ldots, p$. Hence, the proof is analogous to the proof of part (i).

Similarly, we can show that $\sum\limits_{i=1}^p \widehat{\alpha}_{n,i} < 1$ holds in $\mathds{P}_{\nu_{\fett\alpha_0,G_0}, \fett\alpha_0, G_0}$-probability, because
\begin{align*}
\sum\limits_{i=1}^p \widehat{\alpha}_{n,i} &= \sum\limits_{i=1}^p (\widehat{\alpha}_{n,i}-\alpha_{0,i}) + \sum\limits_{i=1}^p \alpha_{0,i}, 
\end{align*}
where the first term converges to zero in probability due to Theorem \ref{cons} and the second term is smaller than 1 according to Assumption \ref{ass1}. 
\end{proof}

\bigskip

\begin{Lemma}[Partial derivatives of $\Psi_n$]\label{lem_frechet_derivation}
Under the assumptions of Lemma \ref{frechet}, for $p=1$ and for all $k\in\mathbb{N}_0$, we have
\begin{align}
\frac{\partial \Psi_{n1}(\alpha,G)}{\partial \alpha} &= \frac{1}{n} \sum\limits_{t=0}^n \frac{\frac{\partial^2}{\partial \alpha^2} P_{X_{t-1},X_t}^{\alpha,G}}{P_{X_{t-1},X_t}^{\alpha,G}}
- \frac{1}{n} \sum\limits_{t=0}^n \dot{l}_\alpha^2(X_{t-1},X_t;\alpha,G),    \label{partial_Psi_n1_alpha}\\
\frac{\partial \Psi_{n1}(\alpha,G)}{\partial G(k)} &= \frac{1}{n} \sum\limits_{t=0}^n  \frac{\frac{\partial}{\partial G(k)} \frac{\partial}{\partial \alpha} P_{X_{t-1},X_t}^{\alpha,G} }{P_{X_{t-1},X_t}^{\alpha,G}}
- \frac{1}{n} \sum\limits_{t=0}^n  \dot{l}_\alpha(X_{t-1},X_t;\alpha,G) \frac{\frac{\partial}{\partial G(k)}P_{X_{t-1},X_t}^{\alpha,G}}{P_{X_{t-1},X_t}^{\alpha,G}}\label{partial_Psi_n1_G}
\end{align}
and, for $h\in\mathcal{H}_1$,
\begin{align}
\frac{\partial \Psi_{n2}(\alpha,G)}{\partial \alpha}h 
&=  \frac{1}{n} \sum\limits_{t=0}^n  \sum\limits_{j=0}^\infty h(j) \frac{\Big(\frac{\partial}{\partial \alpha} P^{\alpha,G}(\varepsilon_t =j, X_t=x_t|X_{t-1}=x_{t-1})\Big)_{|_{(x_t,x_{t-1})=(X_t,X_{t-1})}}}{P^{\alpha,G}_{X_{t-1},X_t}}  \notag \\
&   \qquad 
- \frac{1}{n} \sum\limits_{t=0}^n \sum\limits_{j=0}^\infty h(j) \dot{l}_\alpha(X_{t-1},X_t;\alpha,G) A_{\alpha,G}h(X_{t-1},X_t), \label{partial_Psi_n2_alpha}    \\
\frac{\partial \Psi_{n2}(\alpha,G)}{\partial G(k)}h &= \frac{1}{n} \sum\limits_{t=0}^n  \sum \limits_{j=0}^\infty h(j) \frac{\Big(\frac{\partial}{\partial G(k)} P^{\alpha,G}(\varepsilon_t=j, X_t=x_t | X_{t-1}=x_{t-1})\Big)_{|_{(x_t,x_{t-1})=(X_t,X_{t-1})}}}{P^{\alpha,G}_{X_{t-1},X_t}}   \notag
 \\
& \quad  - \frac{1}{n} \sum\limits_{t=0}^n A_{\alpha,G}h(X_{t-1},X_t) \frac{ \Big(\frac{\partial}{\partial G(k)} P^{\alpha,G}(X_t=x_t | X_{t-1}=x_{t-1})\Big)_{|_{(x_t,x_{t-1})=(X_t,X_{t-1})}}}{P^{\alpha,G}_{X_{t-1},X_t}}  \notag    \\
& \quad - h(k).\label{partial_Psi_n2_G}
\end{align}
\end{Lemma}

\begin{proof}
For the first two derivatives in \eqref{partial_Psi_n1_alpha} and \eqref{partial_Psi_n1_G}, recall from \eqref{eq:psi_n1} that
\begin{align*}
\Psi_{n1}(\alpha, G) &= \frac{1}{n} \sum\limits_{t=0}^n \dot l_{\alpha}(X_{t-1},X_t; \alpha, G),
\end{align*}
where 
\begin{align*}
\dot l_\alpha (x_{t-1},x_t; \alpha, G) = \frac{\partial}{\partial \alpha} \log \left(P^{\alpha, G}_{x_{t-1},x_t}\right).    
\end{align*}
Hence, for \eqref{partial_Psi_n1_alpha}, we get
\begin{align*} 
\frac{\partial \Psi_{n1}(\alpha,G)}{\partial \alpha} &= \frac{\partial}{\partial \alpha} \frac{1}{n} \sum \limits_{t=0}^n \dot{l}_\alpha(X_{t-1},X_t;\alpha,G)= \frac{\partial}{\partial \alpha} \frac{1}{n} \sum \limits_{t=0}^n \frac{\partial}{\partial \alpha} \log \left(P_{X_{t-1},X_t}^{\alpha,G} \right) \notag \\
&= \frac{1}{n} \sum\limits_{t=0}^n \frac{\partial^2}{\partial \alpha^2} \log \left(P_{X_{t-1},X_t}^{\alpha,G} \right)= \frac{1}{n} \sum\limits_{t=0}^n \frac{P_{X_{t-1},X_t}^{\alpha,G} \big(\frac{\partial^2}{\partial \alpha^2} P_{X_{t-1},X_t}^{\alpha,G}\big)- \left( \frac{\partial}{\partial \alpha} P_{X_{t-1},X_t}^{\alpha,G} \right)^2}{\left(P_{X_{t-1},X_t}^{\alpha,G}\right)^2} \notag \\
&=  \frac{1}{n} \sum\limits_{t=0}^n  \left( \frac{\frac{\partial^2}{\partial \alpha^2} P_{X_{t-1},X_t}^{\alpha,G}}{P_{X_{t-1},X_t}^{\alpha,G}} - \left( \frac{\frac{\partial}{\partial \alpha} P_{X_{t-1},X_t}^{\alpha,G}}{P_{X_{t-1},X_t}^{\alpha,G}}  \right)^2     \right) \notag \\
&= \frac{1}{n} \sum\limits_{t=0}^n  \left( \frac{\frac{\partial^2}{\partial \alpha^2} P_{X_{t-1},X_t}^{\alpha,G}}{P_{X_{t-1},X_t}^{\alpha,G}} - \dot{l}_\alpha^2(X_{t-1},X_t;\alpha,G)  \right) \notag \\
&= \frac{1}{n} \sum\limits_{t=0}^n \frac{\frac{\partial^2}{\partial \alpha^2} P_{X_{t-1},X_t}^{\alpha,G}}{P_{X_{t-1},X_t}^{\alpha,G}}
- \frac{1}{n} \sum\limits_{t=0}^n \dot{l}_\alpha^2(X_{t-1},X_t;\alpha,G).
\end{align*}
Similarly, for \eqref{partial_Psi_n1_G} and for all $k\in\mathbb{N}_0$, we have
\begin{align*} 
\frac{\partial \Psi_{n1}(\alpha,G)}{\partial G(k)} &= \frac{\partial}{\partial G(k)} \frac{1}{n} \sum \limits_{t=0}^n \dot{l}_\alpha(X_{t-1},X_t;\alpha,G) \notag = \frac{1}{n} \sum\limits_{t=0}^n \frac{\partial}{\partial G(k)} \frac{\frac{\partial}{\partial \alpha}P_{X_{t-1},X_t}^{\alpha,G} }{P_{X_{t-1},X_t}^{\alpha,G}} \notag \\
&= \frac{1}{n} \sum\limits_{t=0}^n \frac{\left( \frac{\partial}{\partial G(k)} \frac{\partial}{\partial \alpha} P_{X_{t-1},X_t}^{\alpha,G}\right)P_{X_{t-1},X_t}^{\alpha,G} - \left(\frac{\partial}{\partial \alpha} P_{X_{t-1},X_t}^{\alpha,G}\right)\left(\frac{\partial}{\partial G(k)} P_{X_{t-1},X_t}^{\alpha,G}\right)  }{ \left( P_{X_{t-1},X_t}^{\alpha,G}  \right)^2} \notag \\
&= \frac{1}{n} \sum\limits_{t=0}^n \left( \frac{\frac{\partial}{\partial G(k)} \frac{\partial}{\partial \alpha} P_{X_{t-1},X_t}^{\alpha,G} }{P_{X_{t-1},X_t}^{\alpha,G}} - \frac{\left(\frac{\partial}{\partial \alpha} P_{X_{t-1},X_t}^{\alpha,G}\right)\left(\frac{\partial}{\partial G(k)} P_{X_{t-1},X_t}^{\alpha,G}\right)}{\left( P_{X_{t-1},X_t}^{\alpha,G} \right)^2}  \right) \notag \\
&= \frac{1}{n} \sum\limits_{t=0}^n \left( \frac{\frac{\partial}{\partial G(k)} \frac{\partial}{\partial \alpha} P_{X_{t-1},X_t}^{\alpha,G} }{P_{X_{t-1},X_t}^{\alpha,G}} - \dot{l}_\alpha(X_{t-1},X_t;\alpha,G) \frac{\frac{\partial}{\partial G(k)}P_{X_{t-1},X_t}^{\alpha,G}}{P_{X_{t-1},X_t}^{\alpha,G}}  \right) \notag \\
&= \frac{1}{n} \sum\limits_{t=0}^n  \frac{\frac{\partial}{\partial G(k)} \frac{\partial}{\partial \alpha} P_{X_{t-1},X_t}^{\alpha,G} }{P_{X_{t-1},X_t}^{\alpha,G}}
- \frac{1}{n} \sum\limits_{t=0}^n  \dot{l}_\alpha(X_{t-1},X_t;\alpha,G) \frac{\frac{\partial}{\partial G(k)}P_{X_{t-1},X_t}^{\alpha,G}}{P_{X_{t-1},X_t}^{\alpha,G}}.
\end{align*}
For the last two derivatives \eqref{partial_Psi_n2_alpha} and \eqref{partial_Psi_n2_G}, recall from \eqref{eq:psi_n2} that 
\begin{align*}
\Psi_{n2}(\alpha, G) h = \frac{1}{n} \sum\limits_{t=0}^n \left(A_{\alpha,G} h(X_{t-1},X_t)-\int h dG  \right), \quad h \in \mathcal{H}_1,    
\end{align*}
holds, where    
\begin{align}
A_{\alpha,G} h(x_{t-1},x_t) =& \text{E}_{\alpha,G}(h(\varepsilon_t)| X_t=x_t, X_{t-1}=x_{t-1})  \notag  \\
=& \sum\limits_{j=0}^\infty h(j) P^{\alpha,G}(\varepsilon_t =j |X_t=x_t,X_{t-1}=x_{t-1}) \notag \\
=& \sum\limits_{j=0}^\infty h(j) \frac{P^{\alpha,G}(\varepsilon_t =j, X_t=x_t,X_{t-1}=x_{t-1} )}{P^{\alpha,G}(X_t=x_t,X_{t-1}=x_{t-1})} \frac{P^{\alpha,G}(X_{t-1}=x_{t-1})}{P^{\alpha,G}(X_{t-1}=x_{t-1})} \notag \\
=& \sum\limits_{j=0}^\infty h(j) \frac{P^{\alpha,G}(\varepsilon_t =j, X_t=x_t,X_{t-1}=x_{t-1} )}{ P^{\alpha,G}(X_t=x_t|X_{t-1}=x_{t-1})  P^{\alpha,G}(X_{t-1}=x_{t-1})} \notag \\
=& \sum\limits_{j=0}^\infty h(j) \frac{P^{\alpha,G}(\varepsilon_t =j, X_t=x_t|X_{t-1}=x_{t-1} )}{P^{\alpha,G}(X_t=x_t|X_{t-1}=x_{t-1}) }.    \label{eq: A_expanded_representation}
\end{align}
Hence, for \eqref{partial_Psi_n2_alpha}, we get
\begin{align*} 
\frac{\partial \Psi_{n2}(\alpha,G)}{\partial \alpha}h = \frac{\partial}{\partial \alpha} \frac{1}{n} \sum\limits_{t=0}^n \left( A_{\alpha,G} h(X_{t-1},X_t) - \int hdG \right) 
=  \frac{1}{n} \sum\limits_{t=0}^n  \frac{\partial}{\partial \alpha}  A_{\alpha,G} h(X_{t-1},X_t),
\end{align*}
where, for all $x_t,x_{t-1}\in\mathbb{N}_0$, using \eqref{eq: A_expanded_representation}, we have
\begin{align*} 
& \frac{\partial}{\partial \alpha} A_{\alpha,G} h(x_{t-1},x_t)    \\ 
=& \sum\limits_{j=0}^\infty h(j) \frac{\partial}{\partial \alpha} \frac{P^{\alpha,G}(\varepsilon_t =j, X_t=x_t|X_{t-1}=x_{t-1} )}{P^{\alpha,G}(X_t=x_t|X_{t-1}=x_{t-1}) } \\
=& \sum\limits_{j=0}^\infty h(j) \Bigg[\frac{\left(\frac{\partial}{\partial \alpha} P^{\alpha,G}(\varepsilon_t =j, X_t=x_t|X_{t-1}=x_{t-1} )\right) P^{\alpha,G}(X_t=x_t|X_{t-1}=x_{t-1})}{\left( P^{\alpha,G}(X_t=x_t|X_{t-1}=x_{t-1})  \right)^2}  \\
& \quad -\frac{\left(\frac{\partial}{\partial \alpha} P^{\alpha,G}(X_t=x_t|X_{t-1}=x_{t-1})\right)P^{\alpha,G}(\varepsilon_t =j, X_t=x_t|X_{t-1}=x_{t-1} )}{\left( P^{\alpha,G}(X_t=x_t|X_{t-1}=x_{t-1})  \right)^2}\Bigg]
 \\
=& \sum\limits_{j=0}^\infty h(j) \frac{\frac{\partial}{\partial \alpha} P^{\alpha,G}(\varepsilon_t =j, X_t=x_t|X_{t-1}=x_{t-1} )}{P^{\alpha,G}(X_t=x_t|X_{t-1}=x_{t-1})} \\ 
& \quad - \sum\limits_{j=0}^\infty h(j) \dot{l}_\alpha(x_{t-1},x_t;\alpha,G) \frac{P^{\alpha,G}(\varepsilon_t =j, X_t=x_t|X_{t-1}=x_{t-1} )}{P^{\alpha,G}(X_t=x_t|X_{t-1}=x_{t-1})} \\
=&  \sum\limits_{j=0}^\infty h(j) \frac{\frac{\partial}{\partial \alpha} P^{\alpha,G}(\varepsilon_t =j, X_t=x_t|X_{t-1}=x_{t-1} )}{P^{\alpha,G}(X_t=x_t|X_{t-1}=x_{t-1})} - \dot{l}_\alpha(x_{t-1},x_t;\alpha,G) A_{\alpha,G}h(x_{t-1},x_t).
\end{align*}
Recalling that $P^{\alpha,G}(X_t=x_t|X_{t-1}=x_{t-1}) = P^{\alpha,G}_{x_{t-1},x_t}$, altogether, we have
\begin{align*}
\frac{\partial \Psi_{n2}(\alpha,G)}{\partial \alpha}h
&=  \frac{1}{n} \sum\limits_{t=0}^n  \sum\limits_{j=0}^\infty h(j) \frac{\Big(\frac{\partial}{\partial \alpha} P^{\alpha,G}(\varepsilon_t =j, X_t=x_t|X_{t-1}=x_{t-1})\Big)_{|_{(x_t,x_{t-1})=(X_t,X_{t-1})}}}{P^{\alpha,G}_{X_{t-1},X_t}}  \\
&   \qquad 
- \frac{1}{n} \sum\limits_{t=0}^n \dot{l}_\alpha(X_{t-1},X_t;\alpha,G) A_{\alpha,G}h(X_{t-1},X_t). \notag
\end{align*}
Similarly, for \eqref{partial_Psi_n2_G} and for all $k\in\mathbb{N}_0$, we have
\begin{align} \label{eq:fourth_comp}
\frac{\partial \Psi_{n2}(\alpha,G))}{\partial G(k)}h &= \frac{\partial}{\partial G(k)} \frac{1}{n} \sum\limits_{t=0}^n \left( A_{\alpha,G} h(X_{t-1},X_t) - \int hdG   \right), \notag  \\
&= \frac{1}{n} \sum\limits_{t=0}^n  \frac{\partial}{\partial G(k)} A_{\alpha,G} h(X_{t-1},X_t) - \frac{1}{n} \sum\limits_{t=0}^n  \frac{\partial}{\partial G(k)} \int hdG, \notag
\end{align}
where 
\begin{align}
\frac{1}{n} \sum\limits_{t=0}^n  \frac{\partial}{\partial G(k)} \int hdG = \frac{\partial}{\partial G(k)} \sum_{j=0}^\infty h(j)G(j) =  \sum_{j=0}^\infty h(j)\frac{\partial}{\partial G(k)}G(j) = h(k)
\end{align}
and, for all $x_t,x_{t-1}\in\N_0$, using \eqref{eq: A_expanded_representation}, we have
\begin{align*}
& \frac{1}{n} \sum\limits_{t=0}^n  \frac{\partial}{\partial G(k)} A_{\alpha,G} h(x_{t-1},x_t)  \\
=& \frac{1}{n} \sum\limits_{t=0}^n  \sum \limits_{j=0}^\infty h(j) \frac{\partial}{\partial G(k)} \left(\frac{P^{\alpha,G}(\varepsilon_t=j, X_t=x_t | X_{t-1}=x_{t-1})}{P^{\alpha,G}(X_t=x_t | X_{t-1}=x_{t-1})}\right) \notag \\
=& \frac{1}{n} \sum\limits_{t=0}^n  \sum \limits_{j=0}^\infty h(j)\Bigg[ \frac{\left(\frac{\partial}{\partial G(k)} P^{\alpha,G}(\varepsilon_t=j, X_t=x_t | X_{t-1}=x_{t-1})\right) P^{\alpha,G}(X_t=x_t | X_{t-1}=x_{t-1})}{\left( P^{\alpha,G}(X_t=x_t | X_{t-1}=x_{t-1}) \right)^2} \notag \\
& \quad-\frac{\left(\frac{\partial}{\partial G(k)} P^{\alpha,G}(X_t=x_t | X_{t-1}=x_{t-1})\right) P^{\alpha,G}(\varepsilon_t=j, X_t=x_t | X_{t-1}=x_{t-1})}{\left( P^{\alpha,G}(X_t=x_t | X_{t-1}=x_{t-1}) \right)^2}\Bigg] \notag \\
=& \frac{1}{n} \sum\limits_{t=0}^n  \sum \limits_{j=0}^\infty h(j) \frac{\frac{\partial}{\partial G(k)} P^{\alpha,G}(\varepsilon_t=j, X_t=x_t | X_{t-1}=x_{t-1})}{P^{\alpha,G}(X_t=x_t | X_{t-1}=x_{t-1})}  \notag \\
& \quad  - \frac{1}{n} \sum\limits_{t=0}^n  \sum \limits_{j=0}^\infty h(j) \frac{\left(\frac{\partial}{\partial G(k)} P^{\alpha,G}(X_t=x_t | X_{t-1}=x_{t-1}) \right) P^{\alpha,G}(\varepsilon_t=j, X_t=x_t | X_{t-1}=x_{t-1}) }{\left( P^{\alpha,G}(X_t=x_t | X_{t-1}=x_{t-1}) \right)^2}\notag \\
= & \frac{1}{n} \sum\limits_{t=0}^n  \sum \limits_{j=0}^\infty h(j) \frac{\frac{\partial}{\partial G(k)} P^{\alpha,G}(\varepsilon_t=j, X_t=x_t | X_{t-1}=x_{t-1})}{P^{\alpha,G}(X_t=x_t | X_{t-1}=x_{t-1})}
 \\
& \quad  - \frac{1}{n} \sum\limits_{t=0}^n A_{\alpha,G}h(x_{t-1},x_t) \frac{ \frac{\partial}{\partial G(k)} P^{\alpha,G}(X_t=x_t | X_{t-1}=x_{t-1})}{P^{\alpha,G}(X_t=x_t | X_{t-1}=x_{t-1})}. \notag
\end{align*}
Consequently, altogether, we have
\begin{align*}
& \frac{\partial \Psi_{n2}(\alpha,G))}{\partial G(k)}h  \\
&= \frac{1}{n} \sum\limits_{t=0}^n  \sum \limits_{j=0}^\infty h(j) \frac{\Big(\frac{\partial}{\partial G(k)} P^{\alpha,G}(\varepsilon_t=j, X_t=x_t | X_{t-1}=x_{t-1})\Big)_{|_{(x_t,x_{t-1})=(X_t,X_{t-1})}}}{P^{\alpha,G}_{X_{t-1},X_t}}
 \\
& \quad  - \frac{1}{n} \sum\limits_{t=0}^n A_{\alpha,G}h(X_{t-1},X_t) \frac{ \Big(\frac{\partial}{\partial G(k)} P^{\alpha,G}(X_t=x_t | X_{t-1}=x_{t-1})\Big)_{|_{(x_t,x_{t-1})=(X_t,X_{t-1})}}}{P^{\alpha,G}_{X_{t-1},X_t}} - h(k).  
\end{align*}
\end{proof}

\begin{Lemma} \label{Lipschitz_resultsII}
Suppose the Assumptions of Lemma \ref{unif_zws} hold.
Then, for $||\xi_n - \tz || \leq \delta$, we have 
\begin{align} \label{eq:II_final_auxiliary_Lemma}
& \left| \dot{l}_\alpha(X_{t-1},X_t;\tz) - \dot{l}_\alpha(X_{t-1},X_t;\xi_n)  \right| \notag \\
\leq & \delta\left(\frac{2 X_{t-1} + C X_{t-1}^2\left(1+\rho\right)^{X_{t-1}-1}}{P^{\tz}_{X_{t-1}, X_{t}}} + \frac{2 X_{t-1} \left( C X_{t-1}\left(1+\rho\right)^{X_{t-1}-1}+ 1  \right)  }{P^{\tz}_{X_{t-1}, X_{t}} P^{\xi_n}_{X_{t-1}, X_{t}}}\right),
\end{align}
for some (generic) constant $C=C(\delta)$ and some $\rho=\rho(\delta)$, which becomes arbitrarily small for $\delta$ sufficiently small.
\end{Lemma}

\begin{proof}
By plugging-in, we get
\begin{align} \label{eq:anbn}
    & \left| \dot{l}_\alpha(X_{t-1},X_t;\tz) - \dot{l}_\alpha(X_{t-1},X_t;\xi_n) \right| \notag \\
    =& \left| 
    \frac{\sum\limits_{j=0}^{\min(X_{t-1}, X_t)} \binom{X_{t-1}}{j} G_0(X_t-j) (j \alpha_0^{j-1} (1-\alpha_0)^{X_{t-1}-j} - \alpha_0^j(X_{t-1}-j)(1-\alpha_0)^{X_{t-1}-j-1})}{\sum\limits_{j=0}^{\min(X_{t-1}, X_t)} \binom{X_{t-1}}{j} \alpha_0^j(1-\alpha_0)^{X_{t-1}-j}G_0(X_t-j)} \right. \notag \\
    & \left. - \frac{\sum\limits_{j=0}^{\min(X_{t-1}, X_t)} \binom{X_{t-1}}{j} \gx(X_t-j) (j \ax^{j-1} (1-\ax)^{X_{t-1}-j} - \ax^j(X_{t-1}-j)(1-\ax)^{X_{t-1}-j-1})}{\sum\limits_{j=0}^{\min(X_{t-1}, X_t)} \binom{X_{t-1}}{j} \ax^j(1-\ax)^{X_{t-1}-j}\gx(X_t-j)}
    \right| \notag \\
    =: & \left| \frac{a}{b} - \frac{a_n}{b_n} \right| = \left|\frac{a-a_n}{b} + \frac{a_n(b_n-b)}{b b_n} \right| \leq \frac{|a-a_n|}{b} + \frac{|a_n||b_n-b|}{b b_n},
\end{align}
where we used that both $b$ and $b_n$ are transition probabilities with $b,b_n\in[0,1]$. As by Assumption \ref{ass1}, $b>0$ and $\alpha_0 \in (0,1)$ holds and $\ax \rightarrow \az$ by Theorems \ref{cons} and \ref{cons_bs_est}, we also have $b_n>0$ with probability tending to 1. Hence, it remains to consider the numerators $|a-a_n|$ and $|a_n|\cdot|b-b_n|$ separately in the following. First, for $|a-a_n|$, we have
\begin{align} \label{eq:mixed_terms}
    & |a-a_n|  \notag \\
    \leq & \Bigg| \sum\limits_{j=0}^{\min(X_{t-1}, X_t)} \binom{X_{t-1}}{j} G_0(X_t-j) (j \alpha_0^{j-1} (1-\alpha_0)^{X_{t-1}-j} -\alpha_0^j(X_{t-1}-j)(1-\alpha_0)^{X_{t-1}-j-1})  \notag \\
    &    \qquad \qquad \quad -   \binom{X_{t-1}}{j} \gx(X_t-j) (j \alpha_0^{j-1} (1-\alpha_0)^{X_{t-1}-j} -\alpha_0^j(X_{t-1}-j)(1-\alpha_0)^{X_{t-1}-j-1}) \Bigg| \notag \\
    & + \Bigg| \sum\limits_{j=0}^{\min(X_{t-1}, X_t)} \binom{X_{t-1}}{j} \gx(X_t-j) (j \alpha_0^{j-1} (1-\alpha_0)^{X_{t-1}-j} -\alpha_0^j(X_{t-1}-j)(1-\alpha_0)^{X_{t-1}-j-1})  \notag \\
    &  \qquad \qquad \quad -  \binom{X_{t-1}}{j} \gx(X_t-j) (j \ax^{j-1} (1-\ax)^{X_{t-1}-j} -\ax^j(X_{t-1}-j)(1-\ax)^{X_{t-1}-j-1})  \Bigg| \notag \\
    =: & II_{a,1} + II_{a,2}. 
\end{align}
For the first term $II_{a,1}$, we get
\begin{align*}
    II_{a,1} &\leq \sum\limits_{j=0}^{\min(X_{t-1}, X_t)} \binom{X_{t-1}}{j} |\gz(X_t-j)- \gx(X_t-j)| \\ 
    & \qquad \qquad |j \alpha_0^{j-1} (1-\alpha_0)^{X_{t-1}-j} -\alpha_0^j(X_{t-1}-j)(1-\alpha_0)^{X_{t-1}-j-1}| \\
    & \leq \sum\limits_{m=0}^\infty |\gz(m)-\gx(m)| \left( \sum\limits_{j=0}^{\min(X_{t-1}, X_t)} \binom{X_{t-1}}{j} j \alpha_0^{j-1} (1-\alpha_0)^{X_{t-1}-j} \right. \\
    & \qquad \qquad \qquad \qquad \qquad \qquad  \left. + \sum\limits_{j=0}^{\min(X_{t-1}, X_t)} \binom{X_{t-1}}{j} \alpha_0^j(X_{t-1}-j)(1-\alpha_0)^{X_{t-1}-j-1}   \right) \\
    & \leq 2 X_{t-1} \sum\limits_{m=0}^\infty |\gz(m)-\gx(m)|   \\
    \leq & 2 X_{t-1} \delta,
\end{align*}
where we used $\sum\limits_{m=0}^\infty |\gz(m)-\gx(m)|
\leq \|\xi_n-\theta_0\|\leq \delta$ for the last inequality and, for the second last inequality, we made use of the binomial theorem to get
\begin{align} \label{eq:bin_theo1}
    & \sum\limits_{j=0}^{\min(X_{t-1}, X_t)} \binom{X_{t-1}}{j} j \alpha_0^{j-1} (1-\alpha_0)^{X_{t-1}-j} \notag   \\
    &\leq \sum\limits_{j=0}^{X_{t-1}} \binom{X_{t-1}}{j} j \alpha_0^{j-1} (1-\alpha_0)^{X_{t-1}-j} = \sum\limits_{j=1}^{X_{t-1}} \binom{X_{t-1}}{j} j \alpha_0^{j-1} (1-\alpha_0)^{X_{t-1}-j} \notag  \\
    &= \sum\limits_{j=0}^{X_{t-1}-1} \binom{X_{t-1}}{j+1} (j+1) \alpha_0^{j} (1-\alpha_0)^{X_{t-1}-(j+1)} = X_{t-1} \sum\limits_{j=0}^{X_{t-1}-1} \binom{X_{t-1}-1}{j} \alpha_0^{j} (1-\alpha_0)^{X_{t-1}-1-j}  \notag \\
    &= X_{t-1}
\end{align}
as well as
\begin{align} \label{eq:bin_theo2}
   & \quad \sum\limits_{j=0}^{\min(X_{t-1}, X_t)} \binom{X_{t-1}}{j} \alpha_0^j(X_{t-1}-j)(1-\alpha_0)^{X_{t-1}-j-1} \notag \\
   & \leq  \sum\limits_{j=0}^{X_{t-1}} \binom{X_{t-1}}{j} (X_{t-1}-j) \alpha_0^j(1-\alpha_0)^{X_{t-1}-j-1} = \sum\limits_{j=0}^{X_{t-1}-1} \binom{X_{t-1}}{j} (X_{t-1}-j) \alpha_0^j(1-\alpha_0)^{X_{t-1}-j-1} \notag \\
   &= X_{t-1}\sum\limits_{j=0}^{X_{t-1}-1} \binom{X_{t-1}-1}{j} \alpha_0^j(1-\alpha_0)^{X_{t-1}-1-j} = X_{t-1}.
\end{align}
When dealing with the second term $II_{a,2}$, we have
\begin{align} \label{eq:azax}
& \alpha_0^{j-1}(1-\alpha_0)^{X_{t-1}-j} - \ax^{j-1}(1-\ax)^{X_{t-1}-j}  \notag \\
    &= \alpha_0^{j-1}(1-\alpha_0)^{X_{t-1}-j} - \ax^{j-1}(1-\ax)^{X_{t-1}-j} \mp \alpha_0^{j-1}(1-\ax)^{X_{t-1}-j} \notag \\
    &=\alpha_0^{j-1}\Big((1-\alpha_0)^{X_{t-1}-j} - (1-\ax)^{X_{t-1}-j}\Big) + (1-\ax)^{X_{t-1}-j}\Big(\alpha_0^{j-1}-\ax^{j-1}\Big).
\end{align}
Further, for the last expression in brackets, we have
\begin{align} \label{eq:az-ax}
\az^{j-1}-\ax^{j-1}
    &=\az^{j-1}-\ax^{j-1} \mp \az^{j-2}\ax  \notag\\
    &= \az^{j-2}(\az-\ax) + \az^{j-2}\ax - \ax^{j-1} \mp \az^{j-3}\ax^2 \notag \\&= \az^{j-2}(\az-\ax) + \az^{j-3}\ax(\az-\ax) + \az^{j-3}\ax^2
-\ax^{j-1} \mp \az^{j-4}\ax^3 \notag \\
&= \ldots \notag\\
&= (\az -\ax) \sum\limits_{l=0}^{j-2} \az^{j-2-l}\ax^l
\end{align}
and, analogously, 
\begin{align} \label{eq:1-azax}
    (1-\alpha_0)^{X_{t-1}-j} - (1-\ax)^{X_{t-1}-j} = (\ax-\az)\sum\limits_{l=0}^{X_{t-1}-j-1} (1-\az)^{X_{t-1}-j-1-l}(1-\ax)^l.
\end{align}
Consequently, plugging-in \eqref{eq:az-ax} and \eqref{eq:1-azax} in \eqref{eq:azax} leads to
\begin{align*}
& \alpha_0^{j-1}(1-\alpha_0)^{X_{t-1}-j} - \ax^{j-1}(1-\ax)^{X_{t-1}-j} \\
&= \az^{j-1}(\ax-\az) \sum\limits_{l=0}^{X_{t-1}-j-1} (1-\az)^{X_{t-1}-j-1-l}(1-\ax)^l + (1-\ax)^{X_{t-1}-j}(\az-\ax) \sum\limits_{l=0}^{j-2}\az^{j-2-l}\ax^l.
\end{align*}
Following the same procedure, we get
\begin{align*} & \quad \ax^j(1-\ax)^{X_{t-1}-j-1}-\az^j(1-\az)^{X_{t-1}-j-1} \\
&= \ax^j(\az-\ax) \sum\limits_{l=0}^{X_{t-1}-j-2}(1-\ax)^{X_{t-1}-j-2-l}(1-\az)^l + (1-\az)^{X_{t-1}-j-1}(\ax-\az) \sum\limits_{l=0}^{j-1}\ax^{j-1-l} \az^l.
\end{align*}
Altogether, we obtain
\begin{align*}
    II_{a,2}
    &= \left| \sum\limits_{j=0}^{\min(X_{t-1}, X_t)} \binom{X_{t-1}}{j} \gx(X_t-j) \left(j \alpha_0^{j-1}(1-\alpha_0)^{X_{t-1}-j} - j \ax^{j-1}(1-\ax)^{X_{t-1}-j} \right. \right. \\
    & \left. \left. \qquad \qquad \qquad \qquad + \ax^j(X_{t-1}-j)(1-\ax)^{X_{t-1}-j-1} - \alpha_0^j(X_{t-1}-j)(1-\alpha_0)^{X_{t-1}-j-1}  \right)   \right| \\
    & \leq |\az-\ax| \sum\limits_{j=0}^{\min(X_{t-1}, X_t)} \left[ \binom{X_{t-1}}{j} \gx(X_t-j) \left( j
    \left( \az^{j-1} \sum\limits_{l=0}^{X_{t-1}-j-1} (1-\az)^{X_{t-1}-j-1-l}(1-\ax)^l \right. \right. \right. \\
    & \quad\left. \left. \left. + (1-\ax)^{X_{t-1}-j} \sum\limits_{l=0}^{j-2}\az^{j-2-l}\ax^l  \right) + (X_{t-1}-j) \left(  \ax^j \sum\limits_{l=0}^{X_{t-1}-j-2}(1-\ax)^{X_{t-1}-j-2-l}(1-\az)^l \right. \right. \right. \\
    & \quad\left. \left. \left. + (1-\az)^{X_{t-1}-j-1}\sum\limits_{l=0}^{j-1}\ax^{j-1-l} \az^l  \right) \right)\right] \\
    &=: |\az-\ax|\big(II_{a,2,1}+II_{a,2,2}+II_{a,2,3}+II_{a,2,4}\big)
\end{align*}
with an obvious notation for $II_{a,2,1},II_{a,2,2},II_{a,2,3}$ and $II_{a,2,4}$ according to the four terms on the last right-hand side. Let us consider $II_{a,2,1}$ in more detail. Making use of $|\az-\ax|\leq \|\xi_n-\theta_0\|\leq \delta$, we have
\begin{align} \label{4terms}
& \sum\limits_{j=0}^{\min(X_{t-1}, X_t)}  \binom{X_{t-1}}{j} \gx(X_t-j) j\az^{j-1} \sum\limits_{l=0}^{X_{t-1}-j-1} (1-\az)^{X_{t-1}-j-1-l}(1-\ax)^l   \\
=& \sum\limits_{j=0}^{\min(X_{t-1}, X_t)}  \binom{X_{t-1}}{j} \gx(X_t-j) j\az^{j-1} (1-\az)^{X_{t-1}-j-1} \sum\limits_{l=0}^{X_{t-1}-j-1} (1-\az)^{-l}(1-\ax)^l \notag  \\
=& \sum\limits_{j=0}^{\min(X_{t-1}, X_t)}  \binom{X_{t-1}}{j} \gx(X_t-j) j\az^{j-1} (1-\az)^{X_{t-1}-j-1} \sum\limits_{l=0}^{X_{t-1}-j-1} \left(\frac{1-\ax}{1-\az}\right)^l  \notag \\
=& \sum\limits_{j=0}^{\min(X_{t-1}, X_t)}  \binom{X_{t-1}}{j} \gx(X_t-j) j\az^{j-1} (1-\az)^{X_{t-1}-j-1} \sum\limits_{l=0}^{X_{t-1}-j-1} \left(1+\frac{\az-\ax}{1-\az}\right)^l \notag  \\
\leq& \frac{1}{1-\az}\sum\limits_{j=0}^{\min(X_{t-1}, X_t)}  \binom{X_{t-1}}{j} j\az^{j-1} (1-\az)^{X_{t-1}-j} \sum\limits_{l=0}^{X_{t-1}-j-1} \left(1+\frac{|\az-\ax|}{1-\az}\right)^l \notag \\
\leq & \frac{1}{1-\az}\left(\sum\limits_{j=0}^{\min(X_{t-1}, X_t)}  \binom{X_{t-1}}{j} j\az^{j-1} (1-\az)^{X_{t-1}-j}\right) \left(\sum\limits_{l=0}^{X_{t-1}-1} \left(1+\frac{\delta}{1-\az}\right)^l\right) \notag  \\
=& \frac{1}{1-\az}X_{t-1}^2\left(1+\frac{\delta}{1-\az}\right)^{X_{t-1}-1}. \notag
\end{align}
Using the same steps, we get for the three other terms 
\begin{align*}
    II_{a,2,2} & \leq \frac{1}{\ax}X_{t-1}^2\left(1+\frac{\delta}{\ax}\right)^{X_{t-1}-1}, \\
    II_{a,2,3} & \leq \frac{1}{1-\ax}X_{t-1}^2\left(1+\frac{\delta}{1-\ax}\right)^{X_{t-1}-1} \quad  \text{and}\\
    II_{a,2,4} & \leq \frac{1}{\az}X_{t-1}^2\left(1+\frac{\delta}{\az}\right)^{X_{t-1}-1},
\end{align*}
where we used that, e.g., $\sum_{l=0}^{j-2} \az^{j-2-l}\ax^l=\sum_{l=0}^{j-2} \ax^{j-2-l}\az^l$ and $\min(X_t, X_{t-1}) \leq X_{t-1}$. Altogether, using again $|\az-\ax|\leq \|\xi_n-\theta_0\|\leq \delta$, this leads to
\begin{align*}
II_{a,2} \leq C|\az-\ax| X_{t-1}^2\left(1+\rho\right)^{X_{t-1}-1}\leq C\delta X_{t-1}^2\left(1+\rho\right)^{X_{t-1}-1}
\end{align*}
for some (generic) constant $C=C(\delta)$ and some $\rho=\rho(\delta)$ which becomes arbitrarily small for $\delta$ sufficiently small.
Now, consider the second term of \eqref{eq:anbn}. We have 
\begin{align} \label{abs_an}
    |a_n| &= \left| \sum\limits_{j=0}^{\min(X_{t-1}, X_t)} \binom{X_{t-1}}{j} \gx(X_t-j) \right.    \notag \\
    & \qquad \qquad \qquad \left. \left(j \ax^{j-1} (1-\ax)^{X_{t-1}-j} - \ax^j(X_{t-1}-j)(1-\ax)^{X_{t-1}-j-1} \right) \right| \notag \\
    \leq & \sum\limits_{j=0}^{\min(X_{t-1}, X_t)} \binom{X_{t-1}}{j} j \ax^{j-1} (1-\ax)^{X_{t-1}-j} + \sum\limits_{j=0}^{\min(X_{t-1}, X_t)} \binom{X_{t-1}}{j} \ax^j(X_{t-1}-j)(1-\ax)^{X_{t-1}-j-1} \notag   \\
    \leq& 2 X_{t-1},
\end{align}
where we used that $G_{\xi_n}(k)\leq 1$ for all $k\in\mathbb{N}_0$ and the bounds obtained in \eqref{eq:bin_theo1} and \eqref{eq:bin_theo2}. Further, we have
\begin{align*}
    |b_n-b|  \leq & \left|  \sum\limits_{j=0}^{\min(X_{t-1}, X_t)} \binom{X_{t-1}}{j} \gx(X_t-j) \left(  \ax^j(1-\ax)^{X_{t-1}-j} - \az^j(1-\az)^{X_{t-1}-j} \right)  \right| \\
    & + \left|  \sum\limits_{j=0}^{\min(X_{t-1}, X_t)} \binom{X_{t-1}}{j} \big(\gx(X_t-j)-\gz(X_t-j)\big) \left(  \az^j(1-\az)^{X_{t-1}-j} \right) \right| \\
     =: & II_{b,1} + II_{b,2}.
\end{align*}
We can proceed analogously as we did for $II_{a,1}$ and $II_{a,2}$ to get
\begin{align*}
    II_{b,1} &= \left|  \sum\limits_{j=0}^{\min(X_{t-1}, X_t)} \binom{X_{t-1}}{j} \gx(X_t-j) \left( \ax^j(\az-\ax) \sum\limits_{l=0}^{X_{t-1}-j-1} (1-\ax)^{X_{t-1}-j-1-l}(1-\az)^l \right. \right. \\
    & \left. \left. \qquad \qquad \qquad\qquad \qquad \qquad \qquad \qquad  + (1-\az)^{X_{t-1}-j} (\ax-\az) \sum\limits_{l=0}^{j-1} \ax^{j-1-l} \az^l \right) \right| \\
    & \leq |\az-\ax| \sum\limits_{j=0}^{\min(X_{t-1}, X_t)} \binom{X_{t-1}}{j} \left( \ax^j \sum\limits_{l=0}^{X_{t-1}-j-1} (1-\ax)^{X_{t-1}-j-1-l}(1-\az)^l \right. 
    \\ & \left.\qquad \qquad \qquad\qquad \qquad \qquad \qquad \qquad   + (1-\az)^{X_{t-1}-j} \sum\limits_{l=0}^{j-1} \ax^{j-1-l} \az^l  \right) \\
    &\leq C|\az-\ax| X_{t-1}\left(1+\rho\right)^{X_{t-1}-1} \\
    &\leq  C\delta X_{t-1}\left(1+\rho\right)^{X_{t-1}-1}
\end{align*}
and 
\begin{align*}
    II_{b,2} & \leq \sum\limits_{j=0}^{\min(X_{t-1}, X_t)} \binom{X_{t-1}}{j} | \gx(X_t-j) - \gz(X_t-j)| (\az^j(1-\az)^{X_{t-1}-j}) \\
    & \leq \sum\limits_{m=0}^\infty |\gx(m) - \gz(m)| \sum\limits_{j=0}^{X_{t-1}}\binom{X_{t-1}}{j} \az^j(1-\az)^{X_{t-1}-j} \\
    &= \sum\limits_{m=0}^\infty |\gx(m) - \gz(m)|   \\
    &\leq \delta,
\end{align*}
where $C$ and $\rho$ are as above. Altogether, this completes the proof. 
\end{proof}

\begin{Lemma} \label{Lipschitz_resultsI}
Suppose the Assumptions of Lemma \ref{unif_zws} hold.
Then, for $||\xi_n - \tz || \leq \delta$, we have 
\begin{align} \label{eq:I_final_auxiliary_Lemma}
& \left|\frac{\frac{\partial^2}{\partial \alpha^2} P_{X_{t-1},X_t |\alpha=\ax}^{\alpha,\gx}}{P_{X_{t-1},X_t}^{\xi_n}} - \frac{\frac{\partial^2}{\partial \alpha^2} P_{X_{t-1},X_t |\alpha=\alpha_0}^{\alpha,\gz}}{P_{X_{t-1},X_t}^{\tz}}\right| \notag  \\
\leq & \delta\left(\frac{\widetilde C X_{t-1}^3(1+\widetilde\rho)^{X_{t-1}-1}+\widetilde C X_{t-1}^2}{P^{\xi_n}_{X_{t-1}, X_t}} +  \frac{\widetilde C X_{t-1}^2 \left( \widetilde C  X_{t-1} (1+\widetilde \rho)^{X_{t-1}-1} + 1 \right)}{P^{\xi_n}_{X_{t-1}, X_t} P^{\tz}_{X_{t-1}, X_t}}\right)     
\end{align}
for some (generic) constant $\widetilde C=\widetilde C(\delta)$ and some $\widetilde \rho=\widetilde \rho(\delta)$, which becomes arbitrarily small for $\delta$ sufficiently small.
\end{Lemma}

\begin{proof}
We follow the proof technique of Lemma \ref{Lipschitz_resultsII}, but have to deal with the second derivative. First, we get the bound
\begin{align*}
\left|\frac{\frac{\partial^2}{\partial \alpha^2} P_{X_{t-1},X_t |\alpha=\ax}^{\alpha,\gx}}{P_{X_{t-1},X_t}^{\xi_n}} - \frac{\frac{\partial^2}{\partial \alpha^2} P_{X_{t-1},X_t |\alpha=\alpha_0}^{\alpha,\gz}}{P_{X_{t-1},X_t}^{\tz}}\right| =: \left| \frac{c_n}{b_n} - \frac{c}{b}  \right| \leq \frac{|c_n-c|}{b_n} + \frac{|c||b-b_n|}{b_n b},
\end{align*}
where $b_n$ and $b$ are defined as in \eqref{eq:anbn} and
\begin{align*}
    c_n = & \sum\limits_{j=0}^{\min(X_{t-1}, X_t)} \binom{X_{t-1}}{j} \gx(X_t-j) \ax^{j-2} (1-\ax)^{X_{t-1}-j-2}    \\
    & \qquad\qquad\qquad\qquad\qquad\qquad\bigg(\ax^2(X_{t-1}-1)X_{t-1}+j(-2\ax(X_{t-1}-1)-1) +j^2  \bigg),   \\
    c = & \sum\limits_{j=0}^{\min(X_{t-1}, X_t)} \binom{X_{t-1}}{j} \gz(X_t-j) \az^{j-2} (1-\az)^{X_{t-1}-j-2}    \\
    & \qquad\qquad\qquad\qquad\qquad\qquad \bigg( \az^2(X_{t-1}-1)X_{t-1}+j(-2\az(X_{t-1}-1)-1) +j^2  \bigg).
\end{align*}
Hence, it remains to investigate $|c_n-c|$ and $|c|$. First, we note that $c$ (and analogously $c_n$) can equivalently be written as 
\begin{align*} & \sum\limits_{j=0}^{\min(X_{t-1}, X_t)} \binom{X_{t-1}}{j} \gz(X_t-j)\left( \az^j(1-\az)^{X_{t-1}-j-2}(X_{t-1}-1)X_{t-1} \right. \\
& \qquad\left. - \az^{j-1}(1-\az)^{X_{t-1}-j-2}2j(X_{t-1}-1)-\az^{j-2}(1-\az)^{X_{t-1}-j-2}(j-j^2)   \right).    
\end{align*}
Adding and subtracting the mixed terms as we did in \eqref{eq:mixed_terms}, we obtain
\begin{align*}
   |c_n-c|  \leq & \left|   \sum\limits_{j=0}^{\min(X_{t-1}, X_t)} \binom{X_{t-1}}{j} \gx(X_t-j) \left( \ax^j(1-\ax)^{X_{t-1}-j-2}(X_{t-1}-1)X_{t-1} \right. \right. \\
    & \left. \left. \qquad  - \az^j(1-\az)^{X_{t-1}-j-2}(X_{t-1}-1)X_{t-1}  + \az^{j-1}(1-\az)^{X_{t-1}-j-2}2j(X_{t-1}-1)  \right. \right. \\
     & \left. \left. \qquad  - \ax^{j-1}(1-\ax)^{X_{t-1}-j-2}2j(X_{t-1}-1)  + \az^{j-2}(1-\az)^{X_{t-1}-j-2}(j-j^2)  \right. \right. \\
     & \left. \left. \qquad  -\ax^{j-2}(1-\ax)^{X_{t-1}-j-2}(j-j^2)   \right)  \right| \\
     & + \left|  \sum\limits_{j=0}^{\min(X_{t-1}, X_t)} \binom{X_{t-1}}{j} \big(\gx(X_t-j)   - \gz(X_t-j)\big) \left( \az^j(1-\az)^{X_{t-1}-j-2}(X_{t-1}-1)X_{t-1} \right. \right. \\
     & \left. \left.  \qquad \qquad - \az^{j-1}(1-\az)^{X_{t-1}-j-2}2j(X_{t-1}-1) -\az^{j-2}(1-\az)^{X_{t-1}-j-2}(j-j^2)   \right)    \right| \\
     := & I_{c,1} + I_{c,2}.
\end{align*}
For $I_{c,1}$, we rewrite the last factor of the sum analogously to \eqref{eq:azax}, \eqref{eq:az-ax} and \eqref{eq:1-azax} to get
\begin{align*}
    & I_{c,1} \\ 
    \leq & |\az-\ax|  \sum\limits_{j=0}^{\min(X_{t-1}, X_t)} \binom{X_{t-1}}{j}
    \Bigg[   (X_{t-1}-1)X_{t-1} \left( \ax^j\sum\limits_{l=0}^{X_{t-1}-j-3} (1-\ax)^{X_{t-1}-j-3-l}(1-\az)^l \right. \\
    & \left. + (1-\az)^{X_{t-1}-j-2} \sum\limits_{l=0}^{j-1} \ax^{j-1-l} \az^l \right)
    + 2j(X_{t-1}-1) \left( \az^{j-1} \sum\limits_{l=0}^{X_{t-1}-j-3} (1-\az)^{X_{t-1}-j-3-l} (1-\ax)^l \right.  \\
    & \left. + (1-\ax)^{X_{t-1}-j-2} \sum\limits_{l=0}^{j-1} \az^{j-1-l}\ax^l  \right)
    +(j-j^2) \left(  \az^{j-2} \sum\limits_{l=0}^{X_{t-1}-j-3}(1-\az)^{X_{t-1}-j-3-l}(1-\ax)^l \right.  \\
    &  \left. + (1-\ax)^{X_{t-1}-j-2} \sum\limits_{l=0}^{j-3} \az^{j-3-l} \ax^l  \right)  \Bigg] \\
    =: & |\az-\ax| \big(I_{c,1,1} + I_{c,1,2} + I_{c,1,3} + I_{c,1,4} + I_{c,1,5} + I_{c,1,6}\big)
\end{align*}
with an obvious notation for $I_{c,1,1},I_{c,1,2}, I_{c,1,3}, I_{c,1,4}, I_{c,1,5}$ and $ I_{c,1,6}$ according to the six terms on the last right-hand side. With similar steps as in \eqref{4terms}, we get
\begin{align*}
    I_{c,1,1} &\leq \frac{1}{(1-\ax)^3} X_{t-1}^2 (X_{t-1}-1) \left(1+\frac{\delta}{1-\ax}\right)^{X_{t-1}-1}, \\
    I_{c,1,2} &\leq \frac{1}{\az(1-\az)^2} X_{t-1}^2 (X_{t-1}-1) \left(1+\frac{\delta}{\az}\right)^{X_{t-1}-1}, \\
    I_{c,1,3} &\leq \frac{1}{(1-\az)^3} 2X_{t-1}^2 (X_{t-1}-1) \left(1+\frac{\delta}{1-\az}\right)^{X_{t-1}-1}, \\
    I_{c,1,4} &\leq \frac{1}{(1-\ax)^2} 2X_{t-1}^2 (X_{t-1}-1) \left(1+\frac{\delta}{\ax}\right)^{X_{t-1}-1}, \\
    I_{c,1,5} &\leq \frac{1}{(1-\az)^3} X_{t-1}^2 (X_{t-1}-1) \left(1+\frac{\delta}{1-\az}\right)^{X_{t-1}-1} \quad \text{and} \\
    I_{c,1,6} &\leq \frac{1}{\ax(1-\ax)^2} X_{t-1}^2 (X_{t-1}-1) \left(1+\frac{\delta}{\ax}\right)^{X_{t-1}-1}.
\end{align*}
Altogether, this leads to 
\[I_{c,1} \leq \widetilde C |\az-\ax|  X_{t-1}^2(X_{t-1}-1) (1+\widetilde\rho)^{X_{t-1}-1} \leq \widetilde C \delta  X_{t-1}^3 (1+\widetilde\rho)^{X_{t-1}-1}\]
for some (generic) constant $ \widetilde C=\widetilde C(\delta)$ and some $\widetilde\rho=\widetilde\rho(\delta)$ which becomes arbitrary small for $\delta$ sufficiently small.
For $I_{c,2}$, we re-use the binomial theorem as in \eqref{eq:bin_theo1} and \eqref{eq:bin_theo2} and get
\begin{align} \label{eq:viii}
    I_{c,2} & \leq \sum\limits_{j=0}^{\min(X_{t-1}, X_t)} \binom{X_{t-1}}{j} |\gx(X_t-j) - \gz(X_t-j)| (\az^j(1-\az)^{X_{t-1}-j-2}(X_{t-1}-1)X_{t-1} \notag \\
    &  \qquad \qquad - \az^{j-1}(1-\az)^{X_{t-1}-j-2}2j(X_{t-1}-1)-\az^{j-2}(1-\az)^{X_{t-1}-j-2}(j-j^2)) \notag \\
    & \leq \sum\limits_{m=0}^\infty |\gx(m) - \gz(m)| 4 X_{t-1} (X_{t-1}-1) (1-\az)^{-2}    \\
    & \leq \widetilde C\delta  X_{t-1}^2
\end{align}
As last term, we have to consider $|c|$ for which we get
\begin{align} \label{abs_c}
    |c| \leq 4 X_{t-1}(X_{t-1}-1)(1-\az)^{-2} \leq \widetilde C X_{t-1}^2
\end{align}
with the same arguments as in \eqref{eq:viii}. Altogether, using $||\xi_n - \tz || \leq \delta$, this completes the proof.
\end{proof}

\newpage

\section{Additional Tables}\label{sec:appendixC}

\begin{table}[h]
\centering
\begin{tabular}{c|ccc|ccc|ccc}
& \multicolumn{3}{|c|}{$\alpha$} & \multicolumn{3}{|c|}{$G(0)$} & \multicolumn{3}{|c}{$G(1)$} \\
$n$ & 100 & 500 & 1000 & 100 & 500 & 1000& 100 & 500 & 1000 \\
\hline
coverage &0.520&0.806&0.914&0.886&0.918&0.946&0.960&0.948&0.962 \\
average length &0.264&0.163&0.125&0.242&0.115&0.083&0.223&0.098&0.069 \\
& \multicolumn{3}{|c|}{$G(2)$} & \multicolumn{3}{|c|}{$G(3)$} & \multicolumn{3}{|c}{$G(4)$} \\
$n$ & 100 & 500 & 1000 & 100 & 500 & 1000& 100 & 500 & 1000 \\
\hline
coverage &0.918&0.938&0.934&0.820&0.934&0.934&0.652&0.834&0.942 \\
average length &0.191&0.087&0.062&0.113&0.055&0.039&0.047&0.027&0.019 \\
\end{tabular}
\vspace{0.5cm}
\caption{Coverage and average length of the semi-parametrically constructed bootstrap confidence intervals based on the semi-parametric estimation of \citet{drost} and the semi-parametric INAR bootstrap from Section \ref{sec:boot_algo} for $\alpha, G(0), \ldots, G(4)$ in case of a Poi(1)-INAR(1) DGP with $\alpha=0.1$ for different sample sizes.}
\label{tab:poi_sp_1_01}
\end{table}

\begin{table}[h]
\centering
\begin{tabular}{c|ccc|ccc|ccc}
& \multicolumn{3}{|c|}{$\alpha$} & \multicolumn{3}{|c|}{$G(0)$} & \multicolumn{3}{|c}{$G(1)$} \\
$n$ & 100 & 500 & 1000 & 100 & 500 & 1000& 100 & 500 & 1000 \\
\hline
coverage &0.652&0.834&0.930&0.876&0.916&0.944&0.974&0.976&0.986 \\
average length &0.288&0.164&0.124&0.195&0.093&0.067&0.026&0.006&0.003 \\
& \multicolumn{3}{|c|}{$G(2)$} & \multicolumn{3}{|c|}{$G(3)$} & \multicolumn{3}{|c}{$G(4)$} \\
$n$ & 100 & 500 & 1000 & 100 & 500 & 1000& 100 & 500 & 1000 \\
\hline
coverage &0.848&0.908&0.948&0.864&0.898&0.948&0.846&0.872&0.944 \\
average length &0.094&0.0046&0.033&0.061&0.031&0.022&0.024&0.012&0.008 \\
\end{tabular}
\vspace{0.5cm}
\caption{Coverage and average length of the parametrically constructed bootstrap confidence intervals based on \emph{parametric} ML estimation and a \emph{parametric} Poi-INAR bootstrap for $\alpha, G(0), \ldots, G(4)$ in case of a Poi(1)-INAR(1) DGP with $\alpha=0.1$ for different sample sizes.}
\label{tab:poi_p_1_01}
\end{table}

\begin{table}[h]
\centering
\begin{tabular}{c|ccc|ccc|ccc}
& \multicolumn{3}{|c|}{$\alpha$} & \multicolumn{3}{|c|}{$G(0)$} & \multicolumn{3}{|c}{$G(1)$} \\
$n$ & 100 & 500 & 1000 & 100 & 500 & 1000& 100 & 500 & 1000 \\
\hline
coverage &0.410&0.750&0.866&0.754&0.892&0.916&0.856&0.864&0.912 \\
average length &0.224&0.153&0.123&0.110&0.052&0.038&0.199&0.094&0.069 \\
& \multicolumn{3}{|c|}{$G(2)$} & \multicolumn{3}{|c|}{$G(3)$} & \multicolumn{3}{|c}{$G(4)$} \\
$n$ & 100 & 500 & 1000 & 100 & 500 & 1000& 100 & 500 & 1000 \\
\hline
coverage &0.898&0.944&0.942&0.918&0.942&0.942&0.878&0.958&0.940 \\
average length &0.236&0.107&0.076&0.242&0.107&0.076&0.219&0.104&0.074 \\
\end{tabular}
\vspace{0.5cm}
\caption{Coverage and average length of the semi-parametrically constructed bootstrap confidence intervals based on the semi-parametric estimation of \citet{drost} and the semi-parametric INAR bootstrap from Section \ref{sec:boot_algo} for $\alpha, G(0), \ldots, G(4)$ in case of a Poi(3)-INAR(1) DGP with $\alpha=0.1$ for different sample sizes.}
\label{tab:poi_sp_3_01}
\end{table}

\begin{table}[h]
\centering
\begin{tabular}{c|ccc|ccc|ccc}
& \multicolumn{3}{|c|}{$\alpha$} & \multicolumn{3}{|c|}{$G(0)$} & \multicolumn{3}{|c}{$G(1)$} \\
$n$ & 100 & 500 & 1000 & 100 & 500 & 1000& 100 & 500 & 1000 \\
\hline
coverage &0.596&0.842&0.936&0.774&0.896&0.948&0.768&0.910&0.942\\
average length &0.283&0.162&0.122&0.073&0.032&0.023&0.126&0.062&0.045 \\
& \multicolumn{3}{|c|}{$G(2)$} & \multicolumn{3}{|c|}{$G(3)$} & \multicolumn{3}{|c}{$G(4)$} \\
$n$ & 100 & 500 & 1000 & 100 & 500 & 1000& 100 & 500 & 1000 \\
\hline
coverage &0.694&0.880&0.960&0.970&0.978&0.960&0.704&0.866&0.944 \\
average length &0.080&0.045&0.034&0.028&0.007&0.004&0.071&0.034&0.026 \\
\end{tabular}
\vspace{0.5cm}
\caption{Coverage and average length of the parametrically constructed bootstrap confidence intervals based on \emph{parametric} ML estimation and a \emph{parametric} Poi-INAR bootstrap for $\alpha, G(0), \ldots, G(4)$ in case of a Poi(3)-INAR(1) DGP with $\alpha=0.1$ for different sample sizes.}
\label{tab:poi_p_3_01}
\end{table}

\begin{table}[h]
\centering
\begin{tabular}{c|ccc|ccc|ccc}
& \multicolumn{3}{|c|}{$\alpha$} & \multicolumn{3}{|c|}{$G(0)$} & \multicolumn{3}{|c}{$G(1)$} \\
$n$ & 100 & 500 & 1000 & 100 & 500 & 1000& 100 & 500 & 1000 \\
\hline
coverage &0.804&0.952&0.944&0.828&0.928&0.934&0.914&0.938&0.924 \\
average length &0.379&0.171&0.120&0.305&0.141&0.100&0.287&0.129&0.091 \\
& \multicolumn{3}{|c|}{$G(2)$} & \multicolumn{3}{|c|}{$G(3)$} & \multicolumn{3}{|c}{$G(4)$} \\
$n$ & 100 & 500 & 1000 & 100 & 500 & 1000& 100 & 500 & 1000 \\
\hline
coverage &0.878&0.928&0.944&0.732&0.936&0.944&0.486&0.782&0.898 \\
average length &0.253&0.112&0.078&0.143&0.069&0.048&0.058&0.031&0.024 \\
\end{tabular}
\vspace{0.5cm}
\caption{Coverage and average length of the semi-parametrically constructed bootstrap confidence intervals based on the semi-parametric estimation of \citet{drost} and the semi-parametric INAR bootstrap from Section \ref{sec:boot_algo} for $\alpha, G(0), \ldots, G(4)$ in case of a Poi(1)-INAR(1) DGP with $\alpha=0.3$ for different sample sizes.}
\label{tab:poi_sp_1_03}
\end{table}

\begin{table}[h]
\centering
\begin{tabular}{c|ccc|ccc|ccc}
& \multicolumn{3}{|c|}{$\alpha$} & \multicolumn{3}{|c|}{$G(0)$} & \multicolumn{3}{|c}{$G(1)$} \\
$n$ & 100 & 500 & 1000 & 100 & 500 & 1000& 100 & 500 & 1000 \\
\hline
coverage &0.890&0.952&0.954&0.926&0.932&0.954&0.984&0.978&0.982 \\
average length &0.363&0.162&0.114&0.222&0.100&0.071&0.034&0.007&0.004 \\
& \multicolumn{3}{|c|}{$G(2)$} & \multicolumn{3}{|c|}{$G(3)$} & \multicolumn{3}{|c}{$G(4)$} \\
$n$ & 100 & 500 & 1000 & 100 & 500 & 1000& 100 & 500 & 1000 \\
\hline
coverage &0.878&0.924&0.938&0.872&0.924&0.938&0.836&0.916&0.940 \\
average length &0.105&0.050&0.036&0.074&0.034&0.024&0.031&0.013&0.009 \\
\end{tabular}
\vspace{0.5cm}
\caption{Coverage and average length of the parametrically constructed bootstrap confidence intervals based on \emph{parametric} ML estimation and a \emph{parametric} Poi-INAR bootstrap for $\alpha, G(0), \ldots, G(4)$ in case of a Poi(1)-INAR(1) DGP with $\alpha=0.3$ for different sample sizes.}
\label{tab:poi_p_1_03}
\end{table}

\begin{table}[h]
\centering
\begin{tabular}{c|ccc|ccc|ccc}
& \multicolumn{3}{|c|}{$\alpha$} & \multicolumn{3}{|c|}{$G(0)$} & \multicolumn{3}{|c}{$G(1)$} \\
$n$ & 100 & 500 & 1000 & 100 & 500 & 1000& 100 & 500 & 1000 \\
\hline
coverage &0.698&0.966&0.946&0.566&0.872&0.912&0.596&0.890&0.918\\
average length &0.352&0.188&0.128&0.123&0.079&0.058&0.258&0.168&0.122 \\
& \multicolumn{3}{|c|}{$G(2)$} & \multicolumn{3}{|c|}{$G(3)$} & \multicolumn{3}{|c}{$G(4)$} \\
$n$ & 100 & 500 & 1000 & 100 & 500 & 1000& 100 & 500 & 1000 \\
\hline
coverage &0.688&0.884&0.932&0.750&0.902&0.932&0.686&0.836&0.932 \\
average length &0.324&0.211&0.155&0.340&0.231&0.172&0.311&0.224&0.170 \\
\end{tabular}
\vspace{0.5cm}
\caption{Coverage and average length of the semi-parametrically constructed bootstrap confidence intervals based on the semi-parametric estimation of \citet{drost} and the semi-parametric INAR bootstrap from Section \ref{sec:boot_algo} for $\alpha, G(0), \ldots, G(4)$ in case of a Poi(3)-INAR(1) DGP with $\alpha=0.3$ for different sample sizes.}
\label{tab:poi_sp_3_03}
\end{table}

\begin{table}[h]
\centering
\begin{tabular}{c|ccc|ccc|ccc}
& \multicolumn{3}{|c|}{$\alpha$} & \multicolumn{3}{|c|}{$G(0)$} & \multicolumn{3}{|c}{$G(1)$} \\
$n$ & 100 & 500 & 1000 & 100 & 500 & 1000& 100 & 500 & 1000 \\
\hline
coverage &0.904&0.946&0.958&0.842&0.926&0.952&0.872&0.934&0.946\\
average length &0.357&0.157&0.111&0.083&0.036&0.026&0.151&0.070&0.051\\
& \multicolumn{3}{|c|}{$G(2)$} & \multicolumn{3}{|c|}{$G(3)$} & \multicolumn{3}{|c}{$G(4)$} \\
$n$ & 100 & 500 & 1000 & 100 & 500 & 1000& 100 & 500 & 1000 \\
\hline
coverage &0.796&0.902&0.982&0.978&0.980&0.982&0.686&0.882&0.934 \\
average length &0.110&0.053&0.038&0.045&0.010&0.007&0.075&0.038&0.029 \\
\end{tabular}
\vspace{0.5cm}
\caption{Coverage and average length of the parametrically constructed bootstrap confidence intervals based on \emph{parametric} ML estimation and a \emph{parametric} Poi-INAR bootstrap for $\alpha, G(0), \ldots, G(4)$ in case of a Poi(3)-INAR(1) DGP with $\alpha=0.3$ for different sample sizes.}
\label{tab:poi_p_3_03}
\end{table}

\begin{table}[h]
\centering
\begin{tabular}{c|ccc|ccc|ccc}
& \multicolumn{3}{|c|}{$\alpha$} & \multicolumn{3}{|c|}{$G(0)$} & \multicolumn{3}{|c}{$G(1)$} \\
$n$ & 100 & 500 & 1000 & 100 & 500 & 1000& 100 & 500 & 1000 \\
\hline
coverage &0.896&0.964&0.952&0.408&0.684&0.818&0.550&0.794&0.832 \\
average length &0.422&0.213&0.172&0.200&0.165&0.168&0.325&0.254&0.209 \\
& \multicolumn{3}{|c|}{$G(2)$} & \multicolumn{3}{|c|}{$G(3)$} & \multicolumn{3}{|c}{$G(4)$} \\
$n$ & 100 & 500 & 1000 & 100 & 500 & 1000& 100 & 500 & 1000 \\
\hline
coverage &0.592&0.812&0.704&0.692&0.764&0.704&0.664&0.734&0.724 \\
average length &0.385&0.293&0.244&0.423&0.299&0.253&0.389&0.287&0.237 \\
\end{tabular}
\vspace{0.5cm}
\caption{Coverage and average length of the semi-parametrically constructed bootstrap confidence intervals based on the semi-parametric estimation of \citet{drost} and the semi-parametric INAR bootstrap from Section \ref{sec:boot_algo} for $\alpha, G(0), \ldots, G(4)$ in case of a Poi(3)-INAR(1) DGP with $\alpha=0.5$ for different sample sizes.}
\label{tab:poi_sp_3_05}
\end{table}

\begin{table}[h]
\centering
\begin{tabular}{c|ccc|ccc|ccc}
& \multicolumn{3}{|c|}{$\alpha$} & \multicolumn{3}{|c|}{$G(0)$} & \multicolumn{3}{|c}{$G(1)$} \\
$n$ & 100 & 500 & 1000 & 100 & 500 & 1000& 100 & 500 & 1000 \\
\hline
coverage &0.952&0.962&0.960&0.832&0.924&0.930&0.872&0.932&0.952 \\
average length &0.288&0.122&0.087&0.088&0.038&0.030&0.159&0.075&0.056 \\
& \multicolumn{3}{|c|}{$G(2)$} & \multicolumn{3}{|c|}{$G(3)$} & \multicolumn{3}{|c}{$G(4)$} \\
$n$ & 100 & 500 & 1000 & 100 & 500 & 1000& 100 & 500 & 1000 \\
\hline
coverage &0.784&0.904&0.974&0.978&0.990&0.974&0.652&0.878&0.922 \\
average length &0.117&0.055&0.044&0.052&0.011&0.010&0.078&0.041&0.034 \\
\end{tabular}
\vspace{0.5cm}
\caption{Coverage and average length of the parametrically constructed bootstrap confidence intervals based on \emph{parametric} ML estimation and a \emph{parametric} Poi-INAR bootstrap for $\alpha, G(0), \ldots, G(4)$ in case of a Poi(3)-INAR(1) DGP with $\alpha=0.5$ for different sample sizes.}
\label{tab:poi_p_3_05}
\end{table}

\begin{table}[h]
\centering
\begin{tabular}{c|ccc|ccc|ccc}
& \multicolumn{3}{|c|}{$\alpha$} & \multicolumn{3}{|c|}{$G(0)$} & \multicolumn{3}{|c}{$G(1)$} \\
$n$ & 100 & 500 & 1000 & 100 & 500 & 1000& 100 & 500 & 1000 \\
\hline
coverage &0.958&0.946&0.810&0.412&0.258&0.196&0.682&0.534&0.478 \\
average length &0.205&0.096&0.086&0.408&0.222&0.152&0.684&0.432&0.399 \\
& \multicolumn{3}{|c|}{$G(2)$} & \multicolumn{3}{|c|}{$G(3)$} & \multicolumn{3}{|c}{$G(4)$} \\
$n$ & 100 & 500 & 1000 & 100 & 500 & 1000& 100 & 500 & 1000 \\
\hline
coverage &0.808&0.864&0.864&0.728&0.880&0.864&0.488&0.786&0.844 \\
average length &0.559&0.348&0.306&0.401&0.230&0.204&0.235&0.071&0.065 \\
\end{tabular}
\vspace{0.5cm}
\caption{Coverage and average length of the semi-parametrically constructed bootstrap confidence intervals based on the semi-parametric estimation of \citet{drost} and the semi-parametric INAR bootstrap from Section \ref{sec:boot_algo} for $\alpha, G(0), \ldots, G(4)$ in case of a Poi(1)-INAR(1) DGP with $\alpha=0.9$ for different sample sizes.}
\label{tab:poi_sp_1_09}
\end{table}

\begin{table}[h]
\centering
\begin{tabular}{c|ccc|ccc|ccc}
& \multicolumn{3}{|c|}{$\alpha$} & \multicolumn{3}{|c|}{$G(0)$} & \multicolumn{3}{|c}{$G(1)$} \\
$n$ & 100 & 500 & 1000 & 100 & 500 & 1000& 100 & 500 & 1000 \\
\hline
coverage &0.940&0.956&0.952&0.936&0.964&0.952&0.982&0.970&0.980 \\
average length &0.068&0.029&0.022&0.237&0.107&0.078&0.040&0.008&0.012 \\
& \multicolumn{3}{|c|}{$G(2)$} & \multicolumn{3}{|c|}{$G(3)$} & \multicolumn{3}{|c}{$G(4)$} \\
$n$ & 100 & 500 & 1000 & 100 & 500 & 1000& 100 & 500 & 1000 \\
\hline
coverage &0.878&0.954&0.942&0.900&0.948&0.942&0.858&0.926&0.934 \\
average length &0.111&0.053&0.043&0.080&0.036&0.029&0.035&0.014&0.011 \\
\end{tabular}
\vspace{0.5cm}
\caption{Coverage and average length of the parametrically constructed bootstrap confidence intervals based on \emph{parametric} ML estimation and a \emph{parametric} Poi-INAR bootstrap for $\alpha, G(0), \ldots, G(4)$ in case of a Poi(1)-INAR(1) DGP with $\alpha=0.9$ for different sample sizes.}
\label{tab:poi_p_1_09}
\end{table}

\begin{table}[h]
\centering
\begin{tabular}{c|ccc|ccc|ccc}
& \multicolumn{3}{|c|}{$\alpha$} & \multicolumn{3}{|c|}{$G(0)$} & \multicolumn{3}{|c}{$G(1)$} \\
$n$ & 100 & 500 & 1000 & 100 & 500 & 1000& 100 & 500 & 1000 \\
\hline
coverage &0.938&0.920&0.875&0&0&0&0.012&0.008&0.010 \\
average length &0.229&0.082&0.055&0.006&0.001&0.001&0.005&0.001&0.001 \\
& \multicolumn{3}{|c|}{$G(2)$} & \multicolumn{3}{|c|}{$G(3)$} & \multicolumn{3}{|c}{$G(4)$} \\
$n$ & 100 & 500 & 1000 & 100 & 500 & 1000& 100 & 500 & 1000 \\
\hline
coverage &0.020&0.008&0.006&0.032&0.008&0.006&.032&0.006&0.010 \\
average length &0.005&0.001&0.001&0.005&0.001&0.001&0.005&0.001&0.001 \\
\end{tabular}
\vspace{0.5cm}
\caption{Coverage and average length of the semi-parametrically constructed bootstrap confidence intervals based on the semi-parametric estimation of \citet{drost} and the semi-parametric INAR bootstrap from Section \ref{sec:boot_algo} for $\alpha, G(0), \ldots, G(4)$ in case of a Poi(3)-INAR(1) DGP with $\alpha=0.9$ for different sample sizes.}
\label{tab:poi_sp_3_09}
\end{table}

\begin{table}[h]
\centering
\begin{tabular}{c|ccc|ccc|ccc}
& \multicolumn{3}{|c|}{$\alpha$} & \multicolumn{3}{|c|}{$G(0)$} & \multicolumn{3}{|c}{$G(1)$} \\
$n$ & 100 & 500 & 1000 & 100 & 500 & 1000& 100 & 500 & 1000 \\
\hline
coverage &0.948&0.950&0.938&0.850&0.924&0.910&0.884&0.940&0.920 \\
average length &0.061&0.026&0.020&0.094&0.040&0.028&0.166&0.078&0.054 \\
& \multicolumn{3}{|c|}{$G(2)$} & \multicolumn{3}{|c|}{$G(3)$} & \multicolumn{3}{|c}{$G(4)$} \\
$n$ & 100 & 500 & 1000 & 100 & 500 & 1000& 100 & 500 & 1000 \\
\hline
coverage &0.766&0.910&0.964&0.984&0.968&0.964&0.698&0.874&0.880 \\
average length &0.117&0.058&0.040&0.055&0.012&0.006&0.084&0.042&0.030 \\
\end{tabular}
\vspace{0.5cm}
\caption{Coverage and average length of the parametrically constructed bootstrap confidence intervals based on \emph{parametric} ML estimation and a \emph{parametric} Poi-INAR bootstrap for $\alpha, G(0), \ldots, G(4)$ in case of a Poi(3)-INAR(1) DGP with $\alpha=0.9$ for different sample sizes.}
\label{tab:poi_p_3_09}
\end{table}

\begin{table}[h]
\centering
\begin{tabular}{c|ccc|ccc|ccc}
& \multicolumn{3}{|c|}{$\alpha$} & \multicolumn{3}{|c|}{$G(0)$} & \multicolumn{3}{|c}{$G(1)$} \\
$n$ & 100 & 500 & 1000 & 100 & 500 & 1000& 100 & 500 & 1000 \\
\hline
coverage &0.940&0.960&0.950&0.866&0.950&0.952&0.918&0.952&0.958 \\
average length &0.303&0.117&0.081&0.360&0.151&0.105&0.335&0.143&0.100 \\
& \multicolumn{3}{|c|}{$G(2)$} & \multicolumn{3}{|c|}{$G(3)$} & \multicolumn{3}{|c}{$G(4)$} \\
$n$ & 100 & 500 & 1000 & 100 & 500 & 1000& 100 & 500 & 1000 \\
\hline
coverage &0.856&0.946&0.940&0..732&0.938&0.938&0.634&0.852&0.940 \\
average length &0.235&0.101&0.071&0.153&0.072&0.050&0.096&0.050&0.037 \\
\end{tabular}
\vspace{0.5cm}
\caption{Coverage and average length of the semi-parametrically constructed bootstrap confidence intervals based on the semi-parametric estimation of \citet{drost} and the semi-parametric INAR bootstrap from Section \ref{sec:boot_algo} for $\alpha, G(0), \ldots, G(4)$ in case of a NB(1,1/2)-INAR(1) DGP with $\alpha=0.5$ for different sample sizes.}
\label{tab:nb_sp_112_05}
\end{table}

\begin{table}[h]
\centering
\begin{tabular}{c|ccc|ccc|ccc}
& \multicolumn{3}{|c|}{$\alpha$} & \multicolumn{3}{|c|}{$G(0)$} & \multicolumn{3}{|c}{$G(1)$} \\
$n$ & 100 & 500 & 1000 & 100 & 500 & 1000& 100 & 500 & 1000 \\
\hline
coverage &0.926&0.446&0.104&0.150&0&0&0.018&0&0 \\
average length &0.333&0.142&0.100&0.219&0.099&0.070&0.053&0.017&0.011 \\
& \multicolumn{3}{|c|}{$G(2)$} & \multicolumn{3}{|c|}{$G(3)$} & \multicolumn{3}{|c}{$G(4)$} \\
$n$ & 100 & 500 & 1000 & 100 & 500 & 1000& 100 & 500 & 1000 \\
\hline
coverage &0.184&0&0&0.832&0.566&0.304&0.660&0.558&0.364 \\
average length &0.097&0.047&0.034&0.089&0.041&0.029&0.046&0.019&0.013 \\
\end{tabular}
\vspace{0.5cm}
\caption{Coverage and average length of the parametrically constructed bootstrap confidence intervals based on \emph{parametric} ML estimation and a \emph{parametric} Poi-INAR bootstrap for $\alpha, G(0), \ldots, G(4)$ in case of a NB(1,1/2)-INAR(1) DGP with $\alpha=0.5$ for different sample sizes.}
\label{tab:nb_p_112_05}
\end{table}

\begin{table}[h]
\centering
\begin{tabular}{c|ccc|ccc|ccc}
& \multicolumn{3}{|c|}{$\alpha$} & \multicolumn{3}{|c|}{$G(0)$} & \multicolumn{3}{|c}{$G(1)$} \\
$n$ & 100 & 500 & 1000 & 100 & 500 & 1000& 100 & 500 & 1000 \\
\hline
coverage &0.932&0.952&0.940&0.826&0.928&0.926&0.898&0.918&0.930 \\
average length &0.362&0.142&0.100&0.385&0.175&0.124&0.385&0.169&0.120 \\
& \multicolumn{3}{|c|}{$G(2)$} & \multicolumn{3}{|c|}{$G(3)$} & \multicolumn{3}{|c}{$G(4)$} \\
$n$ & 100 & 500 & 1000 & 100 & 500 & 1000& 100 & 500 & 1000 \\
\hline
coverage &0.832&0.938&0.944&0.710&0.912&0.932&0.486&0.756&0.874 \\
average length &0.314&0.131&0.092&0.177&0.081&0.057&0.079&0.039&0.030 \\
\end{tabular}
\vspace{0.5cm}
\caption{Coverage and average length of the semi-parametrically constructed bootstrap confidence intervals based on the semi-parametric estimation of \citet{drost} and the semi-parametric INAR bootstrap from Section \ref{sec:boot_algo} for $\alpha, G(0), \ldots, G(4)$ in case of a NB(10,10/11)-INAR(1) DGP with $\alpha=0.5$ for different sample sizes.}
\label{tab:nb_sp_101011_05}
\end{table}

\begin{table}[h]
\centering
\begin{tabular}{c|ccc|ccc|ccc}
& \multicolumn{3}{|c|}{$\alpha$} & \multicolumn{3}{|c|}{$G(0)$} & \multicolumn{3}{|c}{$G(1)$} \\
$n$ & 100 & 500 & 1000 & 100 & 500 & 1000& 100 & 500 & 1000 \\
\hline
coverage &0.948&0.952&0.922&0.866&0.782&0.662&0.206&0.004&0 \\
average length &0.304&0.131&0.091&0.229&0.103&0.073&0.040&0.008&0.004 \\
& \multicolumn{3}{|c|}{$G(2)$} & \multicolumn{3}{|c|}{$G(3)$} & \multicolumn{3}{|c}{$G(4)$} \\
$n$ & 100 & 500 & 1000 & 100 & 500 & 1000& 100 & 500 & 1000 \\
\hline
coverage &0.806&0.770&0.660&0.890&0.932&0.936&0.794&0.794&0.762 \\
average length &0.107&0.051&0.036&0.081&0.036&0.025&0.036&0.014&0.010 \\
\end{tabular}
\vspace{0.5cm}
\caption{Coverage and average length of the parametrically constructed bootstrap confidence intervals based on \emph{parametric} ML estimation and a \emph{parametric} Poi-INAR bootstrap for $\alpha, G(0), \ldots, G(4)$ in case of a NB(10,10/11)-INAR(1) DGP with $\alpha=0.5$ for different sample sizes.}
\label{tab:nb_p_101011_05}
\end{table}

\end{document}